%% file: main.tex
\documentclass[scheme=plain,11pt]{ctexart}
\usepackage[a4paper,left=2.3cm,right=2.3cm,top=2.5cm,bottom=2.5cm]{geometry}
\usepackage{amsmath,amssymb,amsthm,mathtools,mathrsfs}
\usepackage{booktabs,array,tabularx,enumitem,multirow,float}
\usepackage{graphicx}
\usepackage{flafter}
\usepackage{tikz,pgfplots}
\usetikzlibrary{arrows.meta}
\input{figure/two_edge_path_symbol.tex}
\pgfplotsset{compat=1.18}
\usepackage[hidelinks]{hyperref}
\usepackage[abbrev,nobysame]{amsrefs}
\newtheorem{theorem}{Theorem}[section]
\newtheorem{lemma}[theorem]{Lemma}
\newtheorem{proposition}[theorem]{Proposition}
\newtheorem{corollary}[theorem]{Corollary}
\theoremstyle{definition}
\newtheorem{definition}[theorem]{Definition}

\newtheorem{conjecture}[theorem]{Conjecture}
\newtheorem{example}[theorem]{Example}
\theoremstyle{remark}

\newenvironment{restated}[3][]{\par\medskip\noindent\textbf{#2~\ref{#3}}\if\relax\detokenize{#1}\relax\else\ \textup{(#1)}\fi.\ \itshape}{\par\medskip}
\newcommand{\dist}{d}
\newcommand{\dimhaus}{\dim_{\mathrm H}}
\newcommand{\dimmass}{\dim_{\mathrm M}}
\title{Iterated Graph Systems (II):\\
Bernoulli percolation and critical-exponent universality classes \\
on hierarchical lattices}
\author{Ziyu Neroli\thanks{Email: \href{mailto:z5222549@zmail.unsw.edu.au}{\texttt{z5222549@zmail.unsw.edu.au}}; \href{mailto:ziyu.li21@imperial.ac.uk}{\texttt{ziyu.li21@imperial.ac.uk}}.}\\[3pt]
\small Department of Mathematics, Imperial College London, United Kingdom}
\date{}
\begin{document}
\maketitle
\begin{abstract}
We study Bernoulli percolation and its universality classes on hierarchical lattices generated by edge iterated graph systems (EIGS). 
By examining eight critical exponents, we prove that no finite family of substitution-multiplicative dimensions completely determines universality, providing a rigorous proof of the failure of dimension-based universality discussed by Hauser and Saxena [Phys. Rev. B 34 (1986), 8193–8195]. 
Although the six exponentials theorem provides number-theoretic criteria for scale commensurability, we further disprove the conjecture that critical-exponent universality requires commensurate discrete renormalisation scales.
\end{abstract}
\tableofcontents
\clearpage

\section{Introduction}

The study of statistical-physics problems on graphs with a hierarchical structure rather than Euclidean geometry began around the 1970s~\cites{berker1979renormalisation,kaufman1981exactly,griffiths1982spin}.
Such a model is called a \emph{hierarchical lattice}: a graph with self-similar, fractal structure on which an exact renormalisation can usually be carried out.
The classical example is the diamond hierarchical lattice (DHL), shown in Figure~\ref{fig:dhl-construction}.

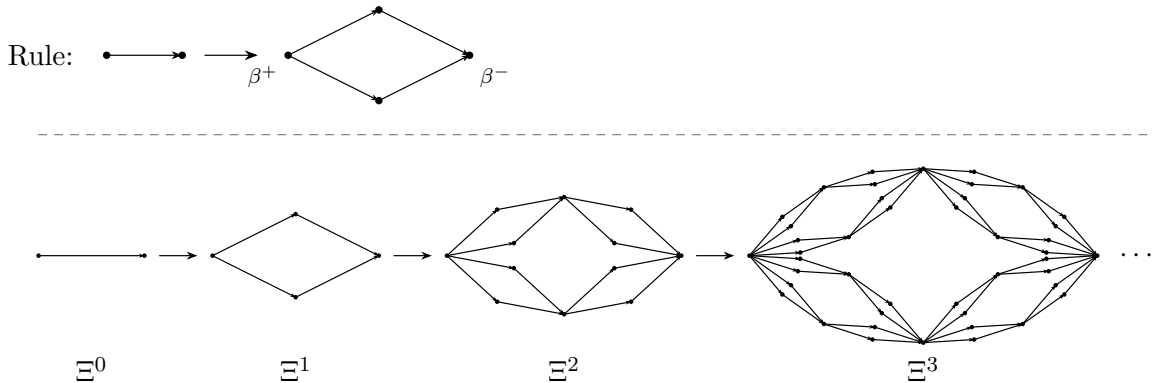
\begin{figure}[tbp]
\centering
\begin{tikzpicture}[line width=0.45pt,line cap=round]
\begin{scope}[xshift=-3.8cm]
\node[anchor=east] at (4.4,2.65) {Rule:};
\draw[-{Stealth[length=4pt]}] (4.7,2.65) -- (5.7,2.65);
\fill (4.7,2.65) circle (1.4pt);
\fill (5.7,2.65) circle (1.4pt);
\draw[-{Stealth[length=5pt]}] (6.0,2.65) -- (6.7,2.65);
\draw[-{Stealth[length=3pt]}] (7.1,2.65) -- (8.3,3.25);
\draw[-{Stealth[length=3pt]}] (8.3,3.25) -- (9.5,2.65);
\draw[-{Stealth[length=3pt]}] (7.1,2.65) -- (8.3,2.05);
\draw[-{Stealth[length=3pt]}] (8.3,2.05) -- (9.5,2.65);
\foreach \x/\y in {7.1/2.65,8.3/3.25,9.5/2.65,8.3/2.05}{\fill (\x,\y) circle (1.4pt);}
\node[below left,font=\scriptsize] at (7.1,2.65) {$\beta^+$};
\node[below right,font=\scriptsize] at (9.5,2.65) {$\beta^-$};
\end{scope}
\draw[dashed,gray] (0,1.6) -- (14.7,1.6);

\begin{scope}[xshift=0cm]
\coordinate (v0) at (0.000000,0.000000);
\coordinate (v1) at (1.400000,0.000000);
\draw[-{Stealth[length=2pt,width=1.4pt]}] (v0) -- (v1);
\fill (v0) circle (0.85pt);
\fill (v1) circle (0.85pt);
\node at (0.7,-1.5) {$\Xi^{0}$};
\end{scope}
\draw[-{Stealth[length=4pt]}] (1.5999999999999999,0) -- (2.0999999999999996,0);
\begin{scope}[xshift=2.3cm]
\coordinate (v0) at (0.000000,0.000000);
\coordinate (v1) at (2.200000,0.000000);
\coordinate (v2) at (1.100000,0.550000);
\coordinate (v3) at (1.100000,-0.550000);
\draw[-{Stealth[length=2pt,width=1.4pt]}] (v0) -- (v2);
\draw[-{Stealth[length=2pt,width=1.4pt]}] (v2) -- (v1);
\draw[-{Stealth[length=2pt,width=1.4pt]}] (v0) -- (v3);
\draw[-{Stealth[length=2pt,width=1.4pt]}] (v3) -- (v1);
\fill (v0) circle (0.85pt);
\fill (v1) circle (0.85pt);
\fill (v2) circle (0.85pt);
\fill (v3) circle (0.85pt);
\node at (1.1,-1.5) {$\Xi^{1}$};
\end{scope}
\draw[-{Stealth[length=4pt]}] (4.7,0) -- (5.2,0);
\begin{scope}[xshift=5.4cm]
\coordinate (v0) at (0.000000,0.000000);
\coordinate (v1) at (3.100000,0.000000);
\coordinate (v2) at (1.550000,0.775000);
\coordinate (v3) at (1.550000,-0.775000);
\coordinate (v4) at (0.664286,0.608929);
\coordinate (v5) at (0.885714,0.166071);
\coordinate (v6) at (2.435714,0.608929);
\coordinate (v7) at (2.214286,0.166071);
\coordinate (v8) at (0.885714,-0.166071);
\coordinate (v9) at (0.664286,-0.608929);
\coordinate (v10) at (2.214286,-0.166071);
\coordinate (v11) at (2.435714,-0.608929);
\draw[-{Stealth[length=2pt,width=1.4pt]}] (v0) -- (v4);
\draw[-{Stealth[length=2pt,width=1.4pt]}] (v4) -- (v2);
\draw[-{Stealth[length=2pt,width=1.4pt]}] (v0) -- (v5);
\draw[-{Stealth[length=2pt,width=1.4pt]}] (v5) -- (v2);
\draw[-{Stealth[length=2pt,width=1.4pt]}] (v2) -- (v6);
\draw[-{Stealth[length=2pt,width=1.4pt]}] (v6) -- (v1);
\draw[-{Stealth[length=2pt,width=1.4pt]}] (v2) -- (v7);
\draw[-{Stealth[length=2pt,width=1.4pt]}] (v7) -- (v1);
\draw[-{Stealth[length=2pt,width=1.4pt]}] (v0) -- (v8);
\draw[-{Stealth[length=2pt,width=1.4pt]}] (v8) -- (v3);
\draw[-{Stealth[length=2pt,width=1.4pt]}] (v0) -- (v9);
\draw[-{Stealth[length=2pt,width=1.4pt]}] (v9) -- (v3);
\draw[-{Stealth[length=2pt,width=1.4pt]}] (v3) -- (v10);
\draw[-{Stealth[length=2pt,width=1.4pt]}] (v10) -- (v1);
\draw[-{Stealth[length=2pt,width=1.4pt]}] (v3) -- (v11);
\draw[-{Stealth[length=2pt,width=1.4pt]}] (v11) -- (v1);
\fill (v0) circle (0.85pt);
\fill (v1) circle (0.85pt);
\fill (v2) circle (0.85pt);
\fill (v3) circle (0.85pt);
\fill (v4) circle (0.85pt);
\fill (v5) circle (0.85pt);
\fill (v6) circle (0.85pt);
\fill (v7) circle (0.85pt);
\fill (v8) circle (0.85pt);
\fill (v9) circle (0.85pt);
\fill (v10) circle (0.85pt);
\fill (v11) circle (0.85pt);
\node at (1.55,-1.5) {$\Xi^{2}$};
\end{scope}
\draw[-{Stealth[length=4pt]}] (8.7,0) -- (9.2,0);
\begin{scope}[xshift=9.4cm]
\coordinate (v0) at (0.000000,0.000000);
\coordinate (v1) at (4.600000,0.000000);
\coordinate (v2) at (2.300000,1.150000);
\coordinate (v3) at (2.300000,-1.150000);
\coordinate (v4) at (0.985714,0.903571);
\coordinate (v5) at (1.314286,0.246429);
\coordinate (v6) at (3.614286,0.903571);
\coordinate (v7) at (3.285714,0.246429);
\coordinate (v8) at (1.314286,-0.246429);
\coordinate (v9) at (0.985714,-0.903571);
\coordinate (v10) at (3.285714,-0.246429);
\coordinate (v11) at (3.614286,-0.903571);
\coordinate (v12) at (0.436384,0.513393);
\coordinate (v13) at (0.549330,0.390178);
\coordinate (v14) at (1.627455,1.108928);
\coordinate (v15) at (1.658259,0.944643);
\coordinate (v16) at (0.641741,0.205357);
\coordinate (v17) at (0.672545,0.041072);
\coordinate (v18) at (1.750670,0.759822);
\coordinate (v19) at (1.863616,0.636607);
\coordinate (v20) at (2.972545,1.108928);
\coordinate (v21) at (2.941741,0.944643);
\coordinate (v22) at (4.163616,0.513393);
\coordinate (v23) at (4.050670,0.390178);
\coordinate (v24) at (2.849330,0.759822);
\coordinate (v25) at (2.736384,0.636607);
\coordinate (v26) at (3.958259,0.205357);
\coordinate (v27) at (3.927455,0.041072);
\coordinate (v28) at (0.672545,-0.041072);
\coordinate (v29) at (0.641741,-0.205357);
\coordinate (v30) at (1.863616,-0.636607);
\coordinate (v31) at (1.750670,-0.759822);
\coordinate (v32) at (0.549330,-0.390178);
\coordinate (v33) at (0.436384,-0.513393);
\coordinate (v34) at (1.658259,-0.944643);
\coordinate (v35) at (1.627455,-1.108928);
\coordinate (v36) at (2.736384,-0.636607);
\coordinate (v37) at (2.849330,-0.759822);
\coordinate (v38) at (3.927455,-0.041072);
\coordinate (v39) at (3.958259,-0.205357);
\coordinate (v40) at (2.941741,-0.944643);
\coordinate (v41) at (2.972545,-1.108928);
\coordinate (v42) at (4.050670,-0.390178);
\coordinate (v43) at (4.163616,-0.513393);
\draw[-{Stealth[length=2pt,width=1.4pt]}] (v0) -- (v12);
\draw[-{Stealth[length=2pt,width=1.4pt]}] (v12) -- (v4);
\draw[-{Stealth[length=2pt,width=1.4pt]}] (v0) -- (v13);
\draw[-{Stealth[length=2pt,width=1.4pt]}] (v13) -- (v4);
\draw[-{Stealth[length=2pt,width=1.4pt]}] (v4) -- (v14);
\draw[-{Stealth[length=2pt,width=1.4pt]}] (v14) -- (v2);
\draw[-{Stealth[length=2pt,width=1.4pt]}] (v4) -- (v15);
\draw[-{Stealth[length=2pt,width=1.4pt]}] (v15) -- (v2);
\draw[-{Stealth[length=2pt,width=1.4pt]}] (v0) -- (v16);
\draw[-{Stealth[length=2pt,width=1.4pt]}] (v16) -- (v5);
\draw[-{Stealth[length=2pt,width=1.4pt]}] (v0) -- (v17);
\draw[-{Stealth[length=2pt,width=1.4pt]}] (v17) -- (v5);
\draw[-{Stealth[length=2pt,width=1.4pt]}] (v5) -- (v18);
\draw[-{Stealth[length=2pt,width=1.4pt]}] (v18) -- (v2);
\draw[-{Stealth[length=2pt,width=1.4pt]}] (v5) -- (v19);
\draw[-{Stealth[length=2pt,width=1.4pt]}] (v19) -- (v2);
\draw[-{Stealth[length=2pt,width=1.4pt]}] (v2) -- (v20);
\draw[-{Stealth[length=2pt,width=1.4pt]}] (v20) -- (v6);
\draw[-{Stealth[length=2pt,width=1.4pt]}] (v2) -- (v21);
\draw[-{Stealth[length=2pt,width=1.4pt]}] (v21) -- (v6);
\draw[-{Stealth[length=2pt,width=1.4pt]}] (v6) -- (v22);
\draw[-{Stealth[length=2pt,width=1.4pt]}] (v22) -- (v1);
\draw[-{Stealth[length=2pt,width=1.4pt]}] (v6) -- (v23);
\draw[-{Stealth[length=2pt,width=1.4pt]}] (v23) -- (v1);
\draw[-{Stealth[length=2pt,width=1.4pt]}] (v2) -- (v24);
\draw[-{Stealth[length=2pt,width=1.4pt]}] (v24) -- (v7);
\draw[-{Stealth[length=2pt,width=1.4pt]}] (v2) -- (v25);
\draw[-{Stealth[length=2pt,width=1.4pt]}] (v25) -- (v7);
\draw[-{Stealth[length=2pt,width=1.4pt]}] (v7) -- (v26);
\draw[-{Stealth[length=2pt,width=1.4pt]}] (v26) -- (v1);
\draw[-{Stealth[length=2pt,width=1.4pt]}] (v7) -- (v27);
\draw[-{Stealth[length=2pt,width=1.4pt]}] (v27) -- (v1);
\draw[-{Stealth[length=2pt,width=1.4pt]}] (v0) -- (v28);
\draw[-{Stealth[length=2pt,width=1.4pt]}] (v28) -- (v8);
\draw[-{Stealth[length=2pt,width=1.4pt]}] (v0) -- (v29);
\draw[-{Stealth[length=2pt,width=1.4pt]}] (v29) -- (v8);
\draw[-{Stealth[length=2pt,width=1.4pt]}] (v8) -- (v30);
\draw[-{Stealth[length=2pt,width=1.4pt]}] (v30) -- (v3);
\draw[-{Stealth[length=2pt,width=1.4pt]}] (v8) -- (v31);
\draw[-{Stealth[length=2pt,width=1.4pt]}] (v31) -- (v3);
\draw[-{Stealth[length=2pt,width=1.4pt]}] (v0) -- (v32);
\draw[-{Stealth[length=2pt,width=1.4pt]}] (v32) -- (v9);
\draw[-{Stealth[length=2pt,width=1.4pt]}] (v0) -- (v33);
\draw[-{Stealth[length=2pt,width=1.4pt]}] (v33) -- (v9);
\draw[-{Stealth[length=2pt,width=1.4pt]}] (v9) -- (v34);
\draw[-{Stealth[length=2pt,width=1.4pt]}] (v34) -- (v3);
\draw[-{Stealth[length=2pt,width=1.4pt]}] (v9) -- (v35);
\draw[-{Stealth[length=2pt,width=1.4pt]}] (v35) -- (v3);
\draw[-{Stealth[length=2pt,width=1.4pt]}] (v3) -- (v36);
\draw[-{Stealth[length=2pt,width=1.4pt]}] (v36) -- (v10);
\draw[-{Stealth[length=2pt,width=1.4pt]}] (v3) -- (v37);
\draw[-{Stealth[length=2pt,width=1.4pt]}] (v37) -- (v10);
\draw[-{Stealth[length=2pt,width=1.4pt]}] (v10) -- (v38);
\draw[-{Stealth[length=2pt,width=1.4pt]}] (v38) -- (v1);
\draw[-{Stealth[length=2pt,width=1.4pt]}] (v10) -- (v39);
\draw[-{Stealth[length=2pt,width=1.4pt]}] (v39) -- (v1);
\draw[-{Stealth[length=2pt,width=1.4pt]}] (v3) -- (v40);
\draw[-{Stealth[length=2pt,width=1.4pt]}] (v40) -- (v11);
\draw[-{Stealth[length=2pt,width=1.4pt]}] (v3) -- (v41);
\draw[-{Stealth[length=2pt,width=1.4pt]}] (v41) -- (v11);
\draw[-{Stealth[length=2pt,width=1.4pt]}] (v11) -- (v42);
\draw[-{Stealth[length=2pt,width=1.4pt]}] (v42) -- (v1);
\draw[-{Stealth[length=2pt,width=1.4pt]}] (v11) -- (v43);
\draw[-{Stealth[length=2pt,width=1.4pt]}] (v43) -- (v1);
\fill (v0) circle (0.85pt);
\fill (v1) circle (0.85pt);
\fill (v2) circle (0.85pt);
\fill (v3) circle (0.85pt);
\fill (v4) circle (0.85pt);
\fill (v5) circle (0.85pt);
\fill (v6) circle (0.85pt);
\fill (v7) circle (0.85pt);
\fill (v8) circle (0.85pt);
\fill (v9) circle (0.85pt);
\fill (v10) circle (0.85pt);
\fill (v11) circle (0.85pt);
\fill (v12) circle (0.85pt);
\fill (v13) circle (0.85pt);
\fill (v14) circle (0.85pt);
\fill (v15) circle (0.85pt);
\fill (v16) circle (0.85pt);
\fill (v17) circle (0.85pt);
\fill (v18) circle (0.85pt);
\fill (v19) circle (0.85pt);
\fill (v20) circle (0.85pt);
\fill (v21) circle (0.85pt);
\fill (v22) circle (0.85pt);
\fill (v23) circle (0.85pt);
\fill (v24) circle (0.85pt);
\fill (v25) circle (0.85pt);
\fill (v26) circle (0.85pt);
\fill (v27) circle (0.85pt);
\fill (v28) circle (0.85pt);
\fill (v29) circle (0.85pt);
\fill (v30) circle (0.85pt);
\fill (v31) circle (0.85pt);
\fill (v32) circle (0.85pt);
\fill (v33) circle (0.85pt);
\fill (v34) circle (0.85pt);
\fill (v35) circle (0.85pt);
\fill (v36) circle (0.85pt);
\fill (v37) circle (0.85pt);
\fill (v38) circle (0.85pt);
\fill (v39) circle (0.85pt);
\fill (v40) circle (0.85pt);
\fill (v41) circle (0.85pt);
\fill (v42) circle (0.85pt);
\fill (v43) circle (0.85pt);
\node at (2.3,-1.5) {$\Xi^{3}$};
\end{scope}
\node at (14.55,0) {$\cdots$};
\end{tikzpicture}
\caption{The diamond hierarchical lattice (DHL)}
\label{fig:dhl-construction}
\end{figure}

In 1986, Hauser and Saxena~\cite{hauser1986geometrical} computed critical exponents on such lattices.
More specifically, they studied percolation on different hierarchical lattices using exact position-space renormalisa\-tion-group schemes and found that the percolation exponent $\nu_p$ depended strongly on lattice topology even when the intrinsic dimensions and connectivities were the same.
They therefore questioned whether dimension-based universality holds on hierarchical lattices.
Halberstam and Hutchcroft~\cite{halberstam2022what} also investigated the limits of universality on self-similar fractals.
They constructed two fractals sharing the Hausdorff, spectral, topological and topological Hausdorff dimensions, but obtained numerical evidence that their percolation Fisher exponents $\tau$ differ.
Their findings suggest that even agreement of several standard fractal dimensions need not determine the percolation universality class.
See also~\cites{gefen1980critical,havlin1983percolation,havlin1984percolation,hu1985problem} for related studies of critical behaviour and universality on hierarchical and fractal lattices.

In 2010, Hambly and Kumagai~\cite{hambly2010diffusion} gave a rigorous treatment of critical percolation on the DHL.
They proved that the scaling limit of the critical percolation cluster is a graph-directed random recursive fractal; in the language of this paper, its recursive construction is a random EIGS, which we introduce later.
This observation is decisive for our approach, since it provides a recursive framework for analysing critical percolation.

Hambly and Kumagai's work enables us to gain insight into universality classes by studying this random EIGS.
We examine eight critical exponents and use four that exist for every classical EIGS hierarchical lattice, with the observation conventions of Section~\ref{sec:static-class}, to define a critical-exponent universality class. We then prove that equality of these four exponents is equivalent to equality of $\dimhaus(\Xi)$, $\dimmass(\mathfrak C)$ and $\dim_{\mathrm P}(\mathfrak C)$:
\begin{itemize}
    \item the Hausdorff dimension $\dimhaus(\Xi)$ of the hierarchical lattice itself, which we call the ambient dimension;
    \item the pivotal dimension $\dim_{\mathrm P}(\mathfrak C)$ (sometimes called the thermal dimension in the physics literature), which describes how the expected number of pivotal edges grows with scale and depends on the derivative of the reliability function at the critical probability;
    \item and the mass dimension $\dimmass(\mathfrak C)$ of the critical percolation cluster, which depends on the spectral radius of its mass matrix.
\end{itemize}

Two hierarchical lattices can readily have the same ambient dimension $\dimhaus(\Xi)$, since it depends only on the number of edges and the distance between the terminals of the rule graph, the latter giving the substitution scale.
Equality of their pivotal and critical mass dimensions is more restrictive.
At a common scale, the former requires their reliability functions to have the same derivative at their respective critical probabilities, even when the functions differ;
the latter requires their different $3\times3$ mass matrices to have the same spectral radius.

These requirements show why numerical agreement alone is insufficient to establish that two hierarchical lattices belong to the same universality class.
The following pair of examples gives an exact instance.

\subsection{A pair of key examples}

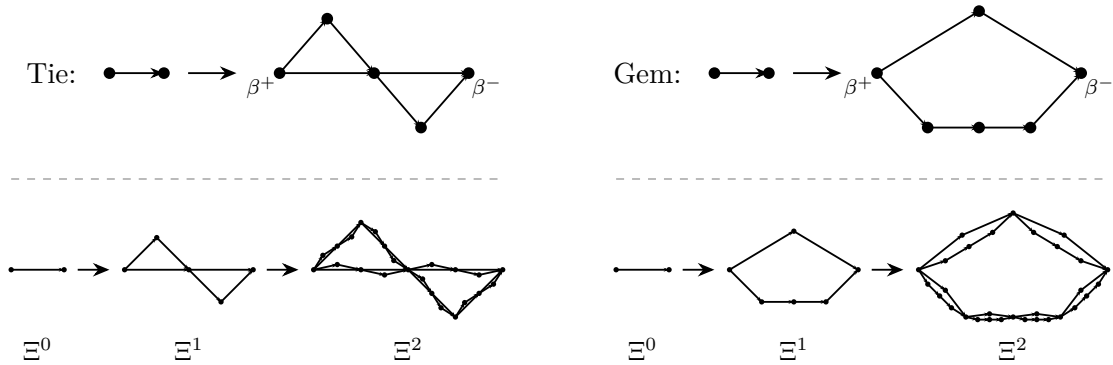
\begin{figure}[tbp]
\centering
\begin{tikzpicture}[
    x=1cm,y=1cm,
    vertex/.style={circle,fill=black,inner sep=1.55pt},
    line/.style={line width=0.7pt,line cap=round},
    repl/.style={-{Stealth[length=2.2mm,width=1.8mm]},line width=0.7pt},
    ilabel/.style={font=\scriptsize,inner sep=1pt},
    rulelabel/.style={font=\normalsize}
]

\begin{scope}[yshift=2.6cm]
\begin{scope}[xshift=1.1cm,yshift=0cm]
  \node[rulelabel,anchor=east] at (-0.10,0) {Tie:};
  \draw[line,-{Stealth[length=1.8mm,width=1.4mm]}] (0.20,0) -- (0.92,0);
  \node[vertex] at (0.20,0) {};
  \node[vertex] at (0.92,0) {};
  \draw[repl] (1.24,0) -- (1.88,0);

  \coordinate (bp) at (2.45,0);
  \coordinate (mid) at (3.70,0);
  \coordinate (bm) at (4.95,0);
  \coordinate (up) at (3.08,0.72);
  \coordinate (down) at (4.32,-0.72);

  \draw[line,-{Stealth[length=3pt]}] (bp)--(up);\draw[line,-{Stealth[length=3pt]}] (up)--(mid);\draw[line,-{Stealth[length=3pt]}] (bp)--(mid);
  \draw[line,-{Stealth[length=3pt]}] (mid)--(down);\draw[line,-{Stealth[length=3pt]}] (down)--(bm);\draw[line,-{Stealth[length=3pt]}] (mid)--(bm);

  \foreach \p in {bp,mid,bm,up,down}
    \node[vertex] at (\p) {};

  \node[ilabel,anchor=north east] at (bp) {$\beta^+$};
  \node[ilabel,anchor=north west] at (bm) {$\beta^-$};
\end{scope}

\begin{scope}[xshift=9.1cm,yshift=0cm]
  \node[rulelabel,anchor=east] at (-0.10,0) {Gem:};
  \draw[line,-{Stealth[length=1.8mm,width=1.4mm]}] (0.20,0) -- (0.92,0);
  \node[vertex] at (0.20,0) {};
  \node[vertex] at (0.92,0) {};
  \draw[repl] (1.24,0) -- (1.88,0);

  \coordinate (bp) at (2.35,0);
  \coordinate (bm) at (5.05,0);
  \coordinate (top) at (3.70,0.82);
  \coordinate (b1) at (3.02,-0.72);
  \coordinate (b2) at (3.70,-0.72);
  \coordinate (b3) at (4.38,-0.72);

  \draw[line,-{Stealth[length=3pt]}] (bp)--(top);\draw[line,-{Stealth[length=3pt]}] (top)--(bm);
  \draw[line,-{Stealth[length=3pt]}] (bp)--(b1);\draw[line,-{Stealth[length=3pt]}] (b1)--(b2);\draw[line,-{Stealth[length=3pt]}] (b2)--(b3);\draw[line,-{Stealth[length=3pt]}] (b3)--(bm);

  \foreach \p in {bp,bm,top,b1,b2,b3}
    \node[vertex] at (\p) {};

  \node[ilabel,anchor=north east] at (bp) {$\beta^+$};
  \node[ilabel,anchor=north west] at (bm) {$\beta^-$};
\end{scope}
\end{scope}
\draw[dashed,gray] (0,1.2) -- (6.5,1.2);
\draw[dashed,gray] (8,1.2) -- (14.5,1.2);
\begin{scope}[xshift=0cm]
\coordinate (q0) at (0.000000,0.000000);
\coordinate (q1) at (0.700000,0.000000);
\draw[line,-{Stealth[length=2pt,width=1.4pt]}] (q0) -- (q1);
\fill (q0) circle (1pt);
\fill (q1) circle (1pt);
\node[font=\small] at (0.35,-1.05) {$\Xi^{0}$};
\end{scope}
\draw[repl] (0.8799999999999999,0) -- (1.2999999999999998,0);
\begin{scope}[xshift=1.5cm]
\coordinate (q0) at (0.000000,0.000000);
\coordinate (q1) at (1.700000,0.000000);
\coordinate (q2) at (0.425000,0.425000);
\coordinate (q3) at (0.850000,0.000000);
\coordinate (q4) at (1.275000,-0.425000);
\draw[line,-{Stealth[length=2pt,width=1.4pt]}] (q0) -- (q2);
\draw[line,-{Stealth[length=2pt,width=1.4pt]}] (q2) -- (q3);
\draw[line,-{Stealth[length=2pt,width=1.4pt]}] (q0) -- (q3);
\draw[line,-{Stealth[length=2pt,width=1.4pt]}] (q3) -- (q4);
\draw[line,-{Stealth[length=2pt,width=1.4pt]}] (q4) -- (q1);
\draw[line,-{Stealth[length=2pt,width=1.4pt]}] (q3) -- (q1);
\fill (q0) circle (1pt);
\fill (q1) circle (1pt);
\fill (q2) circle (1pt);
\fill (q3) circle (1pt);
\fill (q4) circle (1pt);
\node[font=\small] at (0.85,-1.05) {$\Xi^{1}$};
\end{scope}
\draw[repl] (3.3800000000000003,0) -- (3.8000000000000003,0);
\begin{scope}[xshift=4.0cm]
\coordinate (q0) at (0.000000,0.000000);
\coordinate (q1) at (2.500000,0.000000);
\coordinate (q2) at (0.625000,0.625000);
\coordinate (q3) at (1.250000,0.000000);
\coordinate (q4) at (1.875000,-0.625000);
\coordinate (q5) at (0.121875,0.190625);
\coordinate (q6) at (0.312500,0.312500);
\coordinate (q7) at (0.503125,0.434375);
\coordinate (q8) at (0.815625,0.503125);
\coordinate (q9) at (0.937500,0.312500);
\coordinate (q10) at (1.059375,0.121875);
\coordinate (q11) at (0.312500,0.068750);
\coordinate (q12) at (0.625000,0.000000);
\coordinate (q13) at (0.937500,-0.068750);
\coordinate (q14) at (1.440625,-0.121875);
\coordinate (q15) at (1.562500,-0.312500);
\coordinate (q16) at (1.684375,-0.503125);
\coordinate (q17) at (1.996875,-0.434375);
\coordinate (q18) at (2.187500,-0.312500);
\coordinate (q19) at (2.378125,-0.190625);
\coordinate (q20) at (1.562500,0.068750);
\coordinate (q21) at (1.875000,0.000000);
\coordinate (q22) at (2.187500,-0.068750);
\draw[line,-{Stealth[length=2pt,width=1.4pt]}] (q0) -- (q5);
\draw[line,-{Stealth[length=2pt,width=1.4pt]}] (q5) -- (q6);
\draw[line,-{Stealth[length=2pt,width=1.4pt]}] (q0) -- (q6);
\draw[line,-{Stealth[length=2pt,width=1.4pt]}] (q6) -- (q7);
\draw[line,-{Stealth[length=2pt,width=1.4pt]}] (q7) -- (q2);
\draw[line,-{Stealth[length=2pt,width=1.4pt]}] (q6) -- (q2);
\draw[line,-{Stealth[length=2pt,width=1.4pt]}] (q2) -- (q8);
\draw[line,-{Stealth[length=2pt,width=1.4pt]}] (q8) -- (q9);
\draw[line,-{Stealth[length=2pt,width=1.4pt]}] (q2) -- (q9);
\draw[line,-{Stealth[length=2pt,width=1.4pt]}] (q9) -- (q10);
\draw[line,-{Stealth[length=2pt,width=1.4pt]}] (q10) -- (q3);
\draw[line,-{Stealth[length=2pt,width=1.4pt]}] (q9) -- (q3);
\draw[line,-{Stealth[length=2pt,width=1.4pt]}] (q0) -- (q11);
\draw[line,-{Stealth[length=2pt,width=1.4pt]}] (q11) -- (q12);
\draw[line,-{Stealth[length=2pt,width=1.4pt]}] (q0) -- (q12);
\draw[line,-{Stealth[length=2pt,width=1.4pt]}] (q12) -- (q13);
\draw[line,-{Stealth[length=2pt,width=1.4pt]}] (q13) -- (q3);
\draw[line,-{Stealth[length=2pt,width=1.4pt]}] (q12) -- (q3);
\draw[line,-{Stealth[length=2pt,width=1.4pt]}] (q3) -- (q14);
\draw[line,-{Stealth[length=2pt,width=1.4pt]}] (q14) -- (q15);
\draw[line,-{Stealth[length=2pt,width=1.4pt]}] (q3) -- (q15);
\draw[line,-{Stealth[length=2pt,width=1.4pt]}] (q15) -- (q16);
\draw[line,-{Stealth[length=2pt,width=1.4pt]}] (q16) -- (q4);
\draw[line,-{Stealth[length=2pt,width=1.4pt]}] (q15) -- (q4);
\draw[line,-{Stealth[length=2pt,width=1.4pt]}] (q4) -- (q17);
\draw[line,-{Stealth[length=2pt,width=1.4pt]}] (q17) -- (q18);
\draw[line,-{Stealth[length=2pt,width=1.4pt]}] (q4) -- (q18);
\draw[line,-{Stealth[length=2pt,width=1.4pt]}] (q18) -- (q19);
\draw[line,-{Stealth[length=2pt,width=1.4pt]}] (q19) -- (q1);
\draw[line,-{Stealth[length=2pt,width=1.4pt]}] (q18) -- (q1);
\draw[line,-{Stealth[length=2pt,width=1.4pt]}] (q3) -- (q20);
\draw[line,-{Stealth[length=2pt,width=1.4pt]}] (q20) -- (q21);
\draw[line,-{Stealth[length=2pt,width=1.4pt]}] (q3) -- (q21);
\draw[line,-{Stealth[length=2pt,width=1.4pt]}] (q21) -- (q22);
\draw[line,-{Stealth[length=2pt,width=1.4pt]}] (q22) -- (q1);
\draw[line,-{Stealth[length=2pt,width=1.4pt]}] (q21) -- (q1);
\fill (q0) circle (1pt);
\fill (q1) circle (1pt);
\fill (q2) circle (1pt);
\fill (q3) circle (1pt);
\fill (q4) circle (1pt);
\fill (q5) circle (1pt);
\fill (q6) circle (1pt);
\fill (q7) circle (1pt);
\fill (q8) circle (1pt);
\fill (q9) circle (1pt);
\fill (q10) circle (1pt);
\fill (q11) circle (1pt);
\fill (q12) circle (1pt);
\fill (q13) circle (1pt);
\fill (q14) circle (1pt);
\fill (q15) circle (1pt);
\fill (q16) circle (1pt);
\fill (q17) circle (1pt);
\fill (q18) circle (1pt);
\fill (q19) circle (1pt);
\fill (q20) circle (1pt);
\fill (q21) circle (1pt);
\fill (q22) circle (1pt);
\node[font=\small] at (1.25,-1.05) {$\Xi^{2}$};
\end{scope}
\begin{scope}[xshift=8cm]
\coordinate (q0) at (0.000000,0.000000);
\coordinate (q1) at (0.700000,0.000000);
\draw[line,-{Stealth[length=2pt,width=1.4pt]}] (q0) -- (q1);
\fill (q0) circle (1pt);
\fill (q1) circle (1pt);
\node[font=\small] at (0.35,-1.05) {$\Xi^{0}$};
\end{scope}
\draw[repl] (8.88,0) -- (9.30,0);
\begin{scope}[xshift=9.5cm]
\coordinate (q0) at (0.000000,0.000000);
\coordinate (q1) at (1.700000,0.000000);
\coordinate (q2) at (0.850000,0.510000);
\coordinate (q3) at (0.425000,-0.425000);
\coordinate (q4) at (0.850000,-0.425000);
\coordinate (q5) at (1.275000,-0.425000);
\draw[line,-{Stealth[length=2pt,width=1.4pt]}] (q0) -- (q2);
\draw[line,-{Stealth[length=2pt,width=1.4pt]}] (q2) -- (q1);
\draw[line,-{Stealth[length=2pt,width=1.4pt]}] (q0) -- (q3);
\draw[line,-{Stealth[length=2pt,width=1.4pt]}] (q3) -- (q4);
\draw[line,-{Stealth[length=2pt,width=1.4pt]}] (q4) -- (q5);
\draw[line,-{Stealth[length=2pt,width=1.4pt]}] (q5) -- (q1);
\fill (q0) circle (1pt);
\fill (q1) circle (1pt);
\fill (q2) circle (1pt);
\fill (q3) circle (1pt);
\fill (q4) circle (1pt);
\fill (q5) circle (1pt);
\node[font=\small] at (0.85,-1.05) {$\Xi^{1}$};
\end{scope}
\draw[repl] (11.38,0) -- (11.80,0);
\begin{scope}[xshift=12.0cm]
\coordinate (q0) at (0.000000,0.000000);
\coordinate (q1) at (2.500000,0.000000);
\coordinate (q2) at (1.250000,0.750000);
\coordinate (q3) at (0.625000,-0.625000);
\coordinate (q4) at (1.250000,-0.625000);
\coordinate (q5) at (1.875000,-0.625000);
\coordinate (q6) at (0.575500,0.457500);
\coordinate (q7) at (0.353750,0.118750);
\coordinate (q8) at (0.666250,0.306250);
\coordinate (q9) at (0.978750,0.493750);
\coordinate (q10) at (1.924500,0.457500);
\coordinate (q11) at (1.521250,0.493750);
\coordinate (q12) at (1.833750,0.306250);
\coordinate (q13) at (2.146250,0.118750);
\coordinate (q14) at (0.353750,-0.271250);
\coordinate (q15) at (0.121875,-0.190625);
\coordinate (q16) at (0.278125,-0.346875);
\coordinate (q17) at (0.434375,-0.503125);
\coordinate (q18) at (0.937500,-0.583750);
\coordinate (q19) at (0.781250,-0.659375);
\coordinate (q20) at (0.937500,-0.659375);
\coordinate (q21) at (1.093750,-0.659375);
\coordinate (q22) at (1.562500,-0.583750);
\coordinate (q23) at (1.406250,-0.659375);
\coordinate (q24) at (1.562500,-0.659375);
\coordinate (q25) at (1.718750,-0.659375);
\coordinate (q26) at (2.146250,-0.271250);
\coordinate (q27) at (2.065625,-0.503125);
\coordinate (q28) at (2.221875,-0.346875);
\coordinate (q29) at (2.378125,-0.190625);
\draw[line,-{Stealth[length=2pt,width=1.4pt]}] (q0) -- (q6);
\draw[line,-{Stealth[length=2pt,width=1.4pt]}] (q6) -- (q2);
\draw[line,-{Stealth[length=2pt,width=1.4pt]}] (q0) -- (q7);
\draw[line,-{Stealth[length=2pt,width=1.4pt]}] (q7) -- (q8);
\draw[line,-{Stealth[length=2pt,width=1.4pt]}] (q8) -- (q9);
\draw[line,-{Stealth[length=2pt,width=1.4pt]}] (q9) -- (q2);
\draw[line,-{Stealth[length=2pt,width=1.4pt]}] (q2) -- (q10);
\draw[line,-{Stealth[length=2pt,width=1.4pt]}] (q10) -- (q1);
\draw[line,-{Stealth[length=2pt,width=1.4pt]}] (q2) -- (q11);
\draw[line,-{Stealth[length=2pt,width=1.4pt]}] (q11) -- (q12);
\draw[line,-{Stealth[length=2pt,width=1.4pt]}] (q12) -- (q13);
\draw[line,-{Stealth[length=2pt,width=1.4pt]}] (q13) -- (q1);
\draw[line,-{Stealth[length=2pt,width=1.4pt]}] (q0) -- (q14);
\draw[line,-{Stealth[length=2pt,width=1.4pt]}] (q14) -- (q3);
\draw[line,-{Stealth[length=2pt,width=1.4pt]}] (q0) -- (q15);
\draw[line,-{Stealth[length=2pt,width=1.4pt]}] (q15) -- (q16);
\draw[line,-{Stealth[length=2pt,width=1.4pt]}] (q16) -- (q17);
\draw[line,-{Stealth[length=2pt,width=1.4pt]}] (q17) -- (q3);
\draw[line,-{Stealth[length=2pt,width=1.4pt]}] (q3) -- (q18);
\draw[line,-{Stealth[length=2pt,width=1.4pt]}] (q18) -- (q4);
\draw[line,-{Stealth[length=2pt,width=1.4pt]}] (q3) -- (q19);
\draw[line,-{Stealth[length=2pt,width=1.4pt]}] (q19) -- (q20);
\draw[line,-{Stealth[length=2pt,width=1.4pt]}] (q20) -- (q21);
\draw[line,-{Stealth[length=2pt,width=1.4pt]}] (q21) -- (q4);
\draw[line,-{Stealth[length=2pt,width=1.4pt]}] (q4) -- (q22);
\draw[line,-{Stealth[length=2pt,width=1.4pt]}] (q22) -- (q5);
\draw[line,-{Stealth[length=2pt,width=1.4pt]}] (q4) -- (q23);
\draw[line,-{Stealth[length=2pt,width=1.4pt]}] (q23) -- (q24);
\draw[line,-{Stealth[length=2pt,width=1.4pt]}] (q24) -- (q25);
\draw[line,-{Stealth[length=2pt,width=1.4pt]}] (q25) -- (q5);
\draw[line,-{Stealth[length=2pt,width=1.4pt]}] (q5) -- (q26);
\draw[line,-{Stealth[length=2pt,width=1.4pt]}] (q26) -- (q1);
\draw[line,-{Stealth[length=2pt,width=1.4pt]}] (q5) -- (q27);
\draw[line,-{Stealth[length=2pt,width=1.4pt]}] (q27) -- (q28);
\draw[line,-{Stealth[length=2pt,width=1.4pt]}] (q28) -- (q29);
\draw[line,-{Stealth[length=2pt,width=1.4pt]}] (q29) -- (q1);
\fill (q0) circle (1pt);
\fill (q1) circle (1pt);
\fill (q2) circle (1pt);
\fill (q3) circle (1pt);
\fill (q4) circle (1pt);
\fill (q5) circle (1pt);
\fill (q6) circle (1pt);
\fill (q7) circle (1pt);
\fill (q8) circle (1pt);
\fill (q9) circle (1pt);
\fill (q10) circle (1pt);
\fill (q11) circle (1pt);
\fill (q12) circle (1pt);
\fill (q13) circle (1pt);
\fill (q14) circle (1pt);
\fill (q15) circle (1pt);
\fill (q16) circle (1pt);
\fill (q17) circle (1pt);
\fill (q18) circle (1pt);
\fill (q19) circle (1pt);
\fill (q20) circle (1pt);
\fill (q21) circle (1pt);
\fill (q22) circle (1pt);
\fill (q23) circle (1pt);
\fill (q24) circle (1pt);
\fill (q25) circle (1pt);
\fill (q26) circle (1pt);
\fill (q27) circle (1pt);
\fill (q28) circle (1pt);
\fill (q29) circle (1pt);
\node[font=\small] at (1.25,-1.05) {$\Xi^{2}$};
\end{scope}
\end{tikzpicture}
\caption{The tie and gem rule graphs.}
\label{fig:terminal-symmetric-rules}
\end{figure}

Figure~\ref{fig:terminal-symmetric-rules} shows two different hierarchical lattices, called Tie and Gem.
Both have Hausdorff dimension $\dimhaus(\Xi)=\log6/\log2$.
Using our results in Sections~\ref{sec:bernoulli} and~\ref{sec:universality}, we can further analyse percolation on them.

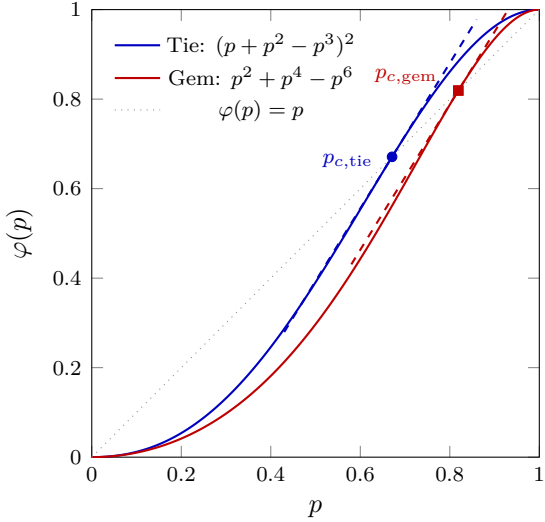
\begin{figure}[tbp]
\centering
\begin{minipage}[c]{0.48\textwidth}
\centering
\begin{tikzpicture}
\begin{axis}[
  width=7.5cm,height=7.5cm,axis equal image,
  xmin=0,xmax=1,ymin=0,ymax=1,
  xlabel={$p$},ylabel={$\varphi(p)$},
  xtick={0,0.2,0.4,0.6,0.8,1},ytick={0,0.2,0.4,0.6,0.8,1},
  tick label style={font=\scriptsize},label style={font=\small},
  legend style={at={(0.03,0.97)},anchor=north west,font=\scriptsize,draw=none,fill=white},
  samples=151,domain=0:1,clip=true]
\addplot[blue!75!black,thick] {(x+x^2-x^3)^2};
\addlegendentry{Tie: $(p+p^2-p^3)^2$}
\addplot[red!75!black,thick] {x^2+x^4-x^6};
\addlegendentry{Gem: $p^2+p^4-p^6$}
\addplot[black!45,thin,dotted] {x};
\addlegendentry{$\varphi(p)=p$}
\addplot[blue!75!black,dashed,thick,domain=0.43:0.86] {0.6710436067+1.6239089906*(x-0.6710436067)};
\addplot[red!75!black,dashed,thick,domain=0.58:0.97] {0.8191725134+1.6239089906*(x-0.8191725134)};
\addplot[only marks,mark=*,mark size=1.8pt,blue!75!black] coordinates {(0.6710436067,0.6710436067)};
\addplot[only marks,mark=square*,mark size=1.8pt,red!75!black] coordinates {(0.8191725134,0.8191725134)};
\node[font=\scriptsize,blue!75!black,anchor=east] at (axis cs:0.65,0.66) {$p_{c,\mathrm{tie}}$};
\node[font=\scriptsize,red!75!black,anchor=east] at (axis cs:0.795,0.85) {$p_{c,\mathrm{gem}}$};
\end{axis}
\end{tikzpicture}
\caption{The reliability functions.}
\label{fig:intro-tie-gem-response}
\end{minipage}\hfill
\begin{minipage}[c]{0.50\textwidth}
\centering\small
$\mathbf M_{\mathfrak C_{\mathrm{tie}}}(p_{c,\mathrm{tie}})\approx$
\[
\begin{pmatrix}
4.5605 & 1.0850 & 0.3546\\[2pt]
2.8567 & 2.1125 & 0.9118\\[2pt]
1.4284 & 0.2443 & 2.0798
\end{pmatrix}
\]
\vspace{0.5em}
$\mathbf M_{\mathfrak C_{\mathrm{gem}}}(p_{c,\mathrm{gem}})\approx$
\[
\begin{pmatrix}
5.1288 & 0.4886 & 0.2851\\[2pt]
3.4075 & 1.6239 & 0.7522\\[2pt]
1.7037 & 0.0000 & 2.0000
\end{pmatrix}
\]
\end{minipage}
\end{figure}

We display their reliability functions (see Definition~\ref{def:bernoulli}) alongside their critical percolation mass matrices (see Definition~\ref{def:live-matrix}).
We see that these two hierarchical lattices look different and have different reliability functions, critical probabilities $p_c$, and mass matrices.
However, the results in Table~\ref{tab:intro-tie-gem} below tell a different story.

\begin{table}[H]
\centering\scriptsize
\setlength{\tabcolsep}{2pt}
\renewcommand{\arraystretch}{1.25}
\resizebox{\textwidth}{!}{%
\begin{tabular}{l|cc|cccc|ccc|c}
\toprule
 & $\varphi(p)$ & $p_c$ & $\ell$ & $m$
 & $\rho(\mathbf M_{\mathfrak C})$ & $\varphi'(p_c)$
 & $\dimhaus(\Xi)$ & $\dimmass(\mathfrak C)$ & $\dim_{\mathrm P}(\mathfrak C)$
 & $\{\beta,\nu,\delta,\eta_{\mathrm{av}}\}$\\
\midrule
Tie & $(p+p^2-p^3)^2$ & $0.6710$ & $2$ & $6$
 & $5.7087$ & $1.6239$
 & $\dfrac{\log6}{\log2}$ & $\approx\dfrac{\log5.7087}{\log2}$ & $\approx\dfrac{\log1.6239}{\log2}$
 & \multirow{2}{*}{equal}\\
Gem & $p^2+p^4-p^6$ & $0.8192$ & $2$ & $6$
 & $5.7087$ & $1.6239$
 & $\dfrac{\log6}{\log2}$ & $\approx\dfrac{\log5.7087}{\log2}$ & $\approx\dfrac{\log1.6239}{\log2}$ & \\
\bottomrule
\end{tabular}}
\caption{Tie versus gem: the reliability maps and critical points differ, but the lattices have the same $\dimhaus(\Xi)$, $\dimmass(\mathfrak C)$ and $\dim_{\mathrm P}(\mathfrak C)$.}
\label{tab:intro-tie-gem}
\end{table}

The equality of their critical dimensions is exact and follows from the cyclic interchange of the two edge substitutions; see Lemma~\ref{lem:transposition-invariance}.
The examples in Figure~\ref{fig:terminal-symmetric-rules} show that distinct hierarchical lattices can belong to the same critical-exponent universality class.
The question is therefore how to identify such classes, and whether geometric data such as ambient dimension and connectivity suffice.

This pair therefore gives us a reason to formulate the problem carefully.
For instance, although hierarchical lattices appear frequently in the literature, what precise definition should we use? Which systems admit exact renormalisation, and what properties do these systems possess?
To address these questions, we introduce iterated graph systems.

\subsection{Edge iterated graph systems}
Griffiths and Kaufman~\cite{griffiths1982spin} gave a general definition of hierarchical lattices, using a hierarchy of units with uniformly bounded boundaries and a uniformly bounded number of constituent units.
We follow the distinction in~\cite{neroli2026iterated} between the generating system, its finite substitution networks, and the resulting hierarchical lattice.
Also see \cites{li2024on,xi2017fractality,li2025reducible}.

\begin{samepage}
\begin{definition}[\cite{neroli2024fractal}]\label{definition:EIGS}
A \emph{random edge iterated graph system} (random EIGS) is specified by the following data and iteration.
{\small
\begin{enumerate}[label=\arabic*.,leftmargin=*,topsep=3pt,itemsep=1pt,parsep=0pt]
\item \textbf{Initial graph.} A finite directed graph $\Xi^0$ with edge colours in $[K]:=\{1,\ldots,K\}$; parallel edges are allowed. For a single initial edge, write its endpoints as $\mathfrak v_+,\mathfrak v_-$.
\item \textbf{Rule laws.} For each colour, a fixed probability law on finitely many finite directed rule graphs, with the same colour set and two ordered distinct \emph{planting vertices}.
\item \textbf{Substitution.} Replace each edge $(a,b)$ by a fresh rule sampled according to its colour, identifying the planting vertices with $a,b$ and making no further identifications. Samples are independent across edges, with fresh randomness at each step.
\item \textbf{Iteration.} Simultaneous substitution produces the \emph{substitution networks} $\Xi^0,\Xi^1,\ldots$. The EIGS is \emph{deterministic} if every rule law is concentrated on one graph.
\end{enumerate}}
\end{definition}
\end{samepage}

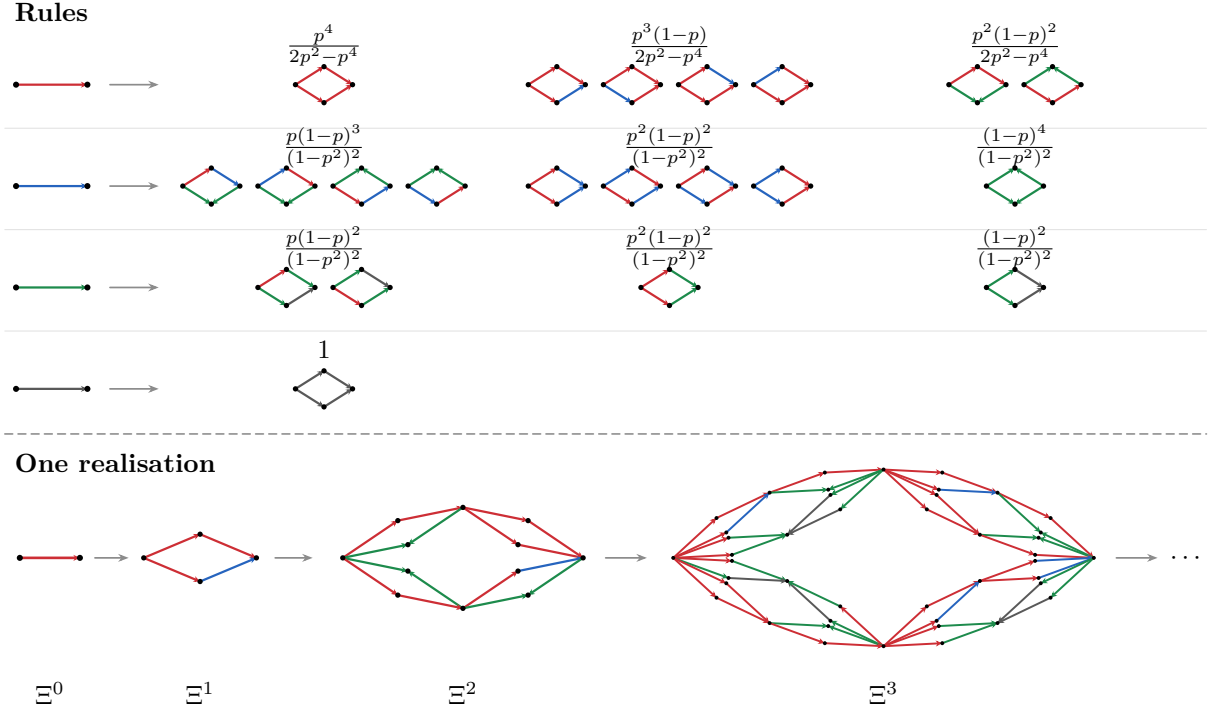
\begin{figure}[tbp]
\centering
\resizebox{\linewidth}{!}{\input{figure/random_DHL_rules.tex}}
\caption{A particular example of a random EIGS; it is induced by Bernoulli bond percolation on the DHL.}
\label{fig:random-eigs-dhl}
\end{figure}

For example, Figure~\ref{fig:random-eigs-dhl} gives the four-colour rules for critical percolation on the diamond hierarchical lattice, with $p=p_c=(\sqrt5-1)/2$.
Each displayed probability applies to each rule immediately below it; the bottom row shows one realisation of the first three substitutions.
Starting from one red edge, the red subgraph of $\Xi^n$ has the law of the level-$n$ percolation cluster conditioned on the two terminals being connected; the other colours record auxiliary terminal states.

Following~\cite{neroli2026iterated}, when the substitution networks admit a limit, we call it an \emph{EIGS hierarchical lattice}, distinguishing the combinatorial limit from the metric scaling limit when needed.
The EIGS is the generating system, and $(\Xi^n)_{n\in\mathbb N}$ is its sequence of finite approximating networks.
We retain the chosen EIGS construction when discussing a hierarchical lattice, so that its substitution steps and renormalisation scale are specified.

In this paper the underlying lattices are generated from one edge by deterministic single-colour systems.
Each edge is therefore replaced by the same two-terminal rule graph at every step.
Bernoulli percolation on such a deterministic system already produces a random coloured EIGS, whose colours record the two-terminal activity states of a cell.
This setting lets us study that mechanism while keeping the deterministic construction simple.
We use the term \emph{classical EIGS hierarchical lattice} for the following class.

\begin{definition}\label{def:standing}

A \textbf{classical EIGS hierarchical lattice} is an EIGS hierarchical lattice generated from a single directed edge by a deterministic single-colour EIGS whose two-terminal rule graph $R=(V(R),E(R);\beta^+,\beta^-)$ is finite, connected and simple, and satisfies the following conditions.
In particular, $R$ has no loops or parallel edges.
Connectivity, distances and percolation events refer to the underlying undirected graph of $R$.

\begin{enumerate}
\item \textbf{Canonicality.}
Every edge $e\in E(R)$ belongs to at least one simple path in the underlying undirected graph of $R$ connecting $\beta^+$ and $\beta^-$.
This excludes electrically irrelevant decorations which never participate in two-terminal connections.

\item \textbf{Non-degenerate two-terminal geometry.}
The terminal distance and the terminal edge-connectivity satisfy
\[
   \ell:=\dist_R(\beta^+,\beta^-)\ge 2,
   \qquad
   \lambda_R(\beta^+,\beta^-)\ge 2,
\]
where $\lambda_R(\beta^+,\beta^-)$ denotes the minimum size of an edge cut separating the two terminals.

\item \textbf{Terminal symmetry.}

The underlying undirected graph of $R$ admits an involutive automorphism $\vartheta:V(R)\longrightarrow V(R)$ that exchanges the terminals:
\[
   \vartheta^2=\mathrm{id},\qquad
   \vartheta(\beta^+)=\beta^-,\qquad
   \vartheta(\beta^-)=\beta^+.
\]
We call such a rule graph \emph{terminal-symmetric}.
This assumption identifies the left and right one-arm laws used below, so the state $u$ can be defined without any averaging convention.
\end{enumerate}
\end{definition}

The two geometric bounds in Definition~\ref{def:standing} characterise the existence of a nontrivial two-terminal crossing transition for a finite connected rule graph; see~\cite{alves2026percolation}*{Theorem~3.1 and Remark~5}.
In particular, once $\ell\ge2$, the condition $\lambda_R(\beta^+,\beta^-)\ge2$ is necessary and sufficient for a critical crossing parameter in $(0,1)$ separating vanishing from asymptotically certain terminal connection.
This criterion does not require canonicality or terminal symmetry.

Unless stated otherwise, the hierarchical lattices studied below are classical EIGS hierarchical lattices in the sense of Definition~\ref{def:standing}.
For a single-colour rule graph, we use the following notation throughout: $\ell:=\dist_R(\beta^+,\beta^-)$, and $m:=|E(R)|$.

\subsection{Main results}

Section~\ref{sec:bernoulli} develops the basic tools used throughout the paper.
In particular, Subsection~\ref{sec:percolation-preliminaries} extends the diamond construction of Hambly and Kumagai~\cite{hambly2010diffusion} to classical EIGS hierarchical lattices; some of these results are also contained in the recent work of Alves et al.~\cite{alves2026percolation}.
We regard the following theorem as a starting point for studying percolation throughout this class of hierarchical lattices.
\begin{restated}[Generalisation of the Hambly--Kumagai theorem~\cite{hambly2010diffusion}]{Theorem}{thm:percolation-random-eigs}
On every classical EIGS hierarchical lattice, critical Bernoulli bond percolation conditioned on terminal connection induces a random EIGS representation of its cluster.
\end{restated}
Using the established framework of random EIGS, we can study percolation more systematically.
For example, we analyse the mass matrix $\mathbf M_{\mathfrak C}$, which records the expected numbers of live child cells, and establish several important properties of this matrix.

Section~\ref{sec:universality} studies universality classes.
We examine the eight percolation exponents listed by Grimmett~\cite{grimmett2014criticality} and find that, under our observation conventions, four exist for every classical rule: $\beta,\nu,\delta,\eta_{\mathrm{av}}$.
\begin{restated}{Definition}{def:general-class}
Two classical EIGS hierarchical lattices belong to the same \emph{critical-exponent universality class} if their four exponents $(\beta,\nu,\delta,\eta_{\mathrm{av}})$ agree.
\end{restated}
We further prove the following characterisation.
\begin{restated}{Theorem}{thm:exponent-class-dimensions}
Two classical EIGS hierarchical lattices belong to the same critical-exponent universality class if and only if their ambient dimension $\dimhaus(\Xi)$, mass dimension $\dimmass(\mathfrak C)$ and pivotal dimension $\dim_{\mathrm P}(\mathfrak C)$ agree.
\end{restated}
More explicitly, the dimensions and exponents satisfy
\begin{align*}
 \nu&=\frac{1}{\dim_{\mathrm P}(\mathfrak C)},&
 \beta&=\frac{\dimhaus(\Xi)-\dimmass(\mathfrak C)}{\dim_{\mathrm P}(\mathfrak C)},\\
 \delta&=\frac{\dimmass(\mathfrak C)}{\dimhaus(\Xi)-\dimmass(\mathfrak C)},&
 \eta_{\mathrm{av}}&=2+\dimhaus(\Xi)-2\dimmass(\mathfrak C).
\end{align*}

With these results, we can naturally ask and answer a further question: can these critical dimensions themselves be determined by geometric measurements?
For a positive functional $\Lambda$ multiplicative under uniform edge substitution, its associated dimension is $\log\Lambda(R)/\log\ell_R$.
Reordering substitutions gives the second main result of Section~\ref{sec:universality}.
\begin{restated}{Theorem}{thm:no-multiplicative-classification}
No finite family of substitution-multiplicative dimensions determines the critical-exponent universality class.
In particular, the ambient Hausdorff dimension $\dimhaus(\Xi)$, as a substitution-multiplicative dimension, cannot determine this class.
\end{restated}
This gives a structural proof of Hauser and Saxena's observation~\cite{hauser1986geometrical} that ambient dimension alone cannot determine percolation universality.
The critical mass dimension can retain information about substitution order that substitution-multiplicative dimensions cannot detect.

Discrete scale invariance is an important concept in physics: it describes self-similarity under scale transformations by integer powers of a fixed ratio and can manifest itself through log-periodic corrections to power-law scaling; see Sornette~\cite{sornette1998discrete} and Derrida and Giacomin~\cite{derrida2014logperiodic}.
For a hierarchical lattice, each substitution has a fixed length multiplier $\ell$, and iteration naturally produces the discrete scales $\ell,\ell^2,\ell^3,\ldots$.
In the context of critical-exponent universality classes, this raises a natural question: do identical critical exponents require compatible discrete renormalisation scales?
Specifically, if two hierarchical lattices belong to the same critical-exponent universality class and have length scales $\ell_1$ and $\ell_2$, must there be positive integers $a,b$ such that
\[
\ell_1^a=\ell_2^b,
\qquad\text{equivalently,}\qquad
\frac{\log\ell_1}{\log\ell_2}\in\mathbb Q?
\]
Physically, this means that grouping finitely many hierarchical levels on each side gives a common discrete coarse-graining scale.
Section~\ref{sec:scale} addresses this question.

The six exponentials theorem~\cite{dasgupta2021transcendence} connects this physical question to a problem in number theory.
Finite rule graphs produce algebraic volume, mass and pivotal growth multipliers, whose logarithms, divided by the logarithm of the length scale, give the three critical dimensions.
For two lattices $\Xi_1,\Xi_2$ in the same class, with critical clusters $\mathfrak C_1,\mathfrak C_2$, exponentiating the products of their scale logarithms and their common critical dimensions gives the six algebraic growth multipliers in Table~\ref{tab:six-exponentials-dimensions}.
If the scales were incommensurate and the three dimensions linearly independent over $\mathbb Q$, this would contradict the theorem, which requires at least one of the six numbers to be transcendental.

\begin{table}[htbp]
\centering\small
\renewcommand{\arraystretch}{1.35}
\begin{minipage}{0.34\textwidth}
\centering
\[
\begin{array}{c|ccc}
 & y_1&y_2&y_3\\\hline
 x_1&e^{x_1y_1}&e^{x_1y_2}&e^{x_1y_3}\\
 x_2&e^{x_2y_1}&e^{x_2y_2}&e^{x_2y_3}
\end{array}
\]
\end{minipage}\hfill
\begin{minipage}{0.63\textwidth}
\centering
\[
\begin{array}{c|ccc}
 &\dimhaus(\Xi_1)&\dimmass(\mathfrak C_1)&\dim_{\mathrm P}(\mathfrak C_1)\\\hline
 \log\ell_1&m_1&\rho(\mathbf M_{\mathfrak C_1})&\varphi_{R_1}'(p_{c,1})\\
 \log\ell_2&m_2&\rho(\mathbf M_{\mathfrak C_2})&\varphi_{R_2}'(p_{c,2})
\end{array}
\]
\end{minipage}
\caption{Six exponentials and percolation exponents.}
\label{tab:six-exponentials-dimensions}
\end{table}

We consequently obtain the following partial commensurability result.
\begin{restated}{Theorem}{thm:irreducible-commensurability}
Suppose a classical EIGS hierarchical lattice satisfies
\[
 \frac{\varphi_R(p)-p}{p(1-p)}\text{ is irreducible over }\mathbb Q,
 \qquad \rho(\mathbf M_{\mathfrak C})\notin\mathbb Q(p_c).
\]
Then the scales of all classical EIGS hierarchical lattices in its critical-exponent universality class are pairwise commensurate.
\end{restated}
For example, every hierarchical lattice in the DHL class has scale commensurate with $2$.
These results led us to expect commensurability within every universality class.

However, Section~\ref{sec:counterexample} constructs infinitely many models in one class whose scales remain pairwise incompatible after any finite grouping of substitution steps.
\begin{restated}{Theorem}{thm:infinite-incommensurate-scales}
For every integer $k\ge0$, there is a classical EIGS hierarchical lattice $\Xi$ with scale $(19\cdot661^k)^{100}$ and critical dimensions
\[
 (\dim_{\mathrm H}(\Xi),\dim_{\mathrm M}(\mathfrak C),\dim_{\mathrm P}(\mathfrak C))
 =\left(\frac{58}{25},\frac{219}{100},\frac7{10}\right).
\]
These scales are pairwise incommensurate.
\end{restated}

We further construct an example in which all three critical dimensions of a hierarchical lattice are transcendental but even rationally proportional.

Agreement of all substitution-multiplicative dimensions is insufficient for critical-exponent universality, while commensurability of the underlying substitution scales is not necessary. 
A single critical-exponent universality class can cross infinitely many pairwise incommensurate scale classes. 
Thus, critical-exponent universality does not require a common discrete coarse-graining scale, even after finite regrouping of hierarchical levels.

\subsection{Lean formalisation}

The accompanying Lean~4 development, built on \texttt{mathlib}, brings the results of Sections~2--5 into a single formal framework. It includes the Bernoulli substitution and spectral foundations, the classification by four critical exponents and three dimensions, the noncommutative obstruction, the arithmetic criteria for scale commensurability, and the actual graph constructions underlying the counterexamples. The conditional results for the additional exponents in the appendix are also included, with their hypotheses retained.

Apart from Lean's standard logical axioms, the only external mathematical inputs are the Gelfond--Schneider theorem and the six exponentials theorem. The formal proofs of these conclusions are checked by Lean's kernel, without proof placeholders, for the specified observables and under the stated hypotheses. The four exponentials conjecture appears only as an explicit hypothesis of conditional results. The formal statements, proofs, exact certificates, coverage description and verification instructions are provided in the accompanying \href{https://github.com/Nero-17/LEAN-Formalisation-IGS-Bernoulli-percolation-universality-class}{Lean repository}.

\noindent{\footnotesize\url{https://github.com/Nero-17/LEAN-Formalisation-IGS-Bernoulli-percolation-universality-class}}

\section{Preliminaries: Bernoulli percolation on hierarchical lattices}\label{sec:bernoulli}

In this section, we establish the basic properties of Bernoulli percolation on classical EIGS hierarchical lattices that underlie the subsequent analysis.

Subsection~\ref{sec:percolation-preliminaries} concerns the existence of a nontrivial crossing transition and the representation of critical terminal clusters by a random EIGS. It includes both a generalisation of the construction of Hambly and Kumagai~\cite{hambly2010diffusion} and results overlapping with the recent work of Alves, Baldasso, Moreira and Teixeira~\cite{alves2026percolation}.
For completeness, we include these results and explicitly indicate their relationship to the existing literature.

Subsection~\ref{sec:live-matrix} examines the mass-transfer matrix of the random EIGS induced by percolation and its spectral structure.
These properties provide a foundation for understanding critical clusters and their critical exponents.

\subsection{From percolation to random EIGS}\label{sec:percolation-preliminaries}

On Euclidean lattices, Bernoulli percolation is generally difficult to solve exactly;
on classical EIGS hierarchical lattices, exact renormalisation provides a recursive approach to crossing probabilities and terminal clusters.
Moreover, at criticality and conditional on the two terminals being connected, the terminal cluster can be represented by a random EIGS, allowing us to study its properties through this system.

\begin{definition}\label{def:bernoulli}
On each substitution graph $\Xi^n$, a configuration is a subset $\omega\subseteq E(\Xi^n)$ of open edges. Under $\mathbb P_p$, these edges are independently open with probability $p\in[0,1]$, so a particular configuration has probability $p^{|\omega|}(1-p)^{|E(\Xi^n)|-|\omega|}$. Its clusters are the connected components of the graph with vertex set $V(\Xi^n)$ and edge set $\omega$, including isolated vertices. The inherited terminals are denoted by $\mathfrak v_+,\mathfrak v_-$.

For a finite two-terminal graph $G$, its reliability is
\[
 \varphi_G(p):=\mathbb P_p(\text{the terminals of }G\text{ are connected}).
\]
\end{definition}
Summing the configuration probabilities over the connecting edge subsets makes $\varphi_G$ an integer polynomial. We study the terminal crossing transition and the conditional clusters seen from the terminals. The finite graphs need not form a nested sequence of edge sets; replacing an edge removes it and inserts a new cell.

\begin{lemma}[Also see \cite{alves2026percolation}]\label{lem:renormalisation}
For every $n\ge 1$ and every $p\in[0,1]$,
\[
    \varphi_{\Xi^{n}:\mathfrak v_+\leftrightarrow \mathfrak v_-}(p)
    =
    \varphi_{R:\beta^+\leftrightarrow \beta^-}^{\circ n}(p),
\]
where $\varphi_R^{\circ n}$ denotes the $n$-fold composition of $\varphi_R$ with itself.
\end{lemma}
\begin{proof}
By construction $\Xi^{n}$ is obtained from the rule graph $R$ by substituting a fresh copy of $\Xi^{n-1}$ for each of its edges, the copies being glued only at the planting vertices.
Distinct copies share no edge, so under $\mathbb P_p$ their percolation configurations are independent, and each copy connects the two planting vertices of the edge it replaces with probability $\varphi_{\Xi^{n-1}}(p)$.

The crossing event $\{\mathfrak v_+\leftrightarrow\mathfrak v_-\}$ in $\Xi^{n}$ depends on the percolation configuration only through the family of crossing events of the substituted copies, one per edge of $R$.
These indicators are independent and each has mean $\varphi_{\Xi^{n-1}}(p)$, and the global crossing occurs precisely when the corresponding edges of $R$ form a crossing of $R$.
Hence the crossing of $\Xi^{n}$ has the law of the crossing of $R$ with independent edge-occupation $\varphi_{\Xi^{n-1}}(p)$, which gives
\[
    \varphi_{\Xi^{n}}(p)=\varphi_{R}\bigl(\varphi_{\Xi^{n-1}}(p)\bigr).
\]

For $n=1$ the substitution replaces the single initial edge by one copy of $R$, so $\varphi_{\Xi^{1}}=\varphi_R$.
Iterating the displayed recursion and using $\varphi_{\Xi^{1}}=\varphi_R$ yields $\varphi_{\Xi^{n}}=\varphi_R^{\circ n}$.
\end{proof}

\begin{proposition}[Also see \cite{alves2026percolation}]\label{prop:nontrivial-transition}
For a finite connected two-terminal rule graph $R$, the function $\varphi_R(p)-p$ changes sign in $(0,1)$ if and only if
\[
 \dist_R(\beta^+,\beta^-)\ge2,
 \qquad \lambda_R(\beta^+,\beta^-)\ge2.
\]
\end{proposition}
\begin{proof}
If the terminal distance is at least two, a union bound over the finitely many simple terminal paths gives $\varphi_R(p)=O(p^2)$ as $p\downarrow0$.
If the minimum terminal cut has at least two edges, a union bound over the finitely many terminal cuts gives $1-\varphi_R(p)=O((1-p)^2)$ as $p\uparrow1$.
Thus $\varphi_R(p)<p$ near zero and $\varphi_R(p)>p$ near one.
Continuity gives an interior sign change.

Conversely, if the terminals are adjacent, opening a direct edge suffices for connection, so $\varphi_R(p)\ge p$ throughout $[0,1]$.
If the minimum terminal cut has size one, its unique edge must be open for connection, so $\varphi_R(p)\le p$ throughout $[0,1]$.
Neither case permits a sign change.
\end{proof}

This criterion does not require canonicality or terminal symmetry; see also~\cite{alves2026percolation}*{Remark~5}.
When the terminal distance is at least two and the minimum cut has size one, the latter inequality is strict for $0<p<1$: opening only the separating edge does not connect the terminals.
Consequently $\displaystyle\lim_{n\to\infty}\varphi_R^{\circ n}(p)=0$ for every $p<1$.
The condition $\lambda_R(\beta^+,\beta^-)\ge2$ therefore excludes precisely this trivial crossing regime once $\ell\ge2$.
Under the two geometric bounds, the following theorem gives uniqueness and strict instability of the interior fixed point, as also established in~\cite{alves2026percolation}*{Theorem~3.1}.

\begin{theorem}[Also see \cite{alves2026percolation}]
\label{thm:unique-unstable-critical-point}
For a finite connected two-terminal rule graph with terminal distance at least two and minimum terminal edge cut at least two, the reliability function $\varphi_R$ has a unique fixed point $p_c\in(0,1)$, and this fixed point is unstable:
\[
    \varphi_R'(p_c)>1.
\]
\end{theorem}

\begin{proof}
For $p\in(0,1)$, Russo's formula~\cite{grimmett1989percolation} and the Bernoulli Poincar\'e inequality~\cite{odonnell2014analysis} give $\varphi_R(p)(1-\varphi_R(p))\le p(1-p)\varphi_R'(p)$. Equality in the Poincar\'e inequality forces every term of degree at least two in the $p$-biased Walsh expansion to vanish, so a monotone Boolean function attaining equality is constant or depends on a single coordinate. The crossing indicator is nonconstant and, since the terminal distance is at least two, cannot be determined by a single edge, so the inequality is strict. Consequently, $\log\frac{\varphi_R(p)}{1-\varphi_R(p)}-\log\frac{p}{1-p}$ is strictly increasing on $(0,1)$. Union bounds over terminal paths and separating cuts give $\varphi_R(p)=O(p^{\ell_R})$ as $p\downarrow0$ and $1-\varphi_R(p)=O((1-p)^{\lambda_R(\beta^+,\beta^-)})$ as $p\uparrow1$. Since both geometric exponents are at least two, this log-odds difference tends to $-\infty$ at zero and to $+\infty$ at one. It therefore has exactly one zero $p_c\in(0,1)$, equivalently $\varphi_R(p_c)=p_c$. Evaluating the strict inequality at $p_c$ and dividing by $p_c(1-p_c)$ yields $\varphi_R'(p_c)>1$.
\end{proof}

\begin{corollary}[Also see \cite{alves2026percolation}]
\label{cor:global-renormalisation-dichotomy}
Under the assumptions of Theorem~\ref{thm:unique-unstable-critical-point},
the unique critical fixed point \(p_c\in(0,1)\) separates the renormalisation
flow:
\[
    \varphi_R(p)<p
    \quad\text{for }0<p<p_c,
    \qquad
    \varphi_R(p)>p
    \quad\text{for }p_c<p<1.
\]
Consequently,
\[
    \displaystyle\lim_{n\to\infty}\varphi_R^{\circ n}(p)
    =
    \begin{cases}
        0, & 0\le p<p_c,\\[2mm]
        p_c, & p=p_c,\\[2mm]
        1, & p_c<p\le 1.
    \end{cases}
\]
\end{corollary}

\begin{example}\label{ex:dhl-reliability}
For DHL,
\[
 \varphi_R(p)=1-(1-p^2)^2=2p^2-p^4,\qquad
 p_c=\frac{\sqrt5-1}{2},\qquad
 \varphi_R'(p_c)=4p_c(1-p_c^2)=6-2\sqrt5.
\]
The fixed point therefore has a strictly expanding pivotal response, although the probability of crossing a cell is unchanged there.
\end{example}


Hambly and Kumagai~\cite{hambly2010diffusion} described the critical two-terminal cluster on the diamond lattice through a finite-type random recursive construction and its geometric limit.
We formulate this construction in the language of random EIGS and extend it to every classical EIGS hierarchical lattice.
The random EIGS specifies the recursive generating mechanism, and its finite approximations are random substitution networks. The critical cluster has a graph-directed random recursive scaling limit in the ambient metric.
The auxiliary types record which terminals of a cell belong to the cluster being followed; they are needed even though the final cluster consists only of open edges.

Let $A^+$ and $A^-$ be the open clusters of the two terminals inside a rule cell.
A parent of state $c$ uses the event $\beta^+\leftrightarrow\beta^-$ and its common terminal cluster as the active set.
A parent of state $b$ uses the complementary event and active set $A^+\cup A^-$.
A parent of state $u$ uses the complementary event and active set $A^+$.
Terminal symmetry makes the corresponding one-terminal expectations independent of which terminal is chosen.
For a union $A$ of active open clusters, the child state of $e=\{x,y\}$ is
\[
 \operatorname{st}_{A,\omega}(e):=
 \begin{cases}
 c,&x,y\in A\text{ and }e\text{ is open},\\
 b,&x,y\in A\text{ and }e\text{ is closed},\\
 u,&\text{exactly one of }x,y\text{ lies in }A,\\
 0,&x,y\notin A.
 \end{cases}
\]
An edge in state $u$ is necessarily closed.
The state $0$ records no active endpoint and produces no live descendant under refinement.

\begin{definition}\label{def:random-eigs}
Fix $p\in(0,1)$ and a live parent state $\sigma\in\{c,b,u\}$. Sample a Bernoulli configuration with parameter $p$ in its rule cell conditioned on the parent event just specified, and label every child edge by its active state. Repeat this operation inside the live children, with independent conditional samples for distinct child cells once their parent states are fixed. A state-$0$ cell has no live descendants. For a state-$u$ child its active terminal is inherited from the parent; terminal symmetry identifies the two possible orientations in distribution.
\end{definition}

For $n\ge0$, let $\mathfrak C_n$ denote the common open cluster of the two terminals in $\Xi^n$, under critical Bernoulli percolation conditioned on their connection.
The cluster includes all open branches attached to the terminals, not only edges lying on an open terminal-to-terminal path.

\begin{theorem}[Generalisation of the Hambly--Kumagai theorem~\cite{hambly2010diffusion}]\label{thm:percolation-random-eigs}
On every classical EIGS hierarchical lattice, critical Bernoulli bond percolation conditioned on terminal connection induces a random EIGS representation of its cluster.
\end{theorem}
\begin{proof}
Fix a finite depth and coarsen each bottom-level cell to an open edge precisely when its two terminals are connected inside that cell.
The cells have disjoint edge sets, so their crossing indicators are independent.
Each indicator has probability $\varphi_R(p_c)=p_c$.
Repeating the coarsening therefore gives the critical Bernoulli law at every level.

Conditioning on the full coarse configuration fixes whether each child cell connects its terminals.
The configurations in distinct child cells remain independent under these individual conditions, since the original edge sets are disjoint.
Within a connected cell, the active part is its common terminal cluster.
Within a disconnected cell, the active part is the union of its two terminal clusters, one specified terminal cluster, or neither, according to the active endpoints inherited from the parent.
These are exactly the conditional laws for $c,b,u,0$ in Definition~\ref{def:random-eigs}.
For a one-terminal cell, the active endpoint is retained; terminal symmetry permits its identification with the reference orientation used for state $u$.
Offspring labels within one cell may be dependent, as specified by the sampled configuration.

Starting with the root conditioned to connect, induction on the depth now identifies the joint laws of all labelled cells.
At the last level, an open edge belongs to the root cluster exactly when it has state $c$.
Extracting these edges proves the assertion about $\mathfrak C_n$.
\end{proof}

The fixed-point identity is essential to stationarity of this representation.
For a general parameter $p$, the analogous finite-depth construction uses the successive effective probabilities $\varphi_R^{\circ k}(p)$ and hence level-dependent replacement laws.

Equip $V(\mathfrak C_n)$ with the restriction of the rescaled ambient graph metric $\ell^{-n}d_{\Xi^n}$.
Thus distances are measured in the underlying lattice, rather than along open paths inside the cluster.
The natural recursive coupling is given by Definition~\ref{def:random-eigs} at $p=p_c$, starting from state $c$.
The labels $b,u,0$ are auxiliary states; the state-$c$ edges form the open subgraph extracted from the random EIGS.
We use the ambient metric scaling of~\cite{neroli2024fractal} and extend the graph-directed construction of Hambly and Kumagai~\cite{hambly2010diffusion}*{Section~3.2}.

\begin{theorem}\label{thm:eigs-graph-directed}
On every classical EIGS hierarchical lattice, the critical percolation clusters conditioned on terminal connection, equipped with the rescaled ambient metric, converge almost surely under the natural recursive coupling in the Gromov--Hausdorff topology to a graph-directed random recursive fractal.
\end{theorem}
\begin{proof}
First construct the ambient compact metric space from the rule graph $R$.
Every shortest path between old vertices traverses each visited replacement cell between its terminals, whose distance is $\ell$.
Conversely, replacing each edge of a shortest old path by a shortest terminal path gives
\[
 d_{\Xi^{n+1}}(x,y)=\ell\,d_{\Xi^n}(x,y),
 \qquad x,y\in V(\Xi^n).
\]
Consequently, the metrics $\ell^{-n}d_{\Xi^n}$ agree on the nested vertex sets.
Let $(K,d_K)$ be the completion of their union with this common metric.
Every new vertex at level $j$ lies within $\ell^{-j}\operatorname{diam}(R)$ of a vertex at level $j-1$.
Successively projecting to earlier levels shows that $V(\Xi^n)$ is a finite net with radius at most
\[
 \operatorname{diam}(R)\sum_{j=n+1}^{\infty}\ell^{-j}
 =\frac{\operatorname{diam}(R)}{\ell-1}\ell^{-n}.
\]
Hence $K$ is compact.

For each $e\in E(R)$, the descendant copy of the entire construction defines a map $\psi_e:K\longrightarrow K$.
This map is a similarity of ratio $\ell^{-1}$.
Indeed, a path leaving the copy and returning through its other terminal cannot be shorter than the internal terminal path: it must traverse other cells of the same depth, each with that same terminal distance.
Thus the copy's internal distances agree with the ambient distances, and the assertion follows first on finite vertex sets and then by completion.
Distinct first-level copies meet only at their prescribed common endpoints, since any path between them passes through their terminal boundaries.
In particular, $K=\bigcup_{e\in E(R)}\psi_e(K)$, and every depth-$n$ cell has diameter $\ell^{-n}\operatorname{diam}(K)$.

Use the independent conditional samples of Definition~\ref{def:random-eigs} to couple all levels.
By Theorem~\ref{thm:percolation-random-eigs}, the state-$c$ subgraph at depth $n$ has the law of $\mathfrak C_n$.
Retain the active endpoint of every state-$u$ cell as part of its marking.
Every active vertex persists under refinement: a state-$c$ edge is replaced by a connected terminal cluster, and each active terminal of a state-$b$ or state-$u$ cell remains active.
The coupled sets $V(\mathfrak C_n)$ are consequently nested in $K$.
Define the nonempty compact set
\[
 \mathfrak C_\infty:=
 \overline{\bigcup_{n\ge0}V(\mathfrak C_n)}^{\,K}.
\]
Compactness and nesting give Hausdorff convergence.
Indeed, for every $\varepsilon>0$, density of the union provides an $\varepsilon$-net of $\mathfrak C_\infty$ whose centres belong to this union; compactness allows the net to be chosen finite.
All its centres then belong to $V(\mathfrak C_n)$ for every sufficiently large $n$.
Since $V(\mathfrak C_n)\subseteq\mathfrak C_\infty$, this proves
\[
 \lim_{n\to\infty}
 d_{\mathrm{Haus}}^K\!\left(V(\mathfrak C_n),\mathfrak C_\infty\right)=0.
\]
The passage from nonempty compact subsets of $K$ to their Gromov--Hausdorff isometry classes is continuous.
The inherited metrics therefore give the asserted Gromov--Hausdorff convergence on every coherent history, and such histories have probability one under the recursive coupling.

To identify the limit, apply the same closure construction with root types $c,b,u_+,u_-$, where $u_+$ and $u_-$ specify the active terminal, and write the resulting sets as $\mathfrak C_\infty^\sigma(\omega)$.
Set $\mathfrak C_\infty^0:=\varnothing$.
For a first-level rule sample $\omega$, denote the child type of $e$ by $\operatorname{st}_{\sigma,\omega}(e)$ and its subsequent samples by $\omega_e$.
The finite-level recursion, followed by closure in $K$, gives
\[
 \mathfrak C_\infty^\sigma(\omega)
 =\bigcup_{e\in E(R)}
   \psi_e\!\left(
     \mathfrak C_\infty^{\operatorname{st}_{\sigma,\omega}(e)}(\omega_e)
   \right).
\]
Each active vertex belongs to a live child through one of its active endpoints, so this equation also retains isolated active terminals.
The construction graph has one vertex for each live type and one directed edge for each possible live child occurrence, carrying its contraction and endpoint data.
An entire offspring family is sampled jointly according to the conditional rule law; samples in distinct cells are independent given their types.
This is a finite-type graph-directed random recursive construction, and $\mathfrak C_\infty=\mathfrak C_\infty^c$.
Terminal symmetry identifies $u_+$ and $u_-$ after reversing their endpoint markings, yielding the three-type description $c,b,u$.
\end{proof}


The mass transfer matrix of the random EIGS records the expected numbers of live child edges of each type.
Terminal symmetry combines the two oriented one-terminal types into the single state $u$, leaving the three states $c,b,u$.
The matrix and expected-mass results below use only the finite-level representation.

\begin{definition}\label{def:live-matrix}
For $p\in(0,1)$, $\mathbf M_{\mathfrak C}(p)$ is the matrix whose entry $M_{\sigma,\tau}(p)$ is the expected number of state-$\tau$ child edges under the conditional parent law of state $\sigma$, with rows and columns ordered as $(c,b,u)$.
At criticality, we write $\mathbf M_{\mathfrak C}:=\mathbf M_{\mathfrak C}(p_c)$ and denote its spectral radius by $\rho(\mathbf M_{\mathfrak C})$.
All entries refer to the conditional live-state matrix, including when we later pass to an invariant two-dimensional block.
\end{definition}

\begin{proposition}\label{prop:live-matrix-primitive}
For every admissible rule graph and every $p\in(0,1)$, the three-state live matrix $\mathbf M_{\mathfrak C}(p)$ is primitive.
\end{proposition}

\begin{proof}
The rule graph has no bridge.
Indeed, by canonicality a bridge belongs to a simple terminal-to-terminal path, and its deletion would therefore separate the terminals, contrary to the terminal edge-connectivity assumption.
Opening every edge except one consequently gives a connected parent cell containing a child of state $b$, so $M_{c,b}(p)>0$.
Closing every edge gives a disconnected parent cell whose active set for state $b$ consists of the two terminals.
Since the terminals are not adjacent, an edge incident to either terminal then has state $u$, so $M_{b,u}(p)>0$.
Opening exactly one edge incident to $\beta^+$ leaves the terminals disconnected and gives a child of state $c$ in a parent of state $u$, so $M_{u,c}(p)>0$.
These entries give the directed cycle $c\to b\to u\to c$.
The all-open configuration gives $M_{c,c}(p)>0$, so this irreducible matrix is aperiodic and therefore primitive.
\end{proof}

\begin{definition}\label{def:critical-mass-dimension}
The expected critical mass dimension is, when the limit exists,
\[
 \dimmass(\mathfrak C):=\lim_{n\to\infty}
 \frac{\log\mathbb E[|E(\mathfrak C_n)|]}{\log\ell^n},
\]
where the expectation is under critical percolation conditioned on terminal connection.
\end{definition}
This is an expected growth exponent; its identification with an almost-sure Hausdorff dimension of a random scaling limit is a further geometric statement.

\begin{lemma}[Also see \cite{alves2026percolation}]\label{thm:mass-transfer-power}
At $p=p_c$, for every $n\ge1$ and all live states $\sigma,\tau\in\{c,b,u\}$, the expected number of state-$\tau$ live edges of $\Xi^{n}$ contained in a level-zero block of state $\sigma$ equals the $(\sigma,\tau)$ entry of $\bigl(\mathbf M_{\mathfrak C}(p_c)\bigr)^{n}$.
\end{lemma}

\begin{proof}
The case $n=1$ is the definition of the live matrix, since at the critical point the edge occupation entering a single substitution cell is $p_c$.

Replacing every edge of $\Xi^{n}$ by a rule cell at occupation $p_c$ turns each cell into a coarse edge that is open with probability $\varphi_R(p_c)=p_c$, so every level of the substitution sees the same edge occupation $p_c$.
Conditionally on the states of the level-$n$ edges, the configurations of the cells refining distinct edges are independent, and the conditional law inside a cell of coarse state $\tau'$ is the conditioned percolation law defining the row $\tau'$ of the live matrix.
Hence the expected state counts at level $n+1$ inside a state-$\sigma$ block are obtained from those at level $n$ by one application of $\mathbf M_{\mathfrak C}(p_c)$ on the right, and the claim follows by induction.
\end{proof}

\begin{theorem}[Also see \cite{alves2026percolation}]\label{thm:critical-mass-dimension}
For every classical EIGS hierarchical lattice, the expected critical mass dimension exists and
\[
 \dimmass(\mathfrak C)=\frac{\log\rho(\mathbf M_{\mathfrak C})}{\log\ell}.
\]
\end{theorem}
\begin{proof}
By Theorem~\ref{thm:percolation-random-eigs}, the edges of $\mathfrak C_n$ are exactly the state-$c$ edges in the random representation started from state $c$.
Lemma~\ref{thm:mass-transfer-power} therefore gives
\[
 \mathbb E[|E(\mathfrak C_n)|]
 =\bigl(\mathbf M_{\mathfrak C}(p_c)^n\bigr)_{c,c}.
\]
The matrix is primitive by Proposition~\ref{prop:live-matrix-primitive}.
The Perron--Frobenius theorem implies that the right-hand side is $C\rho(\mathbf M_{\mathfrak C})^{\,n}(1+o(1))$ for a constant $C>0$.
Taking logarithms and dividing by $n\log\ell$ proves the formula.
\end{proof}

The total expected number of live cells has the same exponent, because every entry of the Perron projection is positive.
Counting actual open cluster edges or all live cells thus gives the same logarithmic growth rate, although the finite-stage expectations differ.

\subsection{Spectral structure of the mass matrix}\label{sec:live-matrix}

In the preceding subsection, we showed how the critical terminal cluster on every classical EIGS hierarchical lattice can be represented by a random EIGS, conditional on terminal connection.
We now examine the mass-transfer matrix of this percolation-induced random EIGS, identify its spectral structure and obtain a spectral expansion.
An interesting feature is that the thermal response multiplier $\varphi_R'(p_c)$ is encoded in this matrix: it is one of its eigenvalues at criticality.

\begin{definition}\label{def:counting-polynomials}
Let $\sigma,\tau\in\{c,b,u\}$.
The counting polynomial $U^{\sigma}_{\tau}(p)$ is the expected number, under $\mathbb P_p$, of edges of $R$ whose state is $\tau$, computed with respect to the active set of the parent state $\sigma$ and restricted to the parent event of $\sigma$.
For $\sigma=c$ the parent event is $\{\beta^+\leftrightarrow\beta^-\}$ and the active set is the common open cluster $A$ of the two terminals.
For $\sigma=b$ and $\sigma=u$ the parent event is $\{\beta^+\nleftrightarrow\beta^-\}$, and the active sets are $A^{+}\cup A^{-}$ and $A^{+}$ respectively, where $A^{\pm}$ denotes the open cluster of $\beta^{\pm}$.
\end{definition}

Each $U^{\sigma}_{\tau}$ is a polynomial with nonnegative integer coefficients in the Bernstein basis $p^{k}(1-p)^{m-k}$, and by construction
\[
   M_{c,\tau}=\frac{U^{c}_{\tau}}{\varphi_R},
   \qquad
   M_{b,\tau}=\frac{U^{b}_{\tau}}{1-\varphi_R},
   \qquad
   M_{u,\tau}=\frac{U^{u}_{\tau}}{1-\varphi_R}.
\]

For a finite graph $G$ with two distinguished vertices, a bond is \emph{pivotal} for their connection if flipping its state changes whether they are connected; write $N_{\mathrm{piv}}(G)$ for the number of such bonds.
On the connected event all pivotal bonds are open, whereas on the disconnected event they are precisely the closed edges joining the two terminal clusters.

\begin{lemma}\label{lem:exact-row-identities}\label{lem:conditional-pivotal-counts}\label{prop:universal-eigenvectors}
For every admissible rule graph the following identities hold in $\mathbb Z[p]$:
\[
   (1-p)\,U^{c}_{c}-p\,\bigl(U^{c}_{b}+U^{c}_{u}\bigr)=p(1-p)\,\varphi_R'(p),
\]
\[
   (1-p)\,U^{u}_{c}-p\,\bigl(U^{u}_{b}+U^{u}_{u}\bigr)=-\,p(1-p)\,\varphi_R'(p).
\]
If the rule graph is in addition terminal symmetric, then
\[
   U^{b}_{c}=2\,U^{u}_{c},
   \qquad
   U^{b}_{b}=2\,U^{u}_{b}+(1-p)\,\varphi_R'(p),
   \qquad
   U^{b}_{u}=2\,U^{u}_{u}-2(1-p)\,\varphi_R'(p).
\]
Consequently, for every classical EIGS hierarchical lattice and every $p\in(0,1)$,
\[
 \bigl(0,-\tfrac12,1\bigr)\mathbf M_{\mathfrak C}(p)
 =\frac{(1-p)\,\varphi_R'(p)}{1-\varphi_R(p)}
 \bigl(0,-\tfrac12,1\bigr).
\]
In particular, at $p=p_c$ the corresponding eigenvalue is $\varphi_R'(p_c)$.
\end{lemma}

\begin{proof}
By Russo's formula, the expected number of pivotal bonds is $\varphi_R'(p)$.
Whether a bond is pivotal is independent of its own state; an open pivotal bond forces connection and a closed one forces disconnection. Hence
\[
 \mathbb E_p[N_{\mathrm{piv}}(R)\mid\beta^+\leftrightarrow\beta^-]
 =\frac{p\varphi_R'(p)}{\varphi_R(p)},\qquad
 \mathbb E_p[N_{\mathrm{piv}}(R)\mid\beta^+\nleftrightarrow\beta^-]
 =\frac{(1-p)\varphi_R'(p)}{1-\varphi_R(p)}.
\]
Both expectations equal $\varphi_R'(p_c)$ at criticality.

For any event $E$ on configurations, differentiation gives
\[
   p(1-p)\,\frac{d}{dp}\,\mathbb P_p(E)
   =\sum_{\omega\in E}
   \bigl((1-p)|\omega|-p(m-|\omega|)\bigr)
   p^{|\omega|}(1-p)^{m-|\omega|}.
\]

Fix the parent state $\sigma=c$ and condition on the active set $A=A_0$.
The conditioning event constrains only the interior edges and closes every edge on the boundary of $A_0$, leaving the edges with both endpoints outside $A_0$ free and independent.
For a free edge the open weight $(1-p)\cdot p$ and the closed weight $p\cdot(1-p)$ agree, so its contribution to the differentiated sum vanishes.
An open edge cannot have exactly one endpoint in $A_0$, so the counted open edges are exactly the state-$c$ edges and the counted closed edges are exactly the state-$b$ and state-$u$ edges.
With $E=\{\beta^+\leftrightarrow\beta^-\}$ this gives the first identity, and with $E$ the complementary event and active set $A^{+}$ it gives the second, since the conditioning on $A^{+}=A_0$ with $\beta^-\notin A_0$ again leaves the outside edges free.

On the disconnected event, an edge with both endpoints in $A^{+}\cup A^{-}$ either has both endpoints in the same terminal cluster or joins the two clusters, in which case it is closed and pivotal for the crossing.
The conditional pivotal identity above gives the expected number of these edges, restricted to the disconnected event, as $(1-p)\varphi_R'(p)$.
The terminal symmetry $\vartheta$ exchanges $A^{+}$ and $A^{-}$ configuration by configuration and preserves the weight, so the counts taken with respect to $A^{+}$ and with respect to $A^{-}$ have equal expectations.
Open edges with both active endpoints lie in a single cluster, giving $U^{b}_{c}=2U^{u}_{c}$; closed edges with both active endpoints split into the two single-cluster counts and the pivotal bridges, giving $U^{b}_{b}=2U^{u}_{b}+(1-p)\varphi_R'(p)$; and edges with exactly one endpoint in $A^{+}$ split into edges with one active endpoint and pivotal bridges, each bridge being seen once from each side, giving $2U^{u}_{u}=U^{b}_{u}+2(1-p)\varphi_R'(p)$.

Finally, divide the last three row identities by $1-\varphi_R(p)$ and subtract half the $b$-row from the $u$-row.
This proves the left-eigenvector identity, whose critical value follows from $\varphi_R(p_c)=p_c$.
\end{proof}

Define
\[
 \mathbf Q(p):=
 \begin{pmatrix}
  1&0&1-p\\
  0&2&-p\\
  0&1&-p
 \end{pmatrix},
 \qquad
 \mathbf B_R(p):=
 \begin{pmatrix}
 M_{c,c}(p)&2M_{c,b}(p)+M_{c,u}(p)\\
 M_{u,c}(p)&2M_{u,b}(p)+M_{u,u}(p)
 \end{pmatrix}.
\]

\begin{theorem}\label{thm:three-state-block}
For every admissible rule graph,
\[
 \mathbf M_{\mathfrak C}(p_c)
 \begin{pmatrix}1-p_c\\-p_c\\-p_c\end{pmatrix}
 =\varphi_R'(p_c)
 \begin{pmatrix}1-p_c\\-p_c\\-p_c\end{pmatrix}.
\]
Moreover, $\det\mathbf Q(p_c)=-p_c\ne0$ and
\[
 \mathbf Q(p_c)^{-1}\mathbf M_{\mathfrak C}(p_c)\mathbf Q(p_c)
 =\begin{pmatrix}
   \mathbf B_R(p_c)&0\\
   0&\varphi_R'(p_c)
  \end{pmatrix}.
\]
The matrix $\mathbf B_R(p_c)$ has positive entries in $\mathbb Q(p_c)$ and two distinct real eigenvalues.
Its larger eigenvalue is $\rho(\mathbf M_{\mathfrak C})$, and
\[
 \rho(\mathbf M_{\mathfrak C})>\varphi_R'(p_c),
 \qquad
 \rho(\mathbf M_{\mathfrak C})>\left|\operatorname{tr}\mathbf B_R(p_c)-\rho(\mathbf M_{\mathfrak C})\right|.
\]
In particular the full three-state matrix is diagonalizable over $\mathbb R$.
\end{theorem}

\begin{proof}
At the fixed point, the first two identities of Lemma~\ref{lem:exact-row-identities} give
\[
 (1-p_c)M_{c,c}(p_c)
 -p_c\bigl(M_{c,b}(p_c)+M_{c,u}(p_c)\bigr)
 =(1-p_c)\varphi_R'(p_c),
\]
\[
 (1-p_c)M_{u,c}(p_c)
 -p_c\bigl(M_{u,b}(p_c)+M_{u,u}(p_c)\bigr)
 =-p_c\varphi_R'(p_c).
\]
The remaining identity follows from the row relation in Lemma~\ref{lem:exact-row-identities} at $p=p_c$.
Thus the last column of $\mathbf Q(p_c)$ is a right eigenvector with eigenvalue $\varphi_R'(p_c)$.
The same row relation shows that the plane spanned by the first two columns is invariant and that its restriction matrix is $\mathbf B_R(p_c)$.
This proves the displayed block decomposition.
The entries of $\mathbf B_R(p_c)$ belong to $\mathbb Q(p_c)$ because the counting polynomials have integer coefficients and the conditional denominators are nonzero.

The configurations in Proposition~\ref{prop:live-matrix-primitive} show that $M_{c,c}(p_c)$, $M_{c,b}(p_c)$ and $M_{u,c}(p_c)$ are positive.
The all-closed configuration also gives $M_{u,u}(p_c)>0$.
Hence every entry of $\mathbf B_R(p_c)$ is positive.
Its characteristic discriminant is
\[
 \bigl(M_{c,c}(p_c)-2M_{u,b}(p_c)-M_{u,u}(p_c)\bigr)^2
 +4M_{u,c}(p_c)\bigl(2M_{c,b}(p_c)+M_{c,u}(p_c)\bigr)>0,
\]
so its eigenvalues are distinct and real.
Primitivity of the full matrix makes its Perron root simple and strictly dominant in modulus.
It cannot be the scalar eigenvalue $\varphi_R'(p_c)$: the corresponding right eigenvector displayed above has entries of both signs, whereas a Perron eigenvector of a primitive matrix has one strict sign.
The Perron root is therefore the larger eigenvalue of the positive two-dimensional block, and the stated inequalities follow.
The direct sum of a real diagonalizable two-dimensional block and a scalar is real diagonalizable, even when the scalar coincides with the smaller eigenvalue of the block.
\end{proof}

Let $\eta_R$ denote the smaller eigenvalue of $\mathbf B_R(p_c)$.
The block decomposition gives
\[
 \det(tI-\mathbf M_{\mathfrak C})
 =(t-\varphi_R'(p_c))
 \bigl(t^2-\operatorname{tr}(\mathbf B_R(p_c))t+\det\mathbf B_R(p_c)\bigr),
\]
whose quadratic factor has coefficients in $\mathbb Q(p_c)$.
We now use the two distinct eigenvalues of this factor to obtain an exact expression for the expected critical cluster mass.

\begin{corollary}\label{cor:SP-closed-form}\label{cor:mass-spectral-expansion}
For every $n\ge1$, the expected critical cluster mass has the exact expansion
\begin{align*}
 \mathbb E[|E(\mathfrak C_n)|]
 &=\frac{M_{c,c}(p_c)-\eta_R}{\rho(\mathbf M_{\mathfrak C})-\eta_R}
       \rho(\mathbf M_{\mathfrak C})^n
 +\frac{\rho(\mathbf M_{\mathfrak C})-M_{c,c}(p_c)}
       {\rho(\mathbf M_{\mathfrak C})-\eta_R}\,\eta_R^n.
\end{align*}
Both coefficients are positive.
\end{corollary}
\begin{proof}
Spectral interpolation for $\mathbf B_R(p_c)$, together with the block decomposition in Theorem~\ref{thm:three-state-block}, gives the displayed expression for the $(c,c)$ entry of $\mathbf M_{\mathfrak C}^{\,n}$.
Lemma~\ref{thm:mass-transfer-power} identifies this entry with the expected cluster mass.
Since the off-diagonal entries of $\mathbf B_R(p_c)$ are positive, its characteristic polynomial is negative at $M_{c,c}(p_c)$; this diagonal entry lies strictly between its two eigenvalues, so both coefficients are positive.
\end{proof}

The one-dimensional crossing-response block contributes zero to this particular mass observable.
The smaller eigenvalue $\eta_R$ governs the second term, while $|\eta_R|<\rho(\mathbf M_{\mathfrak C})$ makes it subdominant.
These formulas remain valid when $\eta_R=\varphi_R'(p_c)$: the two eigenvalues of the positive block are still distinct.

The remaining scalar block records sensitivity to the occupation probability.
We interpret its eigenvalue through the growth of pivotal counts.

\begin{definition}\label{def:pivotal-dimension}
For critical percolation on the hierarchical lattice $\Xi$, the pivotal dimension of the two-terminal cluster $\mathfrak C$ is defined, when the limit exists, by
\[
   \dim_{\mathrm P}(\mathfrak C)
   :=
   \displaystyle\lim_{n\to\infty}\frac{\log \mathbb E_{p_c}\bigl[N_{\mathrm{piv}}(\Xi^{n})\bigr]}{\log \ell^{\,n}},
\]
the growth rate of the expected number of pivotal bonds at criticality, measured against the terminal distance scale $\ell^{\,n}$ of $\Xi^{n}$.
This is also the growth rate of the expected number of pivotal bonds in the cluster conditioned to connect the two terminals.
Indeed, whether an edge is pivotal is independent of its own open state, and a pivotal edge belongs to a connecting configuration exactly when it is open; hence conditioning on connection multiplies its pivotal probability by $p_c/\varphi_{\Xi^n}(p_c)=1$.
\end{definition}

\begin{theorem}\label{thm:pivotal-dimension}
For every classical EIGS hierarchical lattice, the pivotal dimension exists and
\[
 \dim_{\mathrm P}(\mathfrak C)=\frac{\log\varphi_R'(p_c)}{\log\ell},
 \qquad 0<\dim_{\mathrm P}(\mathfrak C)<\dimmass(\mathfrak C).
\]
\end{theorem}
\begin{proof}
Russo's formula and Lemma~\ref{lem:renormalisation} give, for every $n\ge1$,
\[
 \mathbb E_{p_c}[N_{\mathrm{piv}}(\Xi^n)]
 =(\varphi_R^{\circ n})'(p_c)=[\varphi_R'(p_c)]^n,
\]
where the last identity follows from the chain rule at the fixed point.
Taking logarithms and dividing by $n\log\ell$ proves the formula.
The inequalities follow from $\varphi_R'(p_c)>1$, Theorem~\ref{thm:three-state-block} and the mass-dimension formula in Theorem~\ref{thm:critical-mass-dimension}.
\end{proof}

The pivotal response also admits an upper bound in terms of the number of rule edges.
\begin{lemma}\label{lem:pivotal-response-bound}
For every classical EIGS hierarchical lattice,
\[
 [\varphi_R'(p_c)]^2<m.
\]
\end{lemma}
\begin{proof}
Under critical Bernoulli percolation on $R$, let $X$ be the terminal-crossing indicator and $S$ the number of open edges.
Differentiating the Bernoulli expectation gives
\[
 \operatorname{Cov}(X,S)=p_c(1-p_c)\varphi_R'(p_c),\qquad
 \operatorname{Var}X=p_c(1-p_c),\qquad
 \operatorname{Var}S=mp_c(1-p_c).
\]
Cauchy--Schwarz gives $[\varphi_R'(p_c)]^2\le m$.
Equality would make the nonconstant binary variable $X$ an affine function of $S$, although $m\ge4$ and every value $0,\ldots,m$ of $S$ has positive probability.
The inequality is therefore strict.
\end{proof}

\begin{example}\label{ex:dhl-response}\label{ex:dhl-dimensions}
For the diamond cell in Example~\ref{ex:dhl-reliability}, enumeration of the sixteen open-edge configurations gives
\[
\mathbf M_{\mathfrak C}(p)=
\begin{pmatrix}
\dfrac{4(1+p-p^2)}{2-p^2}&\dfrac{4p(1-p)}{2-p^2}&\dfrac{4(1-p)^2}{2-p^2}\\[2ex]
\dfrac{4p}{1+p}&\dfrac{4p}{1+p}&\dfrac{4(1-p)}{1+p}\\[2ex]
\dfrac{2p}{1+p}&0&2
\end{pmatrix}.
\]
At $p_c=(\sqrt5-1)/2$, the two-dimensional block in Theorem~\ref{thm:three-state-block} becomes
\[
\mathbf B_R(p_c)=
\begin{pmatrix}
8p_c^2&4p_c^2\\
2p_c^2&2
\end{pmatrix}.
\]
Its trace is $14-4\sqrt5$ and its determinant is $4(\sqrt5-1)$.
The larger root of its characteristic polynomial is therefore
\[
\rho(\mathbf M_{\mathfrak C})=7-2\sqrt5+\sqrt{73-32\sqrt5}
\approx3.7303.
\]
The other eigenvalue of this block is
\[
 \eta_R=7-2\sqrt5-\sqrt{73-32\sqrt5}.
\]
The scalar eigenvalue is $\varphi_R'(p_c)=6-2\sqrt5$, and $\ell=2$.
Thus the two cluster dimensions are
\[
 \dimmass(\mathfrak C)
 =\frac{\log\bigl(7-2\sqrt5+\sqrt{73-32\sqrt5}\bigr)}{\log2}
 \approx1.8993,\qquad
 \dim_{\mathrm P}(\mathfrak C)
 =\frac{\log(6-2\sqrt5)}{\log2}\approx0.6115.
\]
The ambient edge-volume exponent is $\log4/\log2=2$.
\end{example}

\section{Why multiplicative dimensions cannot classify}\label{sec:universality}\label{sec:converse-rigidity}\label{sec:three-dimensions}

This section addresses the universality question raised by Hauser and Saxena in 1986~\cite{hauser1986geometrical}: can the ambient dimension of a hierarchical lattice determine its percolation critical-exponent universality class? We prove that even an arbitrary finite collection of substitution-multiplicative dimensions is insufficient for this classification. We first specify the critical exponents and universality class used here, and then examine how substitution order affects these exponents. We begin by recalling the three dimensions needed in this section.
First, the Hausdorff dimension of the ambient metric scaling limit, denoted by $\dimhaus(\Xi)$, is given by~\cite{neroli2024fractal}*{Theorem~2.1}:
\[
 \dimhaus(\Xi)=\frac{\log m}{\log\ell}.
\]
We call this the \emph{ambient dimension}.
Second, Theorem~\ref{thm:critical-mass-dimension} gives the expected mass dimension of the critical cluster:
\[
 \dimmass(\mathfrak C)=\frac{\log\rho(\mathbf M_{\mathfrak C})}{\log\ell}.
\]
Third, Theorem~\ref{thm:pivotal-dimension} gives its pivotal dimension:
\[
 \dim_{\mathrm P}(\mathfrak C)=\frac{\log\varphi_R'(p_c)}{\log\ell}.
\]
We refer to these as the \emph{three critical dimensions}.

In Subsection~\ref{sec:static-class}, we examine eight critical exponents and use the four that exist for every classical EIGS hierarchical lattice under our observation conventions to define the critical-exponent universality class. We then prove that two lattices belong to the same class if and only if their three critical dimensions agree.

In Subsection~\ref{sec:noncommuting}, we return to the Tie and Gem example and show that cyclic permutations preserve the three critical dimensions, whereas general rearrangements of substitutions can change the critical mass dimension. Substitution-multiplicative dimensions, however, cannot distinguish these rearrangements. Consequently, no finite family of such dimensions completely determines the critical-exponent universality class; in particular, ambient dimension alone cannot classify it. The obstruction comes from substitution order, rather than the number of multiplicative dimensions considered. This does not contradict classification by the three critical dimensions, since the critical mass growth multiplier is not generally multiplicative under substitution.

\subsection{Eight critical exponents and universality classes}\label{sec:static-class}

The critical exponents and observation conventions used to define universality classes may differ across the literature.
To specify the objects studied here, we start from the eight percolation critical exponents listed by Grimmett~\cite{grimmett2014criticality}*{Table~2.1}.
Table~\ref{tab:eight-exponents} records their status for classical EIGS hierarchical lattices under the observation conventions specified below.
Under these conventions, $\beta,\nu,\delta,\eta_{\mathrm{av}}$ exist for every classical rule and will define our critical-exponent universality class.
The definitions, results and proofs for $\gamma,\Delta,\alpha,\rho$ are collected in Appendix~\ref{app:other-critical-exponents}; they require attention to the side of the transition, the order of the moment, or the distinction between point probabilities and cumulative probabilities.

\begin{table}[H]
\centering\footnotesize
\setlength{\tabcolsep}{3pt}
\renewcommand{\arraystretch}{1.3}
\begin{tabularx}{\textwidth}{@{}>{\raggedright\arraybackslash}p{1.7cm}*{4}{>{\centering\arraybackslash}X}|*{4}{>{\centering\arraybackslash}X}@{}}
\toprule
Exponent & $\beta$ & $\nu$ & $\delta$ & $\eta_{\mathrm{av}}$ & $\gamma$ & $\Delta$ & $\alpha$ & $\rho$\\
\midrule
Observable & Infinite-cluster density & Crossing length & Cluster mass & Averaged connectivity & Finite-cluster size & Moment ratios & Cluster-number response & Cluster radius\\
Scope & All rules & All rules & All rules & All rules & Side dependent & Order dependent & Extra criteria & Point law may fail\\
Result & Power law & Power law & Point law & Cell average & Three regimes & High-order stability & Explicit sufficient criteria & Cumulative law for all rules\\
Class definition & Yes & Yes & Yes & Yes & --- & --- & --- & ---\\
\bottomrule
\end{tabularx}
\caption{Eight critical exponents.}
\label{tab:eight-exponents}
\end{table}

\subsubsection*{Annealed observables}
The lattice is not spatially homogeneous, so a uniformly sampled vertex and a distinguished terminal need not see the same cluster law.
We use volume sampling for bulk observables and specify the crossing and two-point observations separately.

\begin{definition}\label{def:physical-observables}\label{def:annealed-percolation}
Choose a vertex $U_n$ uniformly from $V(\Xi^n)$, independently of the percolation configuration, and let $(G,o)$ have the uniformly rooted local-limit law of $(\Xi^n,U_n)$.
Conditional on $(G,o)$, declare its edges independently open with probability $p$.
The resulting probability and expectation, denoted by $\mathbb P_p^{\mathrm{ann}}$ and $\mathbb E_p^{\mathrm{ann}}$, average over both the rooted graph and the open edges.
Write $C(o)$ for the open cluster of the root and define
\[
 n_s(p):=\frac1s\mathbb P_p^{\mathrm{ann}}(|C(o)|=s),\qquad
 \theta(p):=\mathbb P_p^{\mathrm{ann}}(|C(o)|=\infty)
 =1-\sum_{s\ge1}s n_s(p).
\]
The density $n_s(p)$ is also the limit of the expected number of size-$s$ clusters in $\Xi^n$, divided by $|V(\Xi^n)|$.
All infinite-volume bulk probabilities below refer to this law, with the volume limit taken before the critical or large-cluster limit.
\end{definition}

The uniformly rooted construction and its percolation interpretation are discussed in~\cite{alves2026percolation}.
In particular, $n_s$ counts clusters, whereas $s n_s$ is the mass distribution seen from a uniformly sampled vertex.
The law in Definition~\ref{def:annealed-percolation} is unconditioned; it differs from the terminal-connected cluster law used to define $\dimmass(\mathfrak C)$ in Definition~\ref{def:critical-mass-dimension}.
The relation between their growth rates is a result of the percolation analysis.

For $p<p_c$, the crossing correlation length is
\[
 \xi(p)^{-1}:=-\lim_{n\to\infty}\ell^{-n}
 \log\varphi_R^{\circ n}(p).
\]
For the two-point observable, let
\[
 \mathcal P_n:=\bigl\{(x,y)\in V(\Xi^n)^2:
 \ell^n/4\le\dist_{\Xi^n}(x,y)\le\ell^n/2\bigr\},\qquad
 G_n:=\frac1{|\mathcal P_n|}\sum_{(x,y)\in\mathcal P_n}
 \mathbb P_{p_c}^{\Xi^n}(x\leftrightarrow y).
\]
Distances are measured in the deterministic graph, and the connecting path in this definition lies in $\Xi^n$.
Thus $G_n$ averages over pairs separated on the scale of a cell.
The fixed fractions $1/4$ and $1/2$ may be replaced by any fixed $0<a<b<1$ without changing the exponent.
This averaging convention is part of the definition of $\eta_{\mathrm{av}}$; no equivalence with a fixed-root pointwise or an infinite-volume spherical-shell definition is assumed.

The proofs use the conditional boundary masses and the levels at which internal clusters appear.
Write $d_R:=\deg_R(\beta^+)=\deg_R(\beta^-)$.
For $\sigma\in\{c,b,u\}$, let $X_{\sigma,n}$ count the internal vertices in the active boundary cluster or union of clusters of a depth-$n$ cell, under its state-$\sigma$ conditional law at $p_c$; neither outer terminal is counted.
Let $\mathcal J_n$ be the set of open clusters in $\Xi^n$ which avoid its two terminals but contain at least one internal vertex of its top-level copy of $R$.
Expectations involving $\mathcal J_n$ are unconditioned.

\begin{lemma}\label{lem:conditional-mass-moments}
For every fixed integer $r\ge1$ and every live state $\sigma$,
\[
 2\le d_R<\rho(\mathbf M_{\mathfrak C})<m,\qquad
 \mathbb E X_{\sigma,n}\asymp\rho(\mathbf M_{\mathfrak C})^n,
 \qquad
 \mathbb E X_{\sigma,n}^{r}\le C_r\rho(\mathbf M_{\mathfrak C})^{rn}
 \quad(n\ge1).
\]
Moreover, for every $p\in(0,1)$ and every $s\ge1$,
\begin{equation}\label{eq:birth-cluster-series}
 n_s(p)=\frac{m-1}{|V(R)|-2}
 \sum_{n\ge1}m^{-n}\mathbb E_p
 \#\{C\in\mathcal J_n:|C|=s\}.
\end{equation}
In particular, $\theta(p_c)=0$.
\end{lemma}
\begin{proof}
Condition on the crossing indicators of the child cells of a state-$\sigma$ cell.
If $J_\sigma$ is the number of active internal vertices of the top-level rule, the conditional construction in Theorem~\ref{thm:percolation-random-eigs} gives the exact recursion
\begin{equation}\label{eq:conditional-mass-recursion}
 X_{\sigma,n+1}=J_\sigma+
 \sum_{e\text{ live}}X_{\tau(e),n}^{(e)}.
\end{equation}
Here $0\le J_\sigma\le |V(R)|-2$, and, given this finite configuration, the summands are independent with the indicated state laws.
The reward $J_\sigma$ and the offspring types need not be independent.
Every parent has at least $d_R$ live children, since every edge incident to an active terminal is live.
The states $c,b$ have at least $2d_R$ children because the terminals are not adjacent.
Every parent has at most $m$ children, whereas an all-closed state-$u$ configuration has exactly $d_R<m$.
Primitivity and the strict Perron--Frobenius comparison therefore give $d_R<\rho(\mathbf M_{\mathfrak C})<m$.

Taking expectations in~\eqref{eq:conditional-mass-recursion} and summing the resulting matrix geometric series proves the first-moment estimate: the reward vector is nonnegative and nonzero, and both Perron vectors are strictly positive.
For the $r$th moment, expand the $r$th power conditionally on the finite configuration.
The terms with all $r$ factors in one child give $\mathbf M_{\mathfrak C}$ applied to the vector of $r$th moments.
Every remaining term contains only lower-order moments and bounded rewards.
Induction on $r$ bounds their sum by $C_r\rho(\mathbf M_{\mathfrak C})^{rn}$.
Iteration of the linear term gives the asserted bound, since $\rho(\mathbf M_{\mathfrak C})^r>\rho(\mathbf M_{\mathfrak C})$ for $r>1$.

Let $I_{n,s}$ count the size-$s$ clusters of $\Xi^n$ avoiding its terminals.
Such a cluster either lies inside one child cell and avoids its endpoints, or belongs to $\mathcal J_n$.
Consequently,
\[
 \mathbb E_p I_{n,s}=m\mathbb E_p I_{n-1,s}
 +\mathbb E_p\#\{C\in\mathcal J_n:|C|=s\}.
\]
Iterate this equality, divide by
$|V(\Xi^n)|=2+(|V(R)|-2)(m^n-1)/(m-1)$, and take the volume limit.
At most two additional clusters touch the outer terminals, so they do not affect the cluster density.
This proves~\eqref{eq:birth-cluster-series}; its nonnegative terms also allow summation against any nonnegative function of $s$.
At criticality the expected total mass touching the outer terminals is $O(\rho(\mathbf M_{\mathfrak C})^n)=o(m^n)$.
The same decomposition, now counting vertices, therefore gives $\sum_{s\ge1}s n_s(p_c)=1$, as required.
\end{proof}

The next lemma supplies a local estimate, rather than a cumulative tail estimate, for the mass exponent.

\begin{lemma}[Also see \cite{alves2026percolation}]\label{lem:conditional-mass-local-limit}
For each live state $\sigma$, there is a smooth probability density $w_\sigma$ on $\mathbb R$, vanishing on $(-\infty,0)$, such that
\begin{equation}\label{eq:conditional-mass-llt}
 \lim_{n\to\infty}\sup_{s\in\mathbb Z}
 \left|\rho(\mathbf M_{\mathfrak C})^n
 \mathbb P(X_{\sigma,n}=s)
 -w_\sigma\!\left(\frac{s}{\rho(\mathbf M_{\mathfrak C})^n}\right)\right|=0.
\end{equation}
The density $w_u$ is strictly positive on $(0,\infty)$.
For every fixed integer $r\ge1$,
\begin{equation}\label{eq:conditional-mass-point-bound}
 \mathbb P(X_{\sigma,n}=s)\le
 C_r\rho(\mathbf M_{\mathfrak C})^{-n}
 \left(1+\frac{s}{\rho(\mathbf M_{\mathfrak C})^n}\right)^{-r}
 \qquad(s\ge0).
\end{equation}
\end{lemma}
\begin{proof}
\emph{The branching limit.}
The live cells form a finite-type Galton--Watson process with bounded offspring and primitive mean matrix $\mathbf M_{\mathfrak C}$; see~\cite{athreya1972branching}*{Chapter~V} for the standard multitype branching framework.
Weighting the generation-$n$ population by a positive right Perron vector and dividing by $\rho(\mathbf M_{\mathfrak C})^n$ gives a martingale.
The expectations of its successive conditional variances are bounded by $C\rho(\mathbf M_{\mathfrak C})^{-n}$, so it converges in $L^2$.
The usual vector version follows by decomposing the mean matrix into its Perron projection and its complementary invariant subspace: the latter has strictly smaller spectral radius, and the centred offspring increments have second moment $O(\rho(\mathbf M_{\mathfrak C})^n)$.
Thus the normalised type populations converge in $L^2$ to a scalar limit times a fixed positive left Perron vector.

Sum the rewards in~\eqref{eq:conditional-mass-recursion} by generation.
Their conditional means are linear functions of the type populations, while their centred sums are martingale differences whose variances are $O(\rho(\mathbf M_{\mathfrak C})^n)$ in generation $n$.
After division by $\rho(\mathbf M_{\mathfrak C})^n$, the centred total tends to zero in $L^2$.
For the conditional means, the last finitely many generations converge by the population limit, and the earlier generations are bounded by a geometric tail in $L^2$.
It follows that $\rho(\mathbf M_{\mathfrak C})^{-n}X_{\sigma,n}$ has an $L^2$ limit $W_\sigma$ of positive mean.
These limits satisfy
\begin{equation}\label{eq:mass-smoothing}
 W_\sigma\overset{\mathrm d}=
 \frac1{\rho(\mathbf M_{\mathfrak C})}
 \sum_{e\text{ live}}W_{\tau(e)}^{(e)},
\end{equation}
with conditional independence on the right.
The variable $W_u$ is nonconstant: otherwise the all-closed state-$u$ configuration would force $\rho(\mathbf M_{\mathfrak C})=d_R$.
Primitivity then makes every $W_\sigma$ nonconstant, since each type can produce a type-$u$ descendant with positive probability.

\emph{Fourier inversion.}
If the characteristic function of $W_u$ had modulus one at a nonzero frequency $t$, the all-closed configuration in~\eqref{eq:mass-smoothing} would force the same at $t/\rho(\mathbf M_{\mathfrak C})$, and then at arbitrarily small nonzero frequencies.
This contradicts its finite, strictly positive variance.
For the other types, a finite offspring configuration leading to type $u$ gives the same conclusion.
The integer variables $X_{\sigma,1}$ have span one.
For $u,b$, opening no edges or just one edge incident to a terminal gives masses $0,1$.
For $c$, open a shortest terminal path and then add one adjacent vertex outside that path.
Such a vertex exists: a shortest path containing every vertex would leave no possible extra edge in a simple graph without shortening the path, contrary to terminal edge-connectivity at least two.
These two configurations again have consecutive masses.

Put $\psi_{\sigma,n}(t):=\mathbb E\exp(itX_{\sigma,n})$.
Conditional independence and at least $d_R$ live children give
\[
 \max_\sigma|\psi_{\sigma,n+1}(t)|
 \le\bigl(\max_\sigma|\psi_{\sigma,n}(t)|\bigr)^{d_R}.
\]
The $L^2$ limit gives uniform convergence of the normalised characteristic functions on compact frequency intervals.
On a fixed annulus $a\le |t|\le a\rho(\mathbf M_{\mathfrak C})$, their moduli are therefore bounded by a constant strictly below one for all sufficiently large depths.
For a normalised frequency with $|t|\le\pi\rho(\mathbf M_{\mathfrak C})^n$, choose the intermediate depth at which it enters this annulus, and iterate the preceding inequality through the remaining generations.
If this intermediate depth is bounded, the span-one property bounds the unnormalised characteristic functions strictly below one on the remaining compact subset of $[-\pi,\pi]\setminus\{0\}$.
Combining these two frequency ranges gives a constant $q\in(0,1)$ such that, for all sufficiently large $n$ and every $1\le |t|\le\pi\rho(\mathbf M_{\mathfrak C})^n$, there is an integer $k\ge0$ with
\[
 \left|\psi_{\sigma,n}\!\left(\frac{t}{\rho(\mathbf M_{\mathfrak C})^n}\right)\right|
 \le q^{d_R^k},
 \qquad |t|\le\pi\rho(\mathbf M_{\mathfrak C})^{k+1}.
\]
For every integer $r\ge1$, the sequence
$q^{d_R^k}\rho(\mathbf M_{\mathfrak C})^{r(k+1)}$ is bounded, since $d_R\ge2$.
Together with the bound $|\psi_{\sigma,n}|\le1$, this yields
\[
 \left|\psi_{\sigma,n}\!\left(\frac{t}{\rho(\mathbf M_{\mathfrak C})^n}\right)\right|
 \le C_r(1+|t|)^{-r},
 \qquad |t|\le\pi\rho(\mathbf M_{\mathfrak C})^n.
\]
Increasing $C_r$ covers the finitely many smaller depths as well.
Choosing $r$ arbitrarily large gives an integrable bound after multiplication by any fixed power of $|t|$.
Fourier inversion and dominated convergence prove smoothness of the limiting densities and~\eqref{eq:conditional-mass-llt}, uniformly in the lattice point.
They also give $\sup_s\mathbb P(X_{\sigma,n}=s)\le C\rho(\mathbf M_{\mathfrak C})^{-n}$.

\emph{Positivity and the pointwise bound.}
We first show that every $x>0$ belongs to the support of $W_u$.
In $\Xi^n$, delete the inactive terminal; the remaining graph is connected by canonicality.
Starting at the active terminal, grow a connected vertex set along a spanning tree of this graph, one vertex at a time.
Open its tree edges and close all other edges.
Each resulting state-$u$ configuration has positive probability and specifies the live types at depth $n$ in~\eqref{eq:mass-smoothing}.
The conditional mean of $W_u$ initially has order $(d_R/\rho(\mathbf M_{\mathfrak C}))^n$ and finally has order $(m/\rho(\mathbf M_{\mathfrak C}))^n$.
The maximum degree of $\Xi^n$ is at most $C_Rd_R^n$; adding one vertex therefore changes that conditional mean by at most $C_R(d_R/\rho(\mathbf M_{\mathfrak C}))^n$.
At the first crossing of $x$, the mean tends to $x$.
There are then $O(\rho(\mathbf M_{\mathfrak C})^n)$ live descendants, because their mean masses are bounded below by a positive constant.
Conditional independence makes the variance at most $C\rho(\mathbf M_{\mathfrak C})^{-n}$.
Every neighbourhood of $x$ consequently has positive probability.

Since $w_u$ is continuous and has full support, its positivity set is open and dense in $(0,\infty)$.
For every $x>0$, this set and its reflection about $x/2$ intersect inside $(0,x)$, so $(w_u*w_u)(x)>0$.
All higher convolution powers are positive there as well.
The all-closed state-$u$ configuration in~\eqref{eq:mass-smoothing} implies
\[
 w_u(x)\ge c\rho(\mathbf M_{\mathfrak C})
 w_u^{*d_R}\bigl(\rho(\mathbf M_{\mathfrak C})x\bigr)>0.
\]
Finally, condition on one generation in~\eqref{eq:conditional-mass-recursion} and split its live children into two nonempty groups.
Their mass sums are independent, and each has maximum point probability at most $C\rho(\mathbf M_{\mathfrak C})^{-n}$.
If their sum plus the bounded reward equals $s$, one group has mass at least half of $s-J_\sigma$.
Apply the $r$th-moment bound of Lemma~\ref{lem:conditional-mass-moments} to that group and the maximum point-probability bound to the other.
For $s$ below a fixed multiple of $\rho(\mathbf M_{\mathfrak C})^n$, use the latter bound alone.
This proves~\eqref{eq:conditional-mass-point-bound}.
\end{proof}

\subsubsection*{Four exponents valid for every rule}
\begin{definition}\label{def:four-exponents}
The order-parameter, crossing-length, cluster-mass and averaged-connectivity exponents are defined, when the limits exist, by
\begin{align*}
 \beta&:=\lim_{p\downarrow p_c}\frac{\log\theta(p)}{\log(p-p_c)},&
 \nu&:=-\lim_{p\uparrow p_c}\frac{\log\xi(p)}{\log(p_c-p)},\\
 \frac1\delta&:=-1-\lim_{s\to\infty}
 \frac{\log\mathbb P_{p_c}^{\mathrm{ann}}(|C(o)|=s)}{\log s},&
 \eta_{\mathrm{av}}&:=2-\dimhaus(\Xi)
 -\lim_{n\to\infty}\frac{\log G_n}{n\log\ell}.
\end{align*}
\end{definition}
The normalisation of $\eta_{\mathrm{av}}$ corresponds to the conventional connectivity law
\[
 G_n=(\ell^n)^{2-\dimhaus(\Xi)-\eta_{\mathrm{av}}+o(1)}.
\]
Logarithmic definitions permit bounded multiplicative oscillations associated with discrete changes of scale; a limiting prefactor is not required.

\begin{theorem}[Also see \cite{alves2026percolation}]\label{thm:critical-exponents-dimensions}
For every classical EIGS hierarchical lattice, the exponents $\nu$ and $\beta$ in Definition~\ref{def:four-exponents} exist and satisfy
\[
 \nu=\frac1{\dim_{\mathrm P}(\mathfrak C)},
 \qquad
 \beta=\frac{\dimhaus(\Xi)-\dimmass(\mathfrak C)}{\dim_{\mathrm P}(\mathfrak C)}.
\]
\end{theorem}
\begin{proof}
Fix a sufficiently small $\varepsilon>0$ and let $N(p)$ be the first iteration at which $\varphi_R^{\circ N(p)}(p)$ leaves $(p_c-\varepsilon,p_c+\varepsilon)$.
Since $\varphi_R'(p_c)>1$, Taylor's formula on this interval gives
\[
 \log\frac{|\varphi_R(q)-p_c|}{|q-p_c|}
 =\log\varphi_R'(p_c)+O(|q-p_c|).
\]
The deviations from $p_c$ grow at least geometrically until exit, so their sum is bounded uniformly in $p$.
Summing the displayed identity yields
\begin{equation}\label{eq:critical-escape-time}
 N(p)=\frac{-\log|p-p_c|}{\log\varphi_R'(p_c)}+O(1).
\end{equation}

The shortest terminal paths have length $\ell$, so $\varphi_R(q)/q^\ell$ extends to a strictly positive continuous function on $[0,p_c]$.
Telescoping gives
\[
 -\ell^{-n}\log\varphi_R^{\circ n}(p)
 =-\log p-\sum_{j=0}^{n-1}\ell^{-j-1}
 \log\frac{\varphi_R(\varphi_R^{\circ j}(p))}
 {(\varphi_R^{\circ j}(p))^\ell}.
\]
The series converges uniformly on compact subsets of $(0,p_c)$, with a uniformly bounded correction to $-\log p$.
Its limit is positive for small $p$ and satisfies $\xi(\varphi_R(p))=\xi(p)/\ell$.
Every subcritical orbit eventually reaches the small-$p$ region, so the limit is positive throughout $(0,p_c)$.
The exit points $\varphi_R^{\circ N(p)}(p)$ lie in a fixed compact subinterval of $(0,p_c)$.
It follows that $\xi(p)\asymp\ell^{N(p)}$, proving the formula for $\nu$.

For $\beta$, we use the annealed edge-rooted order-parameter theorem of Alves, Baldasso, Moreira and Teixeira~\cite{alves2026percolation}*{Theorems~7.1 and~8.1}.
Their critical boundary-mass growth factor is $\rho(\mathbf M_{\mathfrak C})$: the expected number of edges incident to the terminal clusters is
$(p_c,1-p_c,0)\mathbf M_{\mathfrak C}^{n}(1,1,1)^{\mathsf T}$,
and the Perron--Frobenius theorem identifies its exponential rate.
Their edge-rooted infinite-cluster probability thus satisfies
\[
 \Theta_E(p)\asymp
 (p-p_c)^{\log(m/\rho(\mathbf M_{\mathfrak C}))/\log\varphi_R'(p_c)}.
\]
We verify that uniform vertex sampling has the same exponent.
Let $D_R$ be the largest degree of an internal vertex of $R$.
Every final-level edge of $\Xi^n$ has a newly introduced endpoint of degree at most $D_R$.
For a connected open cluster with $s$ vertices, its number of open edges is between $s-1$ and $D_Rs$.
For each edge, the event that an endpoint belongs to a cluster of size at least $s$ is increasing.
Harris' inequality~\cite{grimmett1989percolation} compares its probability with the probability that the edge is also open by factors between $p$ and $1$.
Sum over edges and vertices, take the volume limit first and then let $s$ tend to infinity.
Since $|V(\Xi^n)|/|E(\Xi^n)|$ tends to $(|V(R)|-2)/(m-1)$, this gives
\[
 \frac{|V(R)|-2}{m-1}\theta(p)
 \le\Theta_E(p)
 \le\frac{D_R(|V(R)|-2)}{p(m-1)}\theta(p).
\]
The comparison constants stay bounded above and away from zero near $p_c$.
Thus $\theta$ has the same logarithmic exponent as $\Theta_E$, which is the stated value of $\beta$.
\end{proof}

\begin{theorem}\label{thm:delta-eta-dimensions}
For every classical EIGS hierarchical lattice, the exponents $\delta$ and $\eta_{\mathrm{av}}$ in Definition~\ref{def:four-exponents} exist and satisfy
\[
 \delta=\frac{\dimmass(\mathfrak C)}{\dimhaus(\Xi)-\dimmass(\mathfrak C)},
 \qquad
 \eta_{\mathrm{av}}=2+\dimhaus(\Xi)-2\dimmass(\mathfrak C).
\]
\end{theorem}
\begin{proof}
Condition on the top-level configuration of a cluster in $\mathcal J_n$.
Its mass is a bounded integer plus the sum of a bounded number of independent conditional child masses.
There are at least two live children, because every internal vertex of a canonical rule has degree at least two.
The proof of~\eqref{eq:conditional-mass-point-bound} therefore applies to each such cluster.
Since $|\mathcal J_n|\le |V(R)|-2$, equation~\eqref{eq:birth-cluster-series} gives, for any fixed integer $r>1+\log m/\log\rho(\mathbf M_{\mathfrak C})$,
\[
 n_s(p_c)\le C_r\sum_{n\ge1}
 m^{-n}\rho(\mathbf M_{\mathfrak C})^{-n}
 \left(1+\frac{s}{\rho(\mathbf M_{\mathfrak C})^n}\right)^{-r}.
\]
Split this geometric sum where $\rho(\mathbf M_{\mathfrak C})^n$ first exceeds $s$.
This gives
\[
 n_s(p_c)\le Cs^{-1-\log m/\log\rho(\mathbf M_{\mathfrak C})}.
\]
For the reverse bound, close every top-level coarse edge and fix an internal vertex of $R$.
Its cluster has mass one plus the sum of $\deg_R(v)$ independent state-$u$ child masses.
This coarse event has probability independent of $n$ and strictly positive.
Lemma~\ref{lem:conditional-mass-local-limit} and convolution give a local limit density strictly positive on $(0,\infty)$.
Choose $n$ so that $s/\rho(\mathbf M_{\mathfrak C})^{n-1}$ lies in $[1,\rho(\mathbf M_{\mathfrak C})]$.
The corresponding term in~\eqref{eq:birth-cluster-series} is at least $cm^{-n}\rho(\mathbf M_{\mathfrak C})^{-n}$, uniformly for large $s$.
Consequently,
\begin{equation}\label{eq:annealed-mass-point-law}
 \mathbb P_{p_c}^{\mathrm{ann}}(|C(o)|=s)
 =s n_s(p_c)\asymp s^{-\log m/\log\rho(\mathbf M_{\mathfrak C})}.
\end{equation}
Taking logarithms proves the formula for $\delta$.

We next estimate $G_n$, allowing any fixed distance window $[a\ell^n,b\ell^n]$ with $0<a<b<1$.
The deterministic cell has diameter at most $C_R\ell^n$.
Indeed, its maximum distance to either specified terminal is at most that at depth $n-1$ plus $\operatorname{diam}(R)\ell^{n-1}$; summing this recurrence proves the claim.
Distances between coarse vertices are multiplied by $\ell$ at each substitution.
A cell is also isometrically embedded: any excursion outside it between its two terminals is no shorter than an internal terminal geodesic and can be replaced by one.

Choose a fixed $k$ sufficiently large that depth-$(n-k)$ cells have diameter less than $a\ell^n$.
A connected pair in the distance window must lie in different such cells, and each vertex must connect inside its own cell to a boundary vertex.
Coarse vertices themselves may be assigned to any incident cell.
The number of connected pairs is therefore at most the square of the sum of the boundary-cluster masses of these $m^k$ cells, including their endpoints.
Lemma~\ref{lem:conditional-mass-moments} and Cauchy--Schwarz bound its expectation by $C_k\rho(\mathbf M_{\mathfrak C})^{2n}$.

For a lower bound, increase $k$ if necessary and choose two coarse edges near two points of a terminal geodesic separated by $(a+b)\ell^n/2$.
The diameter bound ensures that every pair of internal vertices in their depth-$(n-k)$ cells lies in the prescribed distance window.
This also proves $|\mathcal P_n|\asymp m^{2n}$.
Require all $m^k$ coarse cells to connect; this has probability $p_c^{m^k}>0$.
The two chosen conditional connected clusters are independent, are joined through the coarse graph, and each has expected internal mass at least $c\rho(\mathbf M_{\mathfrak C})^{n-k}$.
Their cross pairs give the reverse numerator bound.
Thus
\[
 G_n\asymp\left(\frac{\rho(\mathbf M_{\mathfrak C})}{m}\right)^{2n},
\]
which gives the stated $\eta_{\mathrm{av}}$ and proves independence of the fixed window.
\end{proof}

\begin{example}\label{ex:dhl-four-exponents}
For the diamond data in Example~\ref{ex:dhl-dimensions}, the formulas give
\[
 \beta=\frac{\log(4/\rho(\mathbf M_{\mathfrak C}))}{\log(6-2\sqrt5)}\approx0.1647,
 \qquad
 \nu=\frac{\log2}{\log(6-2\sqrt5)}\approx1.6353,
\]
\[
 \delta=\frac{\log\rho(\mathbf M_{\mathfrak C})}{\log(4/\rho(\mathbf M_{\mathfrak C}))}\approx18.8585,
 \qquad
 \eta_{\mathrm{av}}=4-\frac{2\log\rho(\mathbf M_{\mathfrak C})}{\log2}\approx0.2014.
\]
\end{example}

\subsubsection*{Universality classes}
The preceding discussion gives the four exponents $\beta,\nu,\delta,\eta_{\mathrm{av}}$ a common interpretation throughout the model class, under the observation conventions specified above.
The other four, treated in Appendix~\ref{app:other-critical-exponents}, require additional qualifications concerning existence, the side of the transition or the observable being measured.
We therefore use the former four to define the universality class; this choice does not assert that all eight existence conditions are determined by three dimensions.

\begin{definition}\label{def:general-class}
Two classical EIGS hierarchical lattices belong to the same \emph{critical-exponent universality class} if their four exponents $(\beta,\nu,\delta,\eta_{\mathrm{av}})$ agree.
\end{definition}

\begin{theorem}\label{thm:exponent-class-dimensions}\label{def:static-class}
Two classical EIGS hierarchical lattices belong to the same critical-exponent universality class if and only if their ambient, critical mass and pivotal dimensions agree.
\end{theorem}
\begin{proof}
The forward formulas are given by Theorems~\ref{thm:critical-exponents-dimensions} and~\ref{thm:delta-eta-dimensions}.
Conversely,
\[
 \dim_{\mathrm P}(\mathfrak C)=\frac1\nu,\qquad
 \dimmass(\mathfrak C)=\frac{\beta\delta}{\nu},\qquad
 \dimhaus(\Xi)=\frac{\beta(\delta+1)}{\nu}.
\]
Thus equality of the four exponents recovers all three dimensions.
\end{proof}

\subsection{The noncommutative obstruction}\label{sec:noncommuting}\label{sec:forward}

This subsection contains one of the central arguments of the paper: following the universality question of Hauser and Saxena, we ask whether substitution-multiplicative dimensions suffice to determine the critical-exponent universality class.

To answer this question, we return to the Tie and Gem rules from the introduction and examine why their distinct graphs belong to the same universality class.
Their agreement comes from reversing two edge substitutions, as Figure~\ref{fig:terminal-symmetric-rules-detail} shows.
The question is how much freedom we have to rearrange substitutions without changing the class.

\begin{example}\label{ex:tie-gem-similarity}
Let $\twoedgepath$ be the two-edge path and let $\triangle$ be a triangle with two designated terminals.
Write $A\circ B$ for the rule obtained by replacing every edge of $B$ by a copy of $A$; compositions are read from right to left.
The tie consists of two triangles in series, while the gem consists of two parallel paths of lengths two and four:
\[
 R_{\mathrm{tie}}=\triangle\circ \twoedgepath,\qquad
 R_{\mathrm{gem}}=\twoedgepath\circ\triangle.
\]
In the figure, $S_1$ replaces an edge by $\twoedgepath$ and $S_2$ replaces an edge by $\triangle$.

\begin{figure}[tbp]
\centering
\resizebox{0.96\textwidth}{!}{%
\begin{tikzpicture}[
    x=1cm,y=1cm,
    vertex/.style={circle,fill=black,inner sep=1.55pt},
    terminal/.style={circle,fill=black,inner sep=2.3pt},
    factorvertex/.style={circle,fill=red!45!black,inner sep=1.45pt},
    factorline/.style={draw=red!45!black,line width=0.75pt,line cap=round},
    line/.style={line width=0.7pt,line cap=round},
    repl/.style={-{Stealth[length=2.2mm,width=1.8mm]},line width=0.7pt},
    ilabel/.style={font=\scriptsize,inner sep=1pt},
    rulelabel/.style={font=\normalsize}
]
\begin{scope}[xshift=0cm,yshift=0cm]
  \node[rulelabel,anchor=east] at (-0.10,0) {Tie:};
  \draw[line] (0.20,0) -- (0.92,0);
  \node[vertex] at (0.20,0) {};
  \node[vertex] at (0.92,0) {};
  \draw[repl] (1.24,0) -- (1.88,0);

  \coordinate (bp) at (2.45,0);
  \coordinate (mid) at (3.70,0);
  \coordinate (bm) at (4.95,0);
  \coordinate (up) at (3.08,0.72);
  \coordinate (down) at (4.32,-0.72);

  \draw[line] (bp)--(up)--(mid)--(bp);
  \draw[line] (mid)--(down)--(bm)--(mid);

  \foreach \p in {bp,mid,bm,up,down}
    \node[vertex] at (\p) {};

  \node[ilabel,anchor=north east] at (bp) {$\beta^+$};
  \node[ilabel,anchor=north west] at (bm) {$\beta^-$};
  \node[terminal] at (bp) {};
  \node[terminal] at (bm) {};

  \draw[factorline] (2.25,-2.15)--(2.70,-1.55)--(3.15,-2.15)--cycle;
  \foreach \x/\y in {2.25/-2.15,2.70/-1.55,3.15/-2.15}
    \node[factorvertex] at (\x,\y) {};
  \node[text=red!45!black] at (2.70,-2.48) {$S_2$};
  \node[text=red!45!black] at (3.65,-1.95) {$\circ$};
  \draw[factorline] (4.10,-1.95)--(5.30,-1.95);
  \foreach \x in {4.10,4.70,5.30}
    \node[factorvertex] at (\x,-1.95) {};
  \node[text=red!45!black] at (4.70,-2.48) {$S_1$};
  \draw[repl,draw=red!45!black]
    (5.58,-1.78) .. controls (5.72,-1.25) and (5.42,-0.87) .. (5.08,-0.67);

\end{scope}

\begin{scope}[xshift=7.00cm,yshift=0cm]
  \node[rulelabel,anchor=east] at (-0.10,0) {Gem:};
  \draw[line] (0.20,0) -- (0.92,0);
  \node[vertex] at (0.20,0) {};
  \node[vertex] at (0.92,0) {};
  \draw[repl] (1.24,0) -- (1.88,0);

  \coordinate (bp) at (2.35,0);
  \coordinate (bm) at (5.05,0);
  \coordinate (top) at (3.70,0.82);
  \coordinate (b1) at (3.02,-0.72);
  \coordinate (b2) at (3.70,-0.72);
  \coordinate (b3) at (4.38,-0.72);

  \draw[line] (bp)--(top)--(bm);
  \draw[line] (bp)--(b1)--(b2)--(b3)--(bm);

  \foreach \p in {bp,bm,top,b1,b2,b3}
    \node[vertex] at (\p) {};

  \node[ilabel,anchor=north east] at (bp) {$\beta^+$};
  \node[ilabel,anchor=north west] at (bm) {$\beta^-$};
  \node[terminal] at (bp) {};
  \node[terminal] at (bm) {};

  \draw[factorline] (2.25,-1.95)--(3.45,-1.95);
  \foreach \x in {2.25,2.85,3.45}
    \node[factorvertex] at (\x,-1.95) {};
  \node[text=red!45!black] at (2.85,-2.48) {$S_1$};
  \node[text=red!45!black] at (3.95,-1.95) {$\circ$};
  \draw[factorline] (4.45,-2.15)--(4.90,-1.55)--(5.35,-2.15)--cycle;
  \foreach \x/\y in {4.45/-2.15,4.90/-1.55,5.35/-2.15}
    \node[factorvertex] at (\x,\y) {};
  \node[text=red!45!black] at (4.90,-2.48) {$S_2$};
  \draw[repl,draw=red!45!black]
    (5.66,-1.60) .. controls (5.80,-1.18) and (5.55,-0.87) .. (5.23,-0.66);

\end{scope}
\end{tikzpicture}}

\caption{The substitution orders $S_2\circ S_1$ (tie) and $S_1\circ S_2$ (gem).}
\label{fig:terminal-symmetric-rules-detail}
\end{figure}
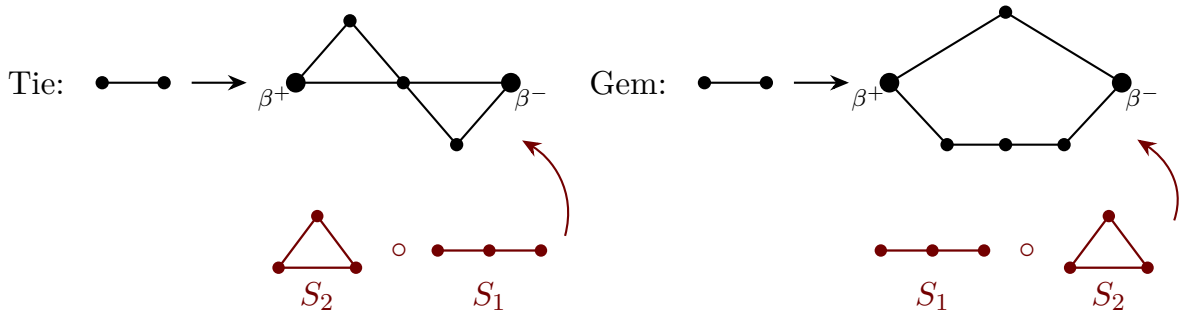

Both rules have six edges and terminal distance two.
Their unit-edge effective resistances also agree, with value $4/3$.
Their reliability functions are
\[
 \varphi_{\mathrm{tie}}(p)=(p+p^2-p^3)^2,\qquad
 \varphi_{\mathrm{gem}}(p)=p^2+p^4-p^6.
\]
The critical points differ, but satisfy
\[
 p_{c,\mathrm{tie}}=p_{c,\mathrm{gem}}^2,\qquad
 \varphi_{\triangle}(p_{c,\mathrm{tie}})=p_{c,\mathrm{gem}}.
\]
Table~\ref{tab:intro-tie-gem} records the resulting dimension values.
Their exact agreement is the two-factor case of Lemma~\ref{lem:transposition-invariance} below.

\end{example}

For terminal-symmetric two-terminal factors, we use $\mathbf M_A$ for the conditional mass matrix associated with $A$.
Factors may individually have terminal distance or minimum cut one, provided their interior conditional laws are defined; this includes $\twoedgepath$ and $\triangle$.
For these factors,
\begin{align}
 m_{A\circ B}&=m_A m_B,&\ell_{A\circ B}&=\ell_A\ell_B,\notag\\
 \varphi_{A\circ B}(p)&=\varphi_B(\varphi_A(p)),&
 \mathbf M_{A\circ B}(p)&=\mathbf M_B(\varphi_A(p))\mathbf M_A(p).\label{eq:substitution-responses}
\end{align}
Each coarse edge contributes $m_A$ fine edges.
A terminal path projects to a coarse terminal walk, and every traversal of a fine cell between its terminals requires at least $\ell_A$ edges.
A shortest coarse path refined by shortest fine paths attains the lower bound $\ell_A\ell_B$.
The fine-cell crossing events are independent with probability $\varphi_A(p)$, giving the reliability composition.
Conditioning on these events and the coarse live states gives the conditional laws of the fine cells; a black cell contributes no live descendants.
The tower property for the expected numbers of child states then gives the displayed matrix product.

\begin{lemma}\label{lem:transposition-invariance}
Let $k\ge2$ and let $A_1,\ldots,A_k$ be terminal-symmetric two-terminal rule graphs whose cyclic composites are admissible.
All cyclic rotations of $A_1\circ A_2\circ\cdots\circ A_k$ have the same three critical dimensions and similar critical three-state matrices over $\mathbb R$.
\end{lemma}
\begin{proof}
It suffices to compare $A\circ B$ and $B\circ A$, where either factor may itself be a composite.
Let $p_c$ be the critical point of $A\circ B$ and put $q_c:=\varphi_A(p_c)$.
Then $\varphi_B(q_c)=p_c$, so $q_c$ is the critical point of $B\circ A$.
The two multipliers are the same product $\varphi_B'(q_c)\varphi_A'(p_c)$.
Their edge counts and terminal distances are the same products by~\eqref{eq:substitution-responses}.
Their critical matrices are
\[
 \mathbf M_B(q_c)\mathbf M_A(p_c),\qquad
 \mathbf M_A(p_c)\mathbf M_B(q_c).
\]
The identity $\det(tI-XY)=\det(tI-YX)$ for square matrices gives equality of their characteristic polynomials, hence of their Perron roots.
This proves equality of the three critical dimensions.
Theorem~\ref{thm:three-state-block} shows that both full matrices are diagonalizable over $\mathbb R$, so their equal characteristic polynomials also give real similarity.
Apply this argument with $A=A_1$ and $B=A_2\circ\cdots\circ A_k$, then repeat to obtain every cyclic rotation.
\end{proof}

In particular, the two orders in the tie--gem example have exactly the same critical dimensions and hence belong to the same universality class.
For an explicit similarity, enumeration of the four configurations of $\twoedgepath$ gives
\[
 \mathbf M_{\twoedgepath}(p)=
 \begin{pmatrix}
 2&0&0\\
 \dfrac{2p}{1+p}&\dfrac{2p}{1+p}&\dfrac{2(1-p)}{1+p}\\[3pt]
 \dfrac{p}{1+p}&0&1
 \end{pmatrix},
 \qquad \det\mathbf M_{\twoedgepath}(p)=\frac{4p}{1+p}>0.
\]
The composition identities therefore yield
\[
 \mathbf M_{\mathrm{gem}}(p_{c,\mathrm{gem}})
 =\mathbf M_{\twoedgepath}(p_{c,\mathrm{gem}})^{-1}
 \mathbf M_{\mathrm{tie}}(p_{c,\mathrm{tie}})
 \mathbf M_{\twoedgepath}(p_{c,\mathrm{gem}}).
\]
Their common spectrum is approximately $5.7087,1.6239,1.4201$.

Cyclic invariance does not extend to arbitrary permutations.
The next construction preserves both the ambient and pivotal dimensions, while changing the critical mass dimension.

\begin{lemma}\label{lem:permutation-obstruction}
There exist admissible rules $A,B$ such that $B\circ A\circ B\circ A$ and $B\circ B\circ A\circ A$ have the same ambient and pivotal dimensions but different critical mass dimensions.
\end{lemma}
\begin{proof}
Let $A:=W$ be the Wheatstone rule shown on the left in Figure~\ref{fig:wheatstone-permutation-rules}, with terminals $s,t$.
Obtain $B$ by replacing the two highlighted opposite edges $su,vt$ by fresh copies of $W$, as shown on the right.
\begin{figure}[tbp]
\centering
\input{figure/wheatstone_permutation_rules.tex}
\caption{The rules $A=W$ and $B$.}
\label{fig:wheatstone-permutation-rules}
\end{figure}
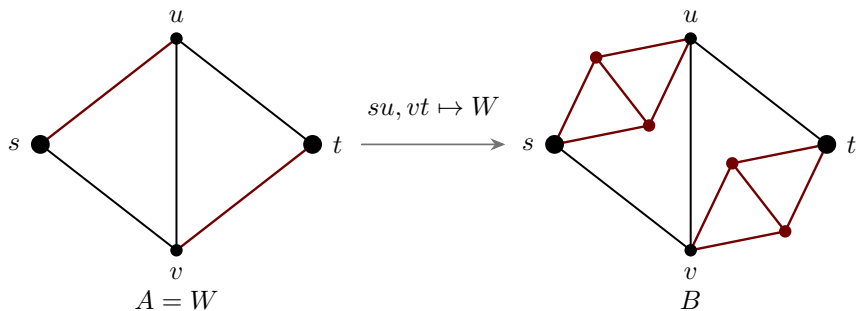
Both are simple, canonical and terminal symmetric, have terminal distance at least two and have two edge-disjoint terminal paths. Their critical points are $1/2$: the five independent child crossings in the latter rule are fair bits, just as for $W$. Their data are
\[
 (\ell_A,m_A,\varphi_A'(1/2))=(2,5,13/8),\qquad
 (\ell_B,m_B,\varphi_B'(1/2))=(3,13,67/32).
\]
For the last thermal value, the opposite outer pair containing $W$ has total derivative weight $3/4$, the other outer pair has weight $3/4$, and the centre has weight $1/8$. Thus it is $(3/4)(13/8)+3/4+1/8=67/32$.
Enumerating the thirty-two configurations of $A$ and the $8192$ configurations of $B$ under the conditional state rule gives
\[
 \mathbf B_A(1/2)=\frac1{16}\begin{pmatrix}53&48\\13&42\end{pmatrix},
 \qquad
 \mathbf B_B(1/2)=\frac1{256}\begin{pmatrix}1896&2292\\627&1380\end{pmatrix}.
\]

Compare the fourfold substitutions $B\circ A\circ B\circ A$ and $B\circ B\circ A\circ A$. Admissibility is preserved by substitution, and both have scale $2^2 3^2=36$. Both have $5^2 13^2=4225$ edges and critical point $1/2$.
Their critical derivatives are both $(\frac{13}{8}\frac{67}{32})^2$, so their ambient and pivotal dimensions agree.

At their common critical point, restrict~\eqref{eq:substitution-responses} to the common invariant plane from Theorem~\ref{thm:three-state-block}.
Their positive mass blocks are respectively $\mathbf B_A\mathbf B_B\mathbf B_A\mathbf B_B$ and $\mathbf B_A^2\mathbf B_B^2$. Their determinants agree. Exact multiplication gives
\[
 \operatorname{tr}(\mathbf B_A^2\mathbf B_B^2)
 -\operatorname{tr}(\mathbf B_A\mathbf B_B\mathbf B_A\mathbf B_B)
 =-\frac{1521}{4194304}\ne0.
\]
If their Perron roots agreed, equality of their nonzero determinants would force their other eigenvalues to agree as well, and hence their traces to agree, a contradiction. They therefore have different Perron roots at the same scale. Their critical mass dimensions therefore differ.
\end{proof}

For each of the two corresponding lattices $\Xi$,
\[
 \dimhaus(\Xi)=\frac{\log4225}{\log36},\qquad
 \dim_{\mathrm P}(\mathfrak C)=\frac{2\log(871/256)}{\log36}.
\]
Only their mass growth distinguishes the two classes.
Scalar multiplicative measurements cannot detect the rearrangement: their values depend on the factors and their multiplicities, but not on their order.
We use the term \emph{dimension-based critical-exponent universality class} when membership in a critical-exponent universality class is determined by a specified family of dimensions.
Here we ask whether substitution-multiplicative dimensions can provide such a classification.

\begin{definition}\label{def:multiplicative-dimension}
A positive functional $\Lambda$ on rule graphs is substitution-multiplicative if $\Lambda(A\circ B)=\Lambda(A)\Lambda(B)$ whenever the indicated rules and their composite are admissible. Its associated dimension is
\[
 \delta_\Lambda(R):=\frac{\log\Lambda(R)}{\log\ell_R}.
\]
\end{definition}
For example, the effective resistance of $A\circ B$, with unit resistances on its final edges, is the resistance of $B$ multiplied by that of $A$: each child cell can be electrically reduced to one resistor, and a uniform multiplication of the resistances multiplies the effective resistance by the same factor. Thus its logarithmic ratio belongs to this family, irrespective of whether it helps determine the three critical dimensions.

\begin{theorem}\label{thm:no-multiplicative-classification}
No finite family of substitution-multiplicative dimensions completely determines $\dimhaus(\Xi)$, $\dimmass(\mathfrak C)$ and $\dim_{\mathrm P}(\mathfrak C)$. In fact, there are two admissible rule graphs on which every such dimension agrees, while their live-mass dimensions differ.
\end{theorem}

\begin{proof}
Take the two composites from Lemma~\ref{lem:permutation-obstruction}.
For every positive substitution-multiplicative functional $\Lambda$, both values equal $\Lambda(A)^2\Lambda(B)^2$.
The terminal scales also agree, so every associated dimension agrees.
Their critical mass dimensions differ by the lemma, proving the assertion.
\end{proof}

The same pair defeats every substitution-multiplicative dimension simultaneously, so the obstruction also applies to an infinite family of these dimensions.
It does not contradict classification by the three critical dimensions: the mass Perron root is not generally multiplicative under substitution.
By Theorem~\ref{thm:exponent-class-dimensions}, the two rules belong to different critical-exponent universality classes despite agreement of every such multiplicative measurement.

\begin{corollary}\label{cor:ambient-does-not-classify}
The ambient dimension alone does not determine the critical-exponent universality class of a classical EIGS hierarchical lattice.
\end{corollary}
\begin{proof}
The two lattices constructed in Lemma~\ref{lem:permutation-obstruction} have the same ambient dimension but different critical mass dimensions, so Theorem~\ref{thm:exponent-class-dimensions} places them in different critical-exponent universality classes.
\end{proof}

This confirms Hauser and Saxena's 1986 conclusion~\cite{hauser1986geometrical} that ambient dimension alone cannot determine the percolation universality class on hierarchical lattices.

\section{The six exponentials theorem and common scales}\label{sec:scale}
Although the preceding section shows that scalar multiplicative dimensions cannot detect substitution order, we still know little about the structural relationships between lattices in the same critical-exponent universality class. Can the graph substitution structure provide a complete classification of critical-exponent universality classes on classical EIGS hierarchical lattices?

A natural direction comes from the physical notion of discrete scale invariance.
For a hierarchical lattice, each substitution has a fixed length-scaling factor, and iteration produces scales given by its integer powers; Sornette~\cite{sornette1998discrete} discusses this notion and its relationship to hierarchical systems.
We may therefore ask whether the discrete scales of two lattices in the same universality class must be compatible. More precisely, can finitely many levels on each side be grouped so that the resulting renormalisation length factors agree?

This naturally leads to the following conjecture.
\begin{conjecture}\label{conj:scale-commensurability}
If two classical EIGS hierarchical lattices belong to the same critical-exponent universality class, then their substitution length scales $\ell_1$ and $\ell_2$ are commensurate: there exist positive integers $a,b$ such that
\[
\ell_1^a=\ell_2^b.
\]
\end{conjecture}

The six exponentials theorem below connects this physical question about critical universality and discrete coarse-graining to a number-theoretic question about rational linear relations among the critical dimensions.
The growth multipliers associated with finite rule graphs are algebraic, while the dimensions are ratios of their logarithms to the logarithm of the length scale. This connection allows us to use transcendence theory to obtain sufficient criteria for scale commensurability.

\begin{theorem}[Six exponentials theorem, \cite{dasgupta2021transcendence}]\label{thm:six-exponentials}
Let $x_1,x_2$ be complex numbers linearly independent over $\mathbb Q$, and let $y_1,y_2,y_3$ also be complex numbers linearly independent over $\mathbb Q$.
Then at least one of the following six numbers is transcendental:
\[
 e^{x_i y_j},\qquad i\in\{1,2\},\quad j\in\{1,2,3\}.
\]
\end{theorem}

Before applying this theorem, we establish the required algebraicity.
\begin{lemma}\label{lem:algebraicity}
For every admissible rule graph, $p_c$, $\varphi_R'(p_c)$, the entries of $\mathbf M_{\mathfrak C}(p_c)$ and its eigenvalues are algebraic.
\end{lemma}
\begin{proof}
The critical point is a root of the nonzero polynomial $\varphi_R(p)-p$ with integer coefficients, and $\varphi_R'(p_c)\in\mathbb Q(p_c)$.
Since $p_c$ and $1-p_c$ are nonzero, the counting-polynomial formulas show that the matrix entries also belong to $\mathbb Q(p_c)$, so its eigenvalues are algebraic.
\end{proof}

By Theorem~\ref{thm:exponent-class-dimensions}, two lattices $\Xi_1,\Xi_2$ in the same critical-exponent universality class, with critical clusters $\mathfrak C_1,\mathfrak C_2$, satisfy $\dimhaus(\Xi_1)=\dimhaus(\Xi_2)$, $\dimmass(\mathfrak C_1)=\dimmass(\mathfrak C_2)$ and $\dim_{\mathrm P}(\mathfrak C_1)=\dim_{\mathrm P}(\mathfrak C_2)$.
Table~\ref{tab:six-exponentials-dimensions} in the introduction arranges these three dimensions against the logarithms of the two scales, displaying the connection between the six exponentials theorem and commensurability.

Multiplying each scale logarithm by the common critical dimensions and exponentiating gives the corresponding volume, critical mass and pivotal growth multipliers.
All these multipliers are algebraic.
Thus, if the two scales are incommensurate and the three common critical dimensions are linearly independent over $\mathbb Q$, the table exhibits precisely what the six exponentials theorem excludes: two linearly independent numbers and three linearly independent numbers whose six pairwise products all have algebraic exponentials.

In fact, the scales themselves are algebraic, so we may also include $1$ in the set of dimensions under consideration.
The same argument applies whenever three of $1$, $\dimhaus(\Xi_1)$, $\dimmass(\mathfrak C_1)$ and $\dim_{\mathrm P}(\mathfrak C_1)$ are linearly independent over $\mathbb Q$.
We first define scale commensurability and then state this observation as a theorem.

\begin{definition}\label{def:commensurate}
Two scales $\ell_1,\ell_2>1$ are \emph{commensurate} if there exist positive integers $a,b$ such that $\ell_1^a=\ell_2^b$; equivalently, $\log\ell_1/\log\ell_2\in\mathbb Q$.
\end{definition}

\begin{theorem}\label{thm:dimension-rank}\label{prop:arithmetic-alternatives}
If a lattice $\Xi$ in a critical-exponent universality class satisfies
\[
 \dim_{\mathbb Q}\operatorname{span}_{\mathbb Q}
 \{1,\dim_{\mathrm H}(\Xi),\dim_{\mathrm M}(\mathfrak C),\dim_{\mathrm P}(\mathfrak C)\}\ge3,
\]
then the scales of all classical EIGS hierarchical lattices in that class are pairwise commensurate.
In fact, there exists an integer $q\ge2$ that is not a proper power of another integer such that every rule graph in the class has scale equal to a positive integer power of $q$.
Consequently, if a class contains incommensurate scales, the rational span above has dimension at most two.
Furthermore, if at least one of the three common critical dimensions is irrational, the $\mathbb Q$-vector space spanned by the logarithms of all scales in the class has dimension at most two.
\end{theorem}
\begin{proof}
Suppose that two rule graphs in the class have incommensurate scales. Their scale logarithms are then linearly independent over $\mathbb Q$.
Choose three linearly independent numbers from $1$, $\dim_{\mathrm H}(\Xi)$, $\dim_{\mathrm M}(\mathfrak C)$ and $\dim_{\mathrm P}(\mathfrak C)$. Multiplying them by the two scale logarithms and exponentiating gives six numbers forming three columns of the following matrix:
\[
 \begin{pmatrix}
 \ell_1&m_1&\rho(\mathbf M_{\mathfrak C_1})&\varphi_{R_1}'(p_{c,1})\\
 \ell_2&m_2&\rho(\mathbf M_{\mathfrak C_2})&\varphi_{R_2}'(p_{c,2})
 \end{pmatrix}.
\]
By Lemma~\ref{lem:algebraicity}, these numbers are algebraic, contradicting Theorem~\ref{thm:six-exponentials}. Hence all scales are pairwise commensurate.

Fix one scale, write its prime factorisation as $\prod_p p^{a_p}$, and divide all nonzero exponents by their greatest common divisor.
Let $q$ be the integer formed from these reduced exponents; then $q$ is not a proper power of another integer.
For any other scale, commensurability implies that its prime-exponent vector is a positive rational multiple of this reduced vector.
Both vectors have integer coordinates, and the coordinates of the reduced vector have greatest common divisor one, so B\'ezout's identity implies that this multiple is an integer.
Thus the scale is a positive integer power of $q$.

The dimension bound in the incommensurate case is the contrapositive of the first assertion.
Finally, suppose that a common critical dimension $D$ is irrational and that the class contains three scale logarithms linearly independent over $\mathbb Q$. Multiplying these logarithms by the linearly independent numbers $1,D$ and exponentiating gives the three integer scales and their corresponding algebraic growth multipliers.
Interchanging the roles of the two groups in the six exponentials theorem again gives a contradiction.
\end{proof}

This result connects scale commensurability to rational linear relations among the three critical dimensions.
The next question is which lattices satisfy this criterion, and what possibilities remain when it fails.

\begin{example}\label{ex:wheatstone-arithmetic}
For the Wheatstone rule $W$ defined in Section~\ref{sec:noncommuting}, $\ell_W=2$, $m_W=5$ and $\varphi_W'(1/2)=13/8$. If integers $a,b,c$ satisfy
\[
 a+b\frac{\log5}{\log2}+c\frac{\log(13/8)}{\log2}=0,
\]
then $2^{a-3c}5^b13^c=1$. Valuations at $13$, $5$ and $2$, in that order, give $c=b=a=0$. Thus, for the lattice $\Xi$ generated by $W$, the numbers $1,\dim_{\mathrm H}(\Xi),\dim_{\mathrm P}(\mathfrak C)$ are independent, and every lattice in its critical-exponent universality class has scale $2^k$ for a positive integer $k$. This conclusion uses no calculation of its mass dimension.
\end{example}

The conclusion concerns the whole universality class, including rules outside the family from which a particular example was chosen.
To explain its geometric meaning, we record what happens when several substitution steps are grouped into one.

For a positive integer $a$, let $R^{\circ a}$ denote the rule obtained by $a$ successive substitutions of $R$.
\begin{proposition}\label{prop:iterated-response}\label{lem:aligned-class}
The iterated rule $R^{\circ a}$ has the same interior critical point as $R$, and
\[
 \ell_{R^{\circ a}}=\ell_R^a,\qquad m_{R^{\circ a}}=m_R^a,\qquad
 \varphi_{R^{\circ a}}'(p_c)=\varphi_R'(p_c)^a,\qquad
 \mathbf M_{R^{\circ a}}(p_c)=\mathbf M_{\mathfrak C}(p_c)^a.
\]
Consequently the ambient, critical mass and pivotal dimensions are unchanged.

Moreover, if two rules have scales satisfying $\ell_1^a=\ell_2^b$ for positive integers $a,b$, the two lattices have the same ambient, critical mass and pivotal dimensions exactly when
\[
 m_1^a=m_2^b,\qquad
 \rho(\mathbf M_{\mathfrak C_1})^{\,a}=\rho(\mathbf M_{\mathfrak C_2})^{\,b},\qquad
 \varphi_{R_1}'(p_{c,1})^a=\varphi_{R_2}'(p_{c,2})^b.
\]
\end{proposition}
\begin{proof}
The reliability is $\varphi_R^{\circ a}$ by the composition identities~\eqref{eq:substitution-responses}.
The strict inequalities on either side of $p_c$ in Corollary~\ref{cor:global-renormalisation-dichotomy} show that this iterate has no further fixed point in $(0,1)$.
The remaining formulas follow by repeated composition and the chain rule at $p_c$.
Both numerator and denominator of each logarithmic dimension are multiplied by $a$.

Multiply each equality of dimensions by the common positive number $a\log\ell_1=b\log\ell_2$ and exponentiate.
The same operations in reverse prove the converse.
\end{proof}

Thus, once the scales are commensurate, equality of the class becomes equality of three algebraic growth multipliers at a common scale.
This comparison does not require equality of effective resistances or of the remaining eigenvalue of the mass matrix, and it does not assert a conjugacy of the graph substitutions.
We now return to the rule graph to verify the independence required by Theorem~\ref{thm:dimension-rank}.
Two algebraic objects have already appeared: the fixed-point polynomial of the reliability map and the quadratic mass block from Theorem~\ref{thm:three-state-block}.
Irreducibility of these two objects will force commensurate scales throughout the universality class.

Removing the two trivial fixed points gives the polynomial
\begin{equation}\label{eq:nontrivial-fixed-point-polynomial}
 \frac{\varphi_R(p)-p}{p(1-p)}\in\mathbb Z[p].
\end{equation}
Here irreducibility will always refer to this complete quotient, not merely to the minimal polynomial of $p_c$.
Values at $p=0$ and $p=1$ refer to the polynomial obtained after cancellation.
We first record the consequences of canonicality and terminal symmetry that the argument needs.

\begin{lemma}\label{lem:fixed-point-polynomial-structure}
For every classical rule graph with $m$ edges,
\[
 \deg\left(\frac{\varphi_R(p)-p}{p(1-p)}\right)=m-2.
\]
This polynomial takes the values $-1$ and $1$ at $p=0$ and $p=1$, respectively, and its reduction modulo every prime is nonconstant.
\end{lemma}
\begin{proof}
The endpoint derivatives vanish because the terminal distance and the minimum terminal cut both have size at least two.
This gives the two endpoint values, which also show that the quotient polynomial is primitive.

For the degree, join the terminals by a new edge $e_*$.
Every original edge lies on a simple terminal path, hence on a cycle containing $e_*$.
The cycle matroid $M(R+e_*)$ is therefore connected, so its Crapo beta invariant is positive; see~\cite{oxley1982crapo}.
Pairing subsets with and without $e_*$ in the defining rank sum gives
\begin{align*}
 \beta\bigl(M(R+e_*)\bigr)
 &=(-1)^{|V(R)|-1}
   \sum_{A\subseteq E(R)}(-1)^{|A|}
   \bigl(\operatorname{rank}(A)-\operatorname{rank}(A\cup\{e_*\})\bigr)\\
 &=(-1)^{|V(R)|-1}
   \sum_{\substack{A\subseteq E(R)\\\beta^+\leftrightarrow\beta^-\text{ in }A}}(-1)^{|A|}
 =(-1)^{m+|V(R)|-1}[p^m]\varphi_R(p).
\end{align*}
Indeed, adding $e_*$ increases the rank exactly when the terminals are disconnected, and the alternating sum over all subsets of $E(R)$ is zero.
The last equality follows directly from the Bernoulli expansion of $\varphi_R$.
Thus $\deg\varphi_R=m$, so the quotient polynomial has degree $m-2$.

To prove the assertion modulo primes, consider the two-step terminal paths and the terminal cuts of size two.
If the number of two-step paths is odd, the terminal-exchanging involution fixes one of their middle vertices.
Deleting a terminal cut of size two leaves exactly two connected components: minimality of the cut forces each deleted edge to join the components containing the terminals, and connectedness excludes any further component.
An invariant cut of size two would make the involution exchange these two components, leaving no fixed vertex.
Consequently, when the number of two-step paths is odd, the cuts of size two occur in pairs and their number is even.

These two counts are the coefficients of $p^2$ in $\varphi_R(p)$ and in $1-\varphi_R(1-p)$, respectively.
They therefore exclude $\varphi_R(p)\equiv p^2\pmod2$.
The sum of the positive-degree coefficients of $\frac{\varphi_R(p)-p}{p(1-p)}$ is the difference between its endpoint values, namely two, so their greatest common divisor divides two.
If all these coefficients were even, the quotient would be congruent to one modulo two, giving $\varphi_R(p)\equiv p^2\pmod2$, a contradiction.
Their greatest common divisor is therefore one, which proves that the quotient remains nonconstant modulo every prime.
\end{proof}

The following criterion applies to the full class of classical EIGS hierarchical lattices.
Its two hypotheses can be checked on a single representative of a universality class.
By Theorem~\ref{thm:three-state-block}, the second hypothesis is equivalent to irreducibility of the characteristic polynomial of $\mathbf B_R(p_c)$ over $\mathbb Q(p_c)$.

\begin{theorem}\label{thm:irreducible-commensurability}
Suppose a classical EIGS hierarchical lattice satisfies
\[
 \frac{\varphi_R(p)-p}{p(1-p)}\text{ is irreducible over }\mathbb Q,
 \qquad
 \rho(\mathbf M_{\mathfrak C})\notin\mathbb Q(p_c).
\]
Then the scales of all classical EIGS hierarchical lattices in its critical-exponent universality class are pairwise commensurate.
\end{theorem}
\begin{proof}
We show that the three critical dimensions of this representative are linearly independent over $\mathbb Q$ and then apply Theorem~\ref{thm:dimension-rank}.

\textbf{The pivotal multiplier has no positive integral power in $\mathbb Z$.}
Suppose, to the contrary, that
\[
 [\varphi_R'(p_c)]^n=N
 \qquad\text{for some integers }n\ge1,\ N\ge2.
\]
Irreducibility makes $\frac{\varphi_R(p)-p}{p(1-p)}$ a rational multiple of the minimal polynomial of $p_c$.
Since this polynomial is primitive, Gauss's lemma gives the divisibility
\[
 \frac{\varphi_R(p)-p}{p(1-p)}\mid [\varphi_R'(p)]^n-N
 \qquad\text{in }\mathbb Z[p].
\]
Fix a prime $q\mid N$.
Lemma~\ref{lem:fixed-point-polynomial-structure} ensures that the reduction of $\frac{\varphi_R(p)-p}{p(1-p)}$ has a root $\alpha$ in a finite extension of $\mathbb F_q$.
At that root, the divisibility gives $\varphi_R'(\alpha)=0$ modulo $q$.
Differentiating the factorisation of $\varphi_R(p)-p$ into $p(1-p)$ and the quotient polynomial then yields
\[
 \alpha(1-\alpha)
 \left.\frac{d}{dp}\left(\frac{\varphi_R(p)-p}{p(1-p)}\right)\right|_{p=\alpha}=-1.
\]
Thus $\alpha$ is a simple root of the reduced polynomial.
Hensel's lemma lifts it to a root $\widetilde\alpha$ of $\frac{\varphi_R(p)-p}{p(1-p)}$ in the ring of integers of a finite unramified extension of $\mathbb Q_q$; see~\cite{milne2020ant}.
The lifted root satisfies $[\varphi_R'(\widetilde\alpha)]^n=N$.
The valuation on this unramified extension, normalised by $v_q(q)=1$, takes integer values on nonzero elements, and hence
\[
 n\,v_q\bigl(\varphi_R'(\widetilde\alpha)\bigr)=v_q(N),
 \qquad n\mid v_q(N).
\]
This holds for every prime divisor of $N$.
Therefore $N$ is a perfect $n$th power in $\mathbb Z$, and positivity implies $\varphi_R'(p_c)\in\mathbb Z$.
Only the existence of one such unramified lift at each prime was used.

It follows that $\frac{\varphi_R(p)-p}{p(1-p)}$ divides $\varphi_R'(p)-\varphi_R'(p_c)$ over $\mathbb Q$.
All $m-2$ roots of the quotient are consequently fixed points with the same multiplier $\varphi_R'(p_c)>1$.
Together with $0$ and $1$, whose multipliers are zero, these are all the fixed points, since $\deg\varphi_R=m$.
They are simple, and the residue of $1/(p-\varphi_R(p))$ at a fixed point is the reciprocal of one minus its multiplier.
The residue at infinity vanishes because $m\ge4$.
Summing the finite residues gives
\[
 2+\frac{m-2}{1-\varphi_R'(p_c)}=0,
 \qquad \varphi_R'(p_c)=\frac m2.
\]
On the other hand, Lemma~\ref{lem:pivotal-response-bound} gives
\[
 [\varphi_R'(p_c)]^2<m.
\]
There are two edge-disjoint terminal paths of length at least two, so $m\ge4$ and $m^2/4\ge m$.
This contradicts $\varphi_R'(p_c)=m/2$ and proves the claim.
In particular, $m$ and $\varphi_R'(p_c)$ are multiplicatively independent: a nontrivial relation between these two positive numbers would put a positive integral power of $\varphi_R'(p_c)$ in $\mathbb Z$.

\textbf{The mass multiplier supplies the third independent direction.}
Theorem~\ref{thm:three-state-block} places $\rho(\mathbf M_{\mathfrak C})$ in a quadratic extension of $\mathbb Q(p_c)$.
Under the second hypothesis, the nonidentity automorphism of this extension sends $\rho(\mathbf M_{\mathfrak C})$ to the other eigenvalue $\eta_R$ of $\mathbf B_R(p_c)$.
Irreducibility of the quadratic factor gives $\eta_R\ne0$, and strict Perron dominance gives $|\eta_R|<\rho(\mathbf M_{\mathfrak C})$.
Suppose that integers $a,b,c$ satisfy
\[
 m^a\rho(\mathbf M_{\mathfrak C})^b[\varphi_R'(p_c)]^c=1.
\]
The automorphism fixes $m$ and $\varphi_R'(p_c)$.
Applying it and dividing by the original equality yields
\[
 \left(\frac{\eta_R}{\rho(\mathbf M_{\mathfrak C})}\right)^b=1.
\]
The strict inequality of moduli forces $b=0$.
The multiplicative independence just proved then gives $a=c=0$.
Dividing the logarithms by $\log\ell$ shows that $\dimhaus(\Xi)$, $\dimmass(\mathfrak C)$ and $\dim_{\mathrm P}(\mathfrak C)$ are linearly independent over $\mathbb Q$.
Theorem~\ref{thm:dimension-rank} now applies to the entire universality class.
\end{proof}

The first part of the proof also gives a conclusion under the first hypothesis alone: if $\frac{\varphi_R(p)-p}{p(1-p)}$ is irreducible, then $\dim_{\mathrm P}(\mathfrak C)$ is transcendental.
Indeed, a rational value $a/b>0$ would give $[\varphi_R'(p_c)]^b=\ell^a\in\mathbb Z$, which was excluded above.
The Gelfond--Schneider theorem then excludes an algebraic irrational value, since both $\ell$ and $\varphi_R'(p_c)$ are algebraic; see~\cite{dasgupta2021transcendence}.

\begin{corollary}\label{cor:diamond-commensurability}\label{ex:diamond-irreducibility}
Every classical EIGS hierarchical lattice in the same critical-exponent universality class as the DHL has scale commensurate with $2$.
\end{corollary}
\begin{proof}
For the DHL, $\frac{\varphi_R(p)-p}{p(1-p)}=p^2+p-1$ is irreducible over $\mathbb Q$.
By Example~\ref{ex:dhl-response},
\[
 \mathbb Q(p_c)=\mathbb Q(\sqrt5),\qquad
 \rho(\mathbf M_{\mathfrak C})=7-2\sqrt5+\sqrt{73-32\sqrt5}.
\]
The norm $73^2-5\cdot32^2=209$ is not a rational square, so $73-32\sqrt5$ is not a square in $\mathbb Q(\sqrt5)$.
Hence $\rho(\mathbf M_{\mathfrak C})\notin\mathbb Q(p_c)$, and Theorem~\ref{thm:irreducible-commensurability} applies.
\end{proof}

Theorem~\ref{thm:irreducible-commensurability} obtains common scales from two irreducibility conditions on one rule graph, without restricting the other members of its class to any particular graph family.
It leaves open the cases where the complete fixed-point quotient is reducible or the mass quadratic splits over the critical field; failure of either hypothesis by itself does not imply incommensurability.
In particular, a rule with rational pivotal dimension must have reducible $\frac{\varphi_R(p)-p}{p(1-p)}$.

Although its hypotheses are strong, Theorem~\ref{thm:irreducible-commensurability} strengthens our confidence in the necessity of scale commensurability.
Are the arithmetic cases not covered by these conditions merely gaps in the proof, or can they be realised by actual hierarchical lattices?

\section{How commensurability fails}\label{sec:counterexample}

Although Section~\ref{sec:scale} strengthens our confidence in the necessity of scale commensurability, we construct two rule graphs with $19^{232}$ and $661^{232}$ edges, respectively, showing that commensurate scales are not necessary for membership in the same critical-exponent universality class.
This section examines the structural limits of universality classes through this construction and its variants.

Subsection~\ref{sec:incommensurate-family} gives a critical-exponent universality class containing infinitely many pairwise incommensurate scales.
As a complement to this construction, we also give an example in which all three critical dimensions are transcendental but rationally dependent.
Thus transcendence of all three dimensions alone does not suffice to meet the commensurability criterion based on the six exponentials theorem in Section~\ref{sec:scale}.

\subsection{Infinitely many scales in one class}\label{sec:incommensurate-family}\label{sec:transcendental-dependence}
We allow different finite graphs in the three slots of the Wheatstone construction from Lemma~\ref{lem:permutation-obstruction}. Each resulting graph is then iterated as one fixed rule. The auxiliary tree below specifies the interior of that finite rule; it does not vary the rule between levels of the hierarchical lattice. The explicit graphs are large, so their allocations are supplied as finite integer certificates.

\begin{definition}\label{def:wheatstone}
The two-terminal Wheatstone skeleton has vertices $s,t,u,v$, terminals $s,t$, and edges $su,ut,sv,vt,uv$. Given terminal-symmetric two-terminal graphs $A,B,C$, let $W(A,B,C)$ replace $su,vt$ by fresh copies of $A$, $ut,sv$ by fresh copies of $B$, and $uv$ by a fresh copy of $C$. All interiors are disjoint. Write $E$ for a single edge and $W:=W(E,E,E)$.
\end{definition}

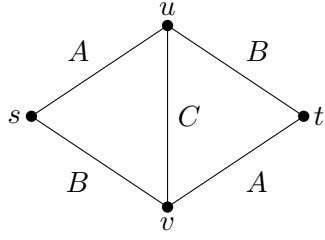
\begin{figure}[tbp]
\centering
\begin{tikzpicture}[scale=1.2,vertex/.style={circle,fill=black,inner sep=1.5pt}]
\coordinate (s) at (0,0);\coordinate (t) at (3,0);
\coordinate (u) at (1.5,1);\coordinate (v) at (1.5,-1);
\draw (s)--node[above left]{$A$}(u)--node[above right]{$B$}(t);
\draw (s)--node[below left]{$B$}(v)--node[below right]{$A$}(t);
\draw (u)--node[right]{$C$}(v);
\foreach \p in {s,t,u,v}{\node[vertex] at (\p) {};}
\node[left] at (s){$s$};\node[right] at (t){$t$};
\node[above] at (u){$u$};\node[below] at (v){$v$};
\end{tikzpicture}
\caption{The Wheatstone substitution $W(A,B,C)$.}
\label{fig:wheatstone}
\end{figure}

All the graphs used here are built from $E$ by finitely many applications of $W$.

\begin{lemma}\label{lem:wheatstone-responses}
Every such graph other than $E$ is admissible and has critical point $1/2$. The operation satisfies
\begin{align*}
 m_{W(A,B,C)}&=2m_A+2m_B+m_C,\\
 \ell_{W(A,B,C)}&=\min\{\ell_A+\ell_B,2\ell_A+\ell_C,2\ell_B+\ell_C\},\\
 \varphi_{W(A,B,C)}'(1/2)
 &=\tfrac34\varphi_A'(1/2)+\tfrac34\varphi_B'(1/2)+\tfrac18\varphi_C'(1/2).
\end{align*}
Here $\varphi_E(p):=p$.
\end{lemma}
\begin{proof}
The simple terminal paths of the Wheatstone skeleton are $sut$, $svt$, $suvt$ and $svut$. Refining their edges proves the distance formula, and counting the five slots proves the edge formula. Every edge of a child graph lies on a simple terminal path of that child, and each skeleton edge lies on a simple terminal path of the skeleton. Concatenation proves canonicality. The two paths through $u$ and through $v$ give two edge-disjoint terminal connections, and every terminal path has at least two edges. No loops or parallel edges arise when fresh interiors are used. The involution $(s\ t)(u\ v)$ exchanges the paired copies of $A$ and of $B$ and reverses the copy of $C$; child terminal symmetry completes this involution.

At occupation $1/2$, induction makes the five child crossings independent fair bits. Exactly sixteen of the thirty-two skeleton configurations cross. Hence the parent crosses with probability $1/2$, which is its unique interior fixed point. Among the sixteen assignments of the other four bits, a fixed outer edge is pivotal in six and the central edge in two. Differentiating with respect to the five child crossing probabilities, then applying the chain rule, gives the thermal formula.
\end{proof}

The corresponding conditional mass kernels, obtained from the same thirty-two configurations, are
\begin{equation}\label{eq:wheatstone-kernels}
 \mathbf K_3:=\frac1{16}
 \begin{pmatrix}22&8&2\\10&14&8\\5&1&16\end{pmatrix},
 \qquad
 \mathbf J_3:=\frac1{16}
 \begin{pmatrix}9&5&2\\6&4&4\\3&1&4\end{pmatrix}.
\end{equation}
The first kernel counts a pair of opposite outer slots, and the second counts the central slot. Conditional independence inside the child graphs gives
\begin{equation}\label{eq:wheatstone-mass}
 \mathbf M_{W(A,B,C)}(1/2)
 =\mathbf K_3\mathbf M_A(1/2)
  +\mathbf K_3\mathbf M_B(1/2)
  +\mathbf J_3\mathbf M_C(1/2).
\end{equation}
Indeed, conditioning the skeleton fixes which child graphs cross, and therefore selects precisely the conditional child laws defining their matrix rows. The copies within a paired slot are already counted in $\mathbf K_3$. There is no additional factor of two in this formula.

To build the seed rules, start with a complete ternary tree of Wheatstone operations of a fixed depth and replace selected leaf edges by further copies of $W$.
The choice of these leaves is the allocation.
The edge and pivotal responses depend on how many outer slots each decorated address traverses; the distance depends on decorations along all-outer addresses.
The mass response also records the order of the kernels along these addresses.
Since $\mathbf K_3\mathbf J_3\ne\mathbf J_3\mathbf K_3$, redistributing decorations can change the mass response while preserving the scalar responses.
The construction uses this freedom to meet the four targets
\begin{equation}\label{eq:seed-targets}
\begin{gathered}
 \ell_R=b^{100},\qquad m_R=b^{232},\qquad
 \varphi_R'(1/2)=b^{70},\\
 \mathbf M_{\mathfrak C}(1/2)\begin{pmatrix}55\\46\\23\end{pmatrix}
 =b^{219}\begin{pmatrix}55\\46\\23\end{pmatrix}.
\end{gathered}
\end{equation}
The last target uses the same positive vector for every seed.
It both identifies the Perron root and makes the mass multipliers multiplicative under compositions of these seeds.

The complete leaf encoding, response formulas, signed corrections and capacity bounds are retained with the exact integer verifier in the
\href{https://github.com/Nero-17/LEAN-Formalisation-IGS-Bernoulli-percolation-universality-class}{construction repository}.
Its README links each part of the construction to the corresponding code and Lean proofs.
The finite feasibility check is the following lemma.

\begin{lemma}\label{lem:integer-certificates}
For $(b,n)=(19,424),(661,936),(739,952)$, the depth-$n$ leaf allocations in the construction repository define finite rules satisfying~\eqref{eq:seed-targets}.
\end{lemma}
\begin{proof}[Computer-assisted proof]
The verifier reconstructs the mass kernels and pivotal weights by enumerating the thirty-two Wheatstone configurations.
It then checks the four targets by exact integer arithmetic, summing ordered matrix products through finite recursions.
Bounds on disjoint correction supports, together with explicit checks of the exceptional addresses, verify every allocation capacity and hence the existence of the encoded finite graphs.
The original three-state eigenvector is checked directly.
The repository contains the inputs, verifier, output hashes and Lean proofs of these certificate identities and their graph interpretation.
\end{proof}

\begin{theorem}\label{thm:incommensurate-class}
There exist admissible rule graphs $R_{19},R_{661},R_{739}$ with critical point $1/2$ such that, for $b\in\{19,661,739\}$,
\[
 (\ell_{R_b},m_{R_b},\rho(\mathbf M_{\mathfrak C_{R_b}}),\varphi_{R_b}'(1/2))
 =(b^{100},b^{232},b^{219},b^{70}).
\]
Their critical mass matrices share the positive eigenvector $(55,46,23)^T$, and each generated lattice $\Xi$ satisfies
\[
 (\dim_{\mathrm H}(\Xi),\dim_{\mathrm M}(\mathfrak C),\dim_{\mathrm P}(\mathfrak C))
 =\left(\frac{58}{25},\frac{219}{100},\frac7{10}\right),
\]
while their scale logarithms are linearly independent over $\mathbb Q$.
\end{theorem}
\begin{proof}
Take the rules supplied by Lemma~\ref{lem:integer-certificates}. They are admissible and have critical point $1/2$ by Lemma~\ref{lem:wheatstone-responses}. The first three targets in~\eqref{eq:seed-targets} give their distance, volume and pivotal multiplier. The positive eigenvector in the fourth target identifies $b^{219}$ as the Perron root. The three logarithmic ratios give the displayed dimensions.

The integers $19,661,739$ are greater than one and pairwise coprime. An integer relation between their logarithms would give a product of integer powers equal to one. Taking a valuation at a prime dividing each base in turn forces each exponent to vanish. Multiplication of the logarithms by $100$ preserves this independence.
\end{proof}

These lattices belong to the same critical-exponent universality class by Theorem~\ref{thm:exponent-class-dimensions}, so their incommensurate scales disprove Conjecture~\ref{conj:scale-commensurability}.

\begin{example}
For $R_{19}$, direct substitution gives
\[
 \frac{\log(19^{232})}{\log(19^{100})}=\frac{58}{25},\qquad
 \frac{\log(19^{219})}{\log(19^{100})}=\frac{219}{100},\qquad
 \frac{\log(19^{70})}{\log(19^{100})}=\frac7{10}.
\]
Theorems~\ref{thm:critical-exponents-dimensions} and~\ref{thm:delta-eta-dimensions} give the four defining exponents
\[
 \beta=\frac{13}{70},\qquad
 \nu=\frac{10}{7},\qquad
 \delta=\frac{219}{13},\qquad
 \eta_{\mathrm{av}}=-\frac3{50}.
\]
The same calculation applies to $661$ and $739$. Despite this agreement, no positive numbers of substitution steps on any two of these lattices give equal length scales.
\end{example}

The common positive eigenvector makes the mass Perron roots multiplicative under these particular compositions. It is not a general property of substitution, as the noncommuting example in Section~\ref{sec:noncommuting} shows.
\begin{theorem}\label{cor:infinite-scales}\label{thm:infinite-incommensurate-scales}
For every integer $k\ge0$, there is a classical EIGS hierarchical lattice $\Xi$ with scale $(19\cdot661^k)^{100}$ and critical dimensions
\[
 (\dim_{\mathrm H}(\Xi),\dim_{\mathrm M}(\mathfrak C),\dim_{\mathrm P}(\mathfrak C))
 =\left(\frac{58}{25},\frac{219}{100},\frac7{10}\right).
\]
These scales are pairwise incommensurate.
\end{theorem}
\begin{proof}
Compose $R_{19}$ with $k$ copies of $R_{661}$ by uniform edge substitution. Composition preserves admissibility: terminal paths and the involution refine through the child graphs, and two edge-disjoint paths in each factor give two in the composite. All factors have critical point $1/2$, so the chain rule multiplies their thermal multipliers. Edge counts and distances multiply by the composition identities~\eqref{eq:substitution-responses}.

Their critical mass matrices have the common positive eigenvector $(55,46,23)^T$. Their product therefore has that eigenvector with eigenvalue $19^{219}661^{219k}$, which is its Perron root. Thus the composite has the same three logarithmic dimensions and scale $(19\cdot661^k)^{100}$.

If two such scales with indices $k,l$ had equal positive integer powers, their $19$-adic valuations would make the two powers equal. Their $661$-adic valuations would then give $k=l$. The scales also grow without bound.
\end{proof}

As a further application, a variant of the same construction gives three transcendental critical dimensions that remain rationally proportional.

\begin{theorem}\label{thm:transcendental-dependent-dimensions}
There is a classical EIGS hierarchical lattice $\Xi$ whose three critical dimensions are all transcendental but span a one-dimensional vector space over $\mathbb Q$. More precisely, its admissible rule $R$ with $p_c=1/2$ satisfies
\[
 (\ell_R,m_R,\rho(\mathbf M_{\mathfrak C_R}),\varphi_R'(1/2))
 =(19^{100}+480,19^{232},19^{219},19^{70}),
\]
and hence
\[
 (\dim_{\mathrm H}(\Xi),\dim_{\mathrm M}(\mathfrak C),\dim_{\mathrm P}(\mathfrak C))
 =\frac{\log19}{\log(19^{100}+480)}(232,219,70).
\]
\end{theorem}
\begin{proof}
Use the modified depth-$424$ allocation in the construction repository.
It uses the same Wheatstone operation, but changes the length target to $\ell_R=19^{100}+480$ while retaining the other three targets in~\eqref{eq:seed-targets} with $b=19$.
The separate certificate verifies all four identities and every allocation capacity by the method of Lemma~\ref{lem:integer-certificates}.
It therefore gives an admissible rule with the displayed multipliers.

Since $\gcd(19,19^{100}+480)=1$, a rational equality
$\log19/\log(19^{100}+480)=a/b$, with positive integers $a,b$, would imply $19^b=(19^{100}+480)^a$, which is impossible by prime factorisation. This logarithmic ratio is therefore irrational. If it were algebraic, the Gelfond--Schneider theorem would make
$(19^{100}+480)^{\log19/\log(19^{100}+480)}=19$
transcendental, a contradiction. The ratio is transcendental, as are its three nonzero rational multiples. Their rational span has dimension one.
\end{proof}

In this example,
\begin{gather*}
 \dim_{\mathrm M}(\mathfrak C)=\frac{219}{232}\dim_{\mathrm H}(\Xi),\qquad
 \dim_{\mathrm P}(\mathfrak C)=\frac{35}{116}\dim_{\mathrm H}(\Xi),\\
 \dim_{\mathbb Q}\operatorname{span}_{\mathbb Q}
 \{1,\dim_{\mathrm H}(\Xi),\dim_{\mathrm M}(\mathfrak C),\dim_{\mathrm P}(\mathfrak C)\}=2.
\end{gather*}
Thus even three transcendental dimensions need not satisfy the rank condition in Theorem~\ref{thm:dimension-rank}. This example does not provide two incommensurate scales in a nonrational class: whether nonrationality of the critical dimensions forces commensurability remains a separate question.

\subsection{Discussion}\label{sec:discussion}

Physically, the three critical dimensions describe the growth of ambient volume, critical cluster mass and perturbations of the occupation probability relative to length scale.
Under our observation conventions, these three growth rates determine exactly the four critical exponents used to define the universality class.
They measure responses relative to changes in length, whereas the length factor associated with each substitution carries separate information.
Our construction shows that different models can have exactly the same normalised growth rates while their scales cannot be synchronised by any finite grouping of levels.
Equality of critical exponents therefore does not require compatible discrete renormalisation steps.

The final example also shows that the arithmetic nature of an individual critical dimension cannot replace independence among the growth rates: even when all three dimensions are transcendental, they may satisfy exact rational proportionality relations.
What matters in Section~\ref{sec:scale} is therefore the arithmetic relationship among these dimensions, rather than merely whether each is rational or transcendental.

This leads us to the following open problems:
\begin{enumerate}[leftmargin=*,itemsep=0.4em]
\item Can one construct two classical EIGS hierarchical lattices in the same critical-exponent universality class whose substitution scales are incommensurate and at least one of whose common critical dimensions is transcendental?
Any such example would disprove the four exponentials conjecture~\cite{dasgupta2021transcendence}.
Indeed, let $D$ be a common transcendental critical dimension and let $\ell_1,\ell_2$ be the two substitution scales.
Then $\log\ell_1,\log\ell_2$ are linearly independent over $\mathbb Q$, and $1,D$ are also linearly independent over $\mathbb Q$.
On the other hand, $\ell_i^D$ is the corresponding edge count, critical mass Perron root or pivotal growth multiplier of the $i$th rule. Hence Lemma~\ref{lem:algebraicity} implies that all four of the following numbers are algebraic:
\[
 \begin{pmatrix}
 e^{\log\ell_1}&e^{D\log\ell_1}\\
 e^{\log\ell_2}&e^{D\log\ell_2}
 \end{pmatrix}
 =\begin{pmatrix}
 \ell_1&\ell_1^D\\
 \ell_2&\ell_2^D
 \end{pmatrix}.
\]
The four exponentials conjecture asserts that, under these two linear independence assumptions, at least one of the exponential values must be transcendental.
Consequently, if the four exponentials conjecture holds, the scales in any class whose three critical dimensions are not all rational must be pairwise commensurate. Here an irrational dimension is necessarily transcendental by the Gelfond--Schneider theorem.
Theorem~\ref{thm:transcendental-dependent-dimensions} gives only a single lattice with transcendental critical dimensions, not such a pair of incommensurate models.
\item For the incommensurate models in the same class constructed here, do suitably normalised cluster mass distributions and averaged connectivity functions admit a common scaling description?
Beyond the four critical exponents, which finer critical quantities retain information about the underlying discrete scales?
\item Can structural conditions on graph substitutions or their induced random EIGS characterise the critical-exponent universality class and explain why different discrete scales can produce the same critical exponents?
Such a description must account both for the effect of substitution order on critical behaviour and for incommensurate scales within a single class.
\end{enumerate}

Although the structure of universality classes is more complex than we initially anticipated, we still believe that a beautiful and deep structural classification exists, one that can help us understand percolation universality anew.

\section*{Acknowledgement}
The author would like to thank Xuanyu Yu and Xiaole Zou for helpful conversations on percolation.
The author also thanks Ben Hambly for valuable discussions.

This work was supported by the Additional Funding Programme for Mathematical Sciences, delivered by EPSRC (EP/V521917/1) and the Heilbronn Institute for Mathematical Research,
and also by the EPSRC Centre for Doctoral Training in Mathematics of Random Systems: Analysis, Modelling and Simulation (EP/S023925/1).

\bibliographystyle{amsref}
\bibliography{references}

\appendix
\ctexset{
  appendix/name={Appendix\space,.},
  section/tocline={Appendix\space\thesection.\space #2}
}
\section{The remaining critical exponents}\label{app:other-critical-exponents}

We retain the annealed law and notation of Subsection~\ref{sec:static-class}.
This appendix gives the definitions and results for $\gamma,\Delta,\alpha,\rho$ summarised in Table~\ref{tab:eight-exponents}.

The finite-cluster moments and cluster-number density are
\[
 \chi_k^{\mathrm f}(p):=\mathbb E_p^{\mathrm{ann}}
 [|C(o)|^k\mathbf1_{\{|C(o)|<\infty\}}]
 =\sum_{s\ge1}s^{k+1}n_s(p),\qquad
 \kappa(p):=\sum_{s\ge1}n_s(p),
\]
where $k\ge1$, and the integrand defining $\chi_k^{\mathrm f}$ is assigned value zero on infinite clusters.
Write $\operatorname{rad}C(o):=\sup\{\dist_G(o,x):x\in C(o)\}$ for the radius in the ambient graph distance.

\begin{definition}\label{def:other-exponents}
Whenever the observables are finite in the relevant punctured neighbourhood and the limits exist, define
\begin{align*}
 \gamma_\pm&:=-\lim_{p\to p_c^\pm}
 \frac{\log\chi_1^{\mathrm f}(p)}{\log|p-p_c|},\\
 \Delta_{k,\pm}&:=-\lim_{p\to p_c^\pm}
 \frac{\log(\chi_{k+1}^{\mathrm f}(p)/\chi_k^{\mathrm f}(p))}{\log|p-p_c|},\\
 \alpha_\pm&:=-1-\lim_{p\to p_c^\pm}
 \frac{\log|\kappa'''(p)|}{\log|p-p_c|},\\
 \frac1\rho&:=-1-\lim_{r\to\infty}
 \frac{\log\mathbb P_{p_c}^{\mathrm{ann}}(\operatorname{rad}C(o)=r)}{\log r}.
\end{align*}
The last limit is along positive integers.
A common gap exponent $\Delta_\pm$ means that $\Delta_{k,\pm}$ has the same value for every $k\ge1$.
For $\alpha_\pm$ we use the full cluster-number density, with no analytic background subtracted, and require $\kappa'''$ to be nonzero sufficiently close to $p_c$ on the chosen side.
\end{definition}

\subsection*{Susceptibility and moment ratios}
The following estimate determines both susceptibility and the moment ratios.

\begin{proposition}\label{prop:annealed-moments}
For $k\ge1$, the moment $\chi_k^{\mathrm f}(p)$ is finite for every $p>p_c$.
For $0<p<p_c$, it is finite if and only if $d_R^{k+1}<m$.
On each side on which it is finite, as $p$ approaches $p_c$,
\begin{equation}\label{eq:near-critical-moment-sum}
 \chi_k^{\mathrm f}(p)\asymp
 \sum_{n=1}^{N(p)}
 \left(\frac{\rho(\mathbf M_{\mathfrak C})^{k+1}}m\right)^n,
\end{equation}
where $N(p)$ is the exit time in~\eqref{eq:critical-escape-time}.
If $\rho(\mathbf M_{\mathfrak C})^{k+1}<m$, the moment converges to its finite critical value on each such side.
\end{proposition}
\begin{proof}
Summing~\eqref{eq:birth-cluster-series} against $s^{k+1}$ gives the nonnegative identity
\begin{equation}\label{eq:birth-moment-series}
 \chi_k^{\mathrm f}(p)=\frac{m-1}{|V(R)|-2}
 \sum_{n\ge1}m^{-n}\mathbb E_p
 \sum_{C\in\mathcal J_n}|C|^{k+1}.
\end{equation}
In this proof, $r\ge2$ denotes the moment order of a cluster counted once, so that $r=k+1$.

\emph{Fixed parameters.}
For $p>p_c$, the probability $1-\varphi_R^{\circ n}(p)$ decreases faster than any exponential in $n$, because a failure to cross requires at least two failed child crossings near $p=1$.
If every top-level cell crosses, $\mathcal J_n$ is empty.
Since all cluster masses are at most $C_Rm^n$,
\[
 \mathbb E_p\sum_{C\in\mathcal J_n}|C|^r
 \le C_rm^{rn}\bigl(1-\varphi_R^{\circ(n-1)}(p)\bigr).
\]
This is summable after multiplication by $m^{-n}$.

For $p<p_c$, let $A_n$ count all edges of $\Xi^n$ incident to the open cluster of a specified terminal, including closed edges.
When every top-level cell fails to cross, only the $d_R$ cells incident to that terminal contribute.
On the complementary event, use $A_{n+1}\le m^{n+1}$ and a union bound.
Minkowski's inequality gives
\[
 \|A_{n+1}\|_r\le d_R\|A_n\|_r
 +m^{n+1}\bigl(m\varphi_R^{\circ n}(p)\bigr)^{1/r}.
\]
The last term decreases faster than any exponential, since $\ell\ge2$ and $\varphi_R^{\circ n}(p)$ tends to zero.
Iteration yields $\mathbb E_p A_n^r\le C_rd_R^{rn}$.
Every cluster in $\mathcal J_n$ contains a top-level internal vertex, and its mass is bounded by one plus the number of incident edges of that vertex's cluster.
Applying the same bound in the finitely many cells incident to those vertices gives
\[
 \mathbb E_p\sum_{C\in\mathcal J_n}|C|^r\le C_rd_R^{rn}.
\]
These estimates are uniform when $p$ ranges in a compact subinterval of $(0,p_c)$.

For a matching lower bound, fix an internal vertex $v$ of the top-level rule.
Its degree in $\Xi^n$ is $\deg_R(v)d_R^{n-1}$.
Let $D_{v,n}$ count its open incident edges, and let $E_n$ be the event that every top-level cell fails to cross.
The graph is simple, so $D_{v,n}$ counts distinct neighbours in the cluster of $v$.
Moreover,
\[
 \mathbb E_p[D_{v,n}\mathbf1_{E_n}]
 \ge\deg_R(v)d_R^{n-1}
 \bigl(p-m\varphi_R^{\circ(n-1)}(p)\bigr)
 \ge c_pd_R^n
\]
for all sufficiently large $n$.
On $E_n$ this cluster belongs to $\mathcal J_n$.
Jensen's inequality gives the reverse $r$th-moment bound $c_{p,r}d_R^{rn}$.
The series~\eqref{eq:birth-moment-series} is therefore finite exactly when $d_R^r<m$, including divergence in the equality case.

\emph{Uniform estimates up to the exit time.}
Write $p_j:=\varphi_R^{\circ j}(p)$.
Every nonzero entry of $\mathbf M_{\mathfrak C}(q)$ is strictly positive and analytic near $p_c$, and its zero pattern is fixed there.
Consequently its entries lie between the corresponding entries of $\mathbf M_{\mathfrak C}$ multiplied by $\exp(-C|q-p_c|)$ and by $\exp(C|q-p_c|)$.
As in the proof of~\eqref{eq:critical-escape-time},
$\sum_{j<N(p)}|p_j-p_c|\le C\varepsilon$.
Every product of consecutive matrices along this part of the orbit is thus bounded above and below, entrywise, by constant multiples of the corresponding power of $\mathbf M_{\mathfrak C}$.
The reward recursion~\eqref{eq:conditional-mass-recursion}, now with its level-dependent conditional laws, gives uniformly for $n\le N(p)$
\[
 \mathbb E_p X_{\sigma,n}\asymp\rho(\mathbf M_{\mathfrak C})^n,
 \qquad
 \mathbb E_p X_{\sigma,n}^r\le C_r\rho(\mathbf M_{\mathfrak C})^{rn}.
\]
Here $X_{\sigma,n}$ denotes the same conditional mass at parameter $p$.
For the higher moments, expand one generation as in Lemma~\ref{lem:conditional-mass-moments} and iterate the linear term using the matrix-product bound; induction bounds all mixed terms by the asserted lower-order moment products.
The bound also applies to any consecutive segment of the orbit.
Conditioning all top-level cells to fail, and using the mass attached to one internal vertex, gives a matching lower bound by Jensen's inequality.
The probability of this finite configuration stays bounded away from zero.
Hence
\begin{equation}\label{eq:pre-exit-birth-moment}
 \mathbb E_p\sum_{C\in\mathcal J_n}|C|^r
 \asymp\rho(\mathbf M_{\mathfrak C})^{rn}
 \qquad(1\le n\le N(p)).
\end{equation}

It remains to control all later birth levels.
Coarsen depth-$N(p)$ cells to edges with parameter $p_{N(p)}$, which lies in a fixed compact interval strictly on the chosen side of $p_c$.
For a coarse cluster $C$, let $A(C)$ be the number of coarse edges incident to it.
Its lifted mass is at most its coarse vertex count plus a sum of at most $A(C)$ conditional depth-$N(p)$ boundary masses.
Conditioning on the coarse configuration and using
$(\sum_{i=1}^a x_i)^r\le a^{r-1}\sum_{i=1}^a x_i^r$ bounds the expected $r$th power by
$C_r\rho(\mathbf M_{\mathfrak C})^{rN(p)}(1+A(C))^r$.
For coarse clusters born at level $j$, the preceding fixed-parameter arguments bound the expected sum of $(1+A(C))^r$ by
\[
 \begin{cases}
 C_rd_R^{rj},&p<p_c,\\
 C_rm^{rj}\bigl(1-\varphi_R^{\circ(j-1)}(p_{N(p)})\bigr),&p>p_c.
 \end{cases}
\]
The constants are uniform over the exit interval.
Their sums against $m^{-j}$ are finite in precisely the regimes already specified.
It follows that the part of~\eqref{eq:birth-moment-series} after $N(p)$ is at most
\[
 C_r\left(\frac{\rho(\mathbf M_{\mathfrak C})^r}{m}\right)^{N(p)}.
\]
Together with~\eqref{eq:pre-exit-birth-moment}, this proves~\eqref{eq:near-critical-moment-sum}.
When $\rho(\mathbf M_{\mathfrak C})^r<m$, the same estimates make the tail after any fixed birth level uniformly small as $p$ approaches $p_c$.
Each finite initial sum is continuous in $p$.
The moment therefore converges to its finite, positive critical value.
\end{proof}

\paragraph{Susceptibility.}
Proposition~\ref{prop:annealed-moments} with $k=1$ shows that $\chi_1^{\mathrm f}(p)$ is finite for every fixed $p>p_c$.
For $0<p<p_c$, it is finite if and only if $d_R^2<m$.
Evaluating the geometric sum in~\eqref{eq:near-critical-moment-sum} and using~\eqref{eq:critical-escape-time} gives three regimes on either side where it is finite:
\[
 \begin{array}{c|c}
 2\dimmass(\mathfrak C)>\dimhaus(\Xi)
 &\chi_1^{\mathrm f}(p)\asymp
 |p-p_c|^{-(2\dimmass(\mathfrak C)-\dimhaus(\Xi))/\dim_{\mathrm P}(\mathfrak C)}\\[2pt]
 2\dimmass(\mathfrak C)=\dimhaus(\Xi)
 &\chi_1^{\mathrm f}(p)\asymp\log(1/|p-p_c|)\\[2pt]
 2\dimmass(\mathfrak C)<\dimhaus(\Xi)
 &\displaystyle\lim_{p\to p_c^\pm}\chi_1^{\mathrm f}(p)=\chi_1^{\mathrm f}(p_c)\in(0,\infty).
 \end{array}
\]
In particular, $\gamma_+$ always has the logarithmic value
$\max\{2\dimmass(\mathfrak C)-\dimhaus(\Xi),0\}/\dim_{\mathrm P}(\mathfrak C)$.
The value zero includes both a logarithmic divergence and a finite limit; it does not assert a positive power-law singularity.
The ordinary diamond has $d_R^2=m=4$, so its annealed subcritical susceptibility is infinite throughout $(0,p_c)$, although $\gamma_+\approx2.9412$.

\paragraph{Moment ratios.}
More generally, $\chi_k^{\mathrm f}(p)$ is finite for every $p>p_c$, whereas for $0<p<p_c$ it is finite exactly when $d_R^{k+1}<m$.
Taking logarithms in~\eqref{eq:near-critical-moment-sum} and subtracting successive moment orders gives, with $(x)_+:=\max\{x,0\}$,
\[
 \Delta_{k,+}=
 \frac{((k+2)\dimmass(\mathfrak C)-\dimhaus(\Xi))_+
 -((k+1)\dimmass(\mathfrak C)-\dimhaus(\Xi))_+}
 {\dim_{\mathrm P}(\mathfrak C)}.
\]
All sufficiently high orders have the common value $\dimmass(\mathfrak C)/\dim_{\mathrm P}(\mathfrak C)$.
The same value holds for every $k\ge1$ precisely when $2\dimmass(\mathfrak C)\ge\dimhaus(\Xi)$.
On the subcritical side, the $k$th ratio requires $d_R^{k+2}<m$, and then has the same formula.
Since $d_R\ge2$, the original ratios cannot all be finite on that side.

\subsection*{Cluster-number response}
Let $h_R(p)$ be the expected number of clusters in the rule graph touching neither terminal.
This is a polynomial obtained by summing over the finitely many open-edge configurations of $R$.

\begin{proposition}\label{prop:cluster-number-response}
The cluster-number density is the unique bounded solution on $[0,1]$ of
\begin{equation}\label{eq:cluster-number-functional}
 m\kappa(p)-\kappa(\varphi_R(p))
 =\frac{m-1}{|V(R)|-2}h_R(p).
\end{equation}
It is analytic away from $p_c$, and is $C^j$ at $p_c$ whenever $j$ is a nonnegative integer with $\varphi_R'(p_c)^j<m$.
In particular it is always $C^2$.
If $j\ge3$, $\varphi_R'(p_c)^j<m$, the derivatives of orders $3,\ldots,j-1$ vanish at $p_c$, and $\kappa^{(j)}(p_c)\ne0$, then
\[
 \alpha_- =\alpha_+=2-j.
\]
\end{proposition}
\begin{proof}
Let $K_n(p)$ count all clusters in $\Xi^n$.
Each last-level cell contributes its internal clusters; the remaining clusters correspond exactly to those of the coarse graph with parameter $\varphi_R(p)$.
Thus $\mathbb E K_n(p)=m^{n-1}h_R(p)+\mathbb E K_{n-1}(\varphi_R(p))$.
Iteration and volume normalisation yield the uniformly convergent series
\begin{equation}\label{eq:cluster-number-series}
 \kappa(p)=\frac{m-1}{|V(R)|-2}
 \sum_{n\ge0}m^{-n-1}h_R(\varphi_R^{\circ n}(p)),
\end{equation}
and hence~\eqref{eq:cluster-number-functional}.
The difference of two bounded solutions is at most $Cm^{-n}$ after $n$ iterations, proving uniqueness.
For a real point other than $p_c$, a complex neighbourhood is mapped after finitely many iterations into an attracting neighbourhood of $0$ or $1$.
The series converges locally uniformly there, proving analyticity, also at the two endpoints.

We prove the regularity assertion by induction on $j$.
Continuity follows from the uniformly convergent series.
Suppose $j\ge1$, $\varphi_R'(p_c)^j<m$, and $\kappa$ is already $C^{j-1}$ near $p_c$.
Differentiating~\eqref{eq:cluster-number-functional} $j$ times away from $p_c$ gives
\[
 \kappa^{(j)}(p)
 =\frac{\varphi_R'(p)^j}{m}\kappa^{(j)}(\varphi_R(p))
   +\mathcal R_j(p),
\]
where $\mathcal R_j$ collects the derivatives of the forcing polynomial and the chain-rule terms involving derivatives of $\kappa$ of order at most $j-1$.
The induction hypothesis makes $\mathcal R_j$ continuous at $p_c$.
The only possible limiting value of the $j$th derivative is therefore
\[
 c_j:=\frac{\mathcal R_j(p_c)}{1-\varphi_R'(p_c)^j/m}.
\]
Choose a sufficiently small neighbourhood of $p_c$ on which
$|\varphi_R'(p)^j/m|\le q<1$, and stop each noncritical forward orbit when it first leaves this neighbourhood.
The exit time tends to infinity as $p$ approaches $p_c$, and the exit points lie in a compact set away from $p_c$, where $\kappa^{(j)}$ is bounded by analyticity.
Subtracting $c_j$ from the recursion leaves a continuous error tending to zero at $p_c$.
Iteration makes the bounded exit contribution tend to zero geometrically and bounds the error contributions by a geometrically weighted sum.
For any fixed number of initial iterates these errors tend to zero as $p$ approaches $p_c$; the remaining terms are uniformly bounded by the tail of a geometric series.
Consequently $\displaystyle\lim_{p\to p_c,\,p\ne p_c}\kappa^{(j)}(p)=c_j$.
The mean value theorem, applied to the continuous function $\kappa^{(j-1)}$, now shows that its derivative at $p_c$ exists and equals $c_j$.
Thus $\kappa$ is $C^j$ near $p_c$, completing the induction.
Lemma~\ref{lem:pivotal-response-bound} gives $[\varphi_R'(p_c)]^2<m$, so $\kappa$ is always $C^2$ at $p_c$.

Under the last hypotheses of the proposition, Taylor's formula gives
\[
 \kappa'''(p)=\frac{\kappa^{(j)}(p_c)}{(j-3)!}(p-p_c)^{j-3}
 +o(|p-p_c|^{j-3}).
\]
Its leading coefficient is nonzero, so the required punctured neighbourhoods contain no zeros and the two logarithmic exponents equal $2-j$.
\end{proof}

The regularity condition is equivalently $j<\dimhaus(\Xi)/\dim_{\mathrm P}(\mathfrak C)$.
It does not determine the first nonzero derivative of order at least three, so it does not by itself identify the full-response exponent with $2-\dimhaus(\Xi)/\dim_{\mathrm P}(\mathfrak C)$.
The derivatives in Proposition~\ref{prop:cluster-number-response} are determined by finite differentiation of~\eqref{eq:cluster-number-functional}.
For example, whenever $\varphi_R'(p_c)^3<m$,
\begin{align*}
 (m-\varphi_R'(p_c))\kappa'(p_c)
 &=\frac{m-1}{|V(R)|-2}h_R'(p_c),\\
 (m-\varphi_R'(p_c)^2)\kappa''(p_c)
 &=\frac{m-1}{|V(R)|-2}h_R''(p_c)
   +\kappa'(p_c)\varphi_R''(p_c),\\
 (m-\varphi_R'(p_c)^3)\kappa'''(p_c)
 &=\frac{m-1}{|V(R)|-2}h_R'''(p_c)
   +\kappa'(p_c)\varphi_R'''(p_c)\\
 &\quad+3\kappa''(p_c)\varphi_R'(p_c)\varphi_R''(p_c).
\end{align*}

\begin{example}\label{ex:cluster-number-responses}
The ordinary diamond has $\alpha_\pm=-1$; Wheatstone has
$\alpha_\pm=2-\log5/\log(13/8)\approx-1.3150$; replacing the central edge of Wheatstone by another Wheatstone gives $\alpha_\pm=-2$.
\end{example}
\begin{proof}
For the diamond, $h_R(p)=2(1-p)^2$, $m=4$ and $\varphi_R(p)=2p^2-p^4$.
Here $p_c=(\sqrt5-1)/2$ and $\varphi_R'(p_c)=6-2\sqrt5$, whose cube is less than $4$.
The derivative identities give $\kappa'''(p_c)=(171+99p_c)/31>0$.
Proposition~\ref{prop:cluster-number-response} with $j=3$ proves the first assertion.

For Wheatstone, direct enumeration of its $32$ configurations gives
\begin{align*}
 \varphi_R(p)&=2p^2+2p^3-5p^4+2p^5,\\
 2h_R(p)&=4-10p+4p^2+8p^3-8p^4+2p^5.
\end{align*}
Here $p_c=1/2$, $m=5$ and $\varphi_R'(p_c)=13/8$.
The identities
\[
 \varphi_R(1-p)=1-\varphi_R(p),\qquad
 2h_R(p)-2h_R(1-p)=4-10p+2\varphi_R(p)
\]
and uniqueness in~\eqref{eq:cluster-number-functional} imply
$\kappa(p)-\kappa(1-p)=1-2p$.
We also need convexity: on any finite graph the cluster count is the vertex count minus the graphic-matroid rank, so its second mixed edge differences are nonnegative by rank submodularity.
Its Bernoulli expectation is convex, and volume limits preserve convexity.
The rule has no bridge, hence $h_R(p)=O((1-p)^2)$, and~\eqref{eq:cluster-number-functional} gives $\kappa'(1)=0$.
Together with reflection this yields
$-1\le\kappa'(p)\le0$ and $\kappa''(p)\ge0$ for $1/2\le p\le1$.

For $p=1/2+q$, $0\le q\le1/2$, the explicit polynomials give
\[
 \varphi_R''(p)=-2q(9-20q^2)\le0,\qquad
 (2h_R)'''(p)\le-18,\qquad
 (2h_R-\varphi_R)'''(p)=-72q.
\]
Differentiating~\eqref{eq:cluster-number-functional} three times gives
\begin{align}\label{eq:wheatstone-third-response}
 5\kappa'''(p)
 &=\varphi_R'(p)^3\kappa'''(\varphi_R(p))
   +3\kappa''(\varphi_R(p))\varphi_R'(p)\varphi_R''(p)\notag\\
 &\quad+\kappa'(\varphi_R(p))\varphi_R'''(p)+2h_R'''(p).
\end{align}
The last three terms have a strictly negative sum for $p>1/2$.
If $\varphi_R'''(p)\ge0$, use $\kappa'\le0$ and $(2h_R)'''\le-18$; if $\varphi_R'''(p)<0$, use $\kappa'\ge-1$ and $(2h_R-\varphi_R)'''=-72q$.
Their sum is bounded above by $-36q$.
It has absolute value at most $Cq$ near $q=0$: it vanishes at $p_c$ by the reflection identity, $\kappa$ is $C^2$, and $\varphi_R''(p_c)=0$.
Iterating~\eqref{eq:wheatstone-third-response} towards $1$ shows $\kappa'''(p)<0$ for $p>1/2$; the remaining derivative term tends to zero since $\varphi_R'(1)=0$ and $\kappa$ is analytic there.

Stop instead at the exit time $N(1/2+q)$ of a fixed small neighbourhood of $1/2$.
The orbit deviations are comparable to $(13/8)^j q$ up to that time, and the derivative products are comparable to $(13/8)^j$ by the summable-distortion argument used in~\eqref{eq:critical-escape-time}.
On the compact exit interval $-\kappa'''$ has a positive minimum.
Iteration of~\eqref{eq:wheatstone-third-response} therefore gives
\[
 -\kappa'''(1/2+q)\asymp
 \left(\frac{(13/8)^3}{5}\right)^{N(1/2+q)}
 +q\sum_{j<N(1/2+q)}\left(\frac{(13/8)^4}{5}\right)^j
 \asymp q^{\log5/\log(13/8)-3}.
\]
Here $(13/8)^3<5<(13/8)^4$.
Reflection gives the same absolute-value estimate below $p_c$, proving the second assertion.

For the last rule, take terminals $0,1$ and edges
$02,21,03,31,24,43,25,53,45$.
Enumeration of its $512$ configurations gives
\begin{align*}
 \varphi_R(p)&=2p^2+3p^4-4p^5-14p^6+28p^7-18p^8+4p^9,\\
 2h_R(p)&=8-18p+4p^2+4p^3+18p^4-2p^5\\
 &\quad-52p^6+60p^7-26p^8+4p^9.
\end{align*}
The rule has $m=9$, $p_c=1/2$ and $\varphi_R'(p_c)=109/64$.
Since $(109/64)^4<9$, fourth differentiation of~\eqref{eq:cluster-number-functional} is legitimate.
It gives
\[
 \kappa'''(1/2)=0,\qquad
 \kappa''''(1/2)=-\frac{72900157636608}{35107478527}<0.
\]
Proposition~\ref{prop:cluster-number-response} with $j=4$ completes the proof.
\end{proof}

These criteria and examples establish the stated exponents without giving necessary and sufficient conditions for the existence of the full-response exponent for every classical rule.

\subsection*{Cluster radius}
The cumulative radius distribution has a uniform law, but its point probabilities need not have a logarithmic exponent.

\begin{proposition}\label{prop:annealed-radius-tail}
For every classical rule,
\[
 \mathbb P_{p_c}^{\mathrm{ann}}(\operatorname{rad}C(o)\ge r)
 \asymp r^{-(\dimhaus(\Xi)-\dimmass(\mathfrak C))}.
\]
If the finite point-probability exponent $\rho$ in Definition~\ref{def:other-exponents} exists, then
$\rho=1/(\dimhaus(\Xi)-\dimmass(\mathfrak C))$.
For the ordinary diamond this point-probability exponent does not exist.
\end{proposition}
\begin{proof}
Every depth-$n$ cell has diameter at most $C_R\ell^n$ and is isometrically embedded, as established in the proof of Theorem~\ref{thm:delta-eta-dimensions}.
A finite cluster first counted in such a cell has no larger ambient radius.
The birth decomposition counts its vertices with weight proportional to $m^{-n}$, and Lemma~\ref{lem:conditional-mass-moments} gives
$\mathbb E_{p_c}\sum_{C\in\mathcal J_n}|C|\le C\rho(\mathbf M_{\mathfrak C})^n$.
Since $\theta(p_c)=0$, all root mass is accounted for.
Thus
\[
 \mathbb P_{p_c}^{\mathrm{ann}}(\operatorname{rad}C(o)\ge r)
 \le C\sum_{C_R\ell^n\ge r}
 \left(\frac{\rho(\mathbf M_{\mathfrak C})}{m}\right)^n
 \le Cr^{-(\dimhaus(\Xi)-\dimmass(\mathfrak C))}.
\]
For a lower bound, choose a fixed coarse depth $k$ containing an edge whose two endpoints are internal to the coarse cell.
Require this edge's depth-$n$ cell to cross and every other edge's depth-$n$ cell to fail to cross.
This event has a fixed positive probability, independent of $n$.
The resulting cluster avoids the enclosing cell's terminals, contains two vertices at distance $\ell^n$, and has expected mass at least $c\rho(\mathbf M_{\mathfrak C})^n$.
Every root in that cluster has radius at least $\ell^n/2$ by the triangle inequality.
Tile large substitution networks by depth-$(n+k)$ cells and count their internal roots.
After volume normalisation this yields a lower bound $c(\rho(\mathbf M_{\mathfrak C})/m)^n$ at radius $\ell^n/2$.
Choosing successive powers of $\ell$ proves the cumulative claim for all large $r$.
If the point-probability logarithmic limit exists, compare its probabilities with $r^{-1-1/\rho\pm\varepsilon}$ and sum them from $r$ to infinity.
The cumulative upper bound first forces $1+1/\rho>1$; these comparisons then identify $1/\rho=\dimhaus(\Xi)-\dimmass(\mathfrak C)$.

For the diamond, a depth-$n$ cell has terminal distance and diameter $L:=2^n$.
Every vertex has a layer coordinate $x$ with terminal distances $x,L-x$.
For vertices in different top-level parallel branches, their distance is
$\min\{x+y,2L-x-y\}$.
These facts follow by induction from the two parallel length-two paths; in particular, every crossing contains an open terminal geodesic.

Inside a depth-$(n+2)$ cell, choose a depth-$n$ subcell whose two terminals are internal, and call it the main cell.
Require all its $16$ depth-$(n-2)$ grandchildren to cross.
Require every other depth-$n$ cell of the enclosing two-level coarse graph to fail to cross.
For each such cell incident to a main-cell terminal, require all its $16$ depth-$(n-2)$ grandchildren to fail to cross as well.
These requirements are compatible and have probability bounded below independently of $n$, because $p_c$ is the crossing fixed point.
The main cluster is isolated from the enclosing terminals, and every attachment outside the main cell stays within distance $L/4$ of its attachment terminal.

Choose a main-cell grandchild lying in one parallel branch with layer interval $[L/4,L/2]$.
Its connected cluster has expected internal mass at least $c\rho(\mathbf M_{\mathfrak C})^{n-2}$.
For each of its vertices with coordinate $x$, the other parallel branch contains an open geodesic and hence a cluster vertex of coordinate $L-x$, at distance exactly $L$.
All main-cell vertices are at distance at most $L$, while each outside attachment is at distance at most
$\max\{x,L-x\}+L/4\le L$.
Every one of these roots therefore has radius exactly $L$.
Tiling by the enclosing cells and passing to the volume limit gives
\begin{equation}\label{eq:diamond-radius-spikes}
 \mathbb P_{p_c}^{\mathrm{ann}}(\operatorname{rad}C(o)=2^n)
 \ge c\left(\frac{\rho(\mathbf M_{\mathfrak C})}{4}\right)^n.
\end{equation}
On the other hand, the cumulative upper bound implies that among the $L$ integers in $[L,2L)$ there is some $r_L$ with
\[
 \mathbb P_{p_c}^{\mathrm{ann}}(\operatorname{rad}C(o)=r_L)
 \le CL^{-1-(\dimhaus(\Xi)-\dimmass(\mathfrak C))}.
\]
The logarithmic slopes along $L=2^n$ and along these $r_L$ differ by at least one in the limit; a zero point probability only strengthens the failure.
Thus the point-probability exponent does not exist, despite the cumulative law.
\end{proof}

\end{document}

%% file: figure/two_edge_path_symbol.tex
\newcommand{\twoedgepathdrawing}{%
  \tikz[x=1em,y=1em]{%
    \draw[red!45!black,line width=.55pt,line cap=round] (0,0)--(1.8,0);
    \foreach \x in {0,.9,1.8}
      \fill[red!45!black] (\x,0) circle[radius=.10];
  }%
}
\newcommand{\twoedgepath}{%
  \mathord{\mathchoice
    {\vcenter{\hbox{\twoedgepathdrawing}}}
    {\vcenter{\hbox{\twoedgepathdrawing}}}
    {\vcenter{\hbox{\scalebox{.72}{\twoedgepathdrawing}}}}
    {\vcenter{\hbox{\scalebox{.55}{\twoedgepathdrawing}}}}}%
}

%% file: figure/random_DHL_rules.tex
\begingroup
\definecolor{eigsred}{RGB}{205,46,53}
\definecolor{eigsblue}{RGB}{38,100,190}
\definecolor{eigsgreen}{RGB}{29,139,76}
\begin{tikzpicture}[x=1cm,y=1cm,line cap=round,line join=round,
 every node/.style={font=\small},
 bond/.style={line width=.85pt,-{Stealth[length=2.5pt,width=2.3pt]}},
 map/.style={black!45,line width=.65pt,-{Stealth[length=4pt]}}]
\path[use as bounding box] (0,-8.7) rectangle (16,.75);
\node[anchor=west,font=\small\bfseries] at (0,.62) {Rules};
\draw[bond,eigsred] (0.1500,-0.3500) -- (1.1000,-0.3500);
\fill (0.1500,-0.3500) circle (1.15pt);
\fill (1.1000,-0.3500) circle (1.15pt);
\draw[map] (1.40,-0.35) -- (2.05,-0.35);
\node at (4.25,0.18000000000000005) {$\frac{p^4}{2p^2-p^4}$};
\draw[bond,eigsred] (3.8700,-0.3500) -- (4.2500,-0.1100);
\draw[bond,eigsred] (4.2500,-0.1100) -- (4.6300,-0.3500);
\draw[bond,eigsred] (3.8700,-0.3500) -- (4.2500,-0.5900);
\draw[bond,eigsred] (4.2500,-0.5900) -- (4.6300,-0.3500);
\fill (3.8700,-0.3500) circle (1.0pt);
\fill (4.2500,-0.1100) circle (1.0pt);
\fill (4.2500,-0.5900) circle (1.0pt);
\fill (4.6300,-0.3500) circle (1.0pt);
\node at (8.85,0.18000000000000005) {$\frac{p^3(1-p)}{2p^2-p^4}$};
\draw[bond,eigsred] (6.9700,-0.3500) -- (7.3500,-0.1100);
\draw[bond,eigsred] (7.3500,-0.1100) -- (7.7300,-0.3500);
\draw[bond,eigsred] (6.9700,-0.3500) -- (7.3500,-0.5900);
\draw[bond,eigsblue] (7.3500,-0.5900) -- (7.7300,-0.3500);
\fill (6.9700,-0.3500) circle (1.0pt);
\fill (7.3500,-0.1100) circle (1.0pt);
\fill (7.3500,-0.5900) circle (1.0pt);
\fill (7.7300,-0.3500) circle (1.0pt);
\draw[bond,eigsred] (7.9700,-0.3500) -- (8.3500,-0.1100);
\draw[bond,eigsred] (8.3500,-0.1100) -- (8.7300,-0.3500);
\draw[bond,eigsblue] (7.9700,-0.3500) -- (8.3500,-0.5900);
\draw[bond,eigsred] (8.3500,-0.5900) -- (8.7300,-0.3500);
\fill (7.9700,-0.3500) circle (1.0pt);
\fill (8.3500,-0.1100) circle (1.0pt);
\fill (8.3500,-0.5900) circle (1.0pt);
\fill (8.7300,-0.3500) circle (1.0pt);
\draw[bond,eigsred] (8.9700,-0.3500) -- (9.3500,-0.1100);
\draw[bond,eigsblue] (9.3500,-0.1100) -- (9.7300,-0.3500);
\draw[bond,eigsred] (8.9700,-0.3500) -- (9.3500,-0.5900);
\draw[bond,eigsred] (9.3500,-0.5900) -- (9.7300,-0.3500);
\fill (8.9700,-0.3500) circle (1.0pt);
\fill (9.3500,-0.1100) circle (1.0pt);
\fill (9.3500,-0.5900) circle (1.0pt);
\fill (9.7300,-0.3500) circle (1.0pt);
\draw[bond,eigsblue] (9.9700,-0.3500) -- (10.3500,-0.1100);
\draw[bond,eigsred] (10.3500,-0.1100) -- (10.7300,-0.3500);
\draw[bond,eigsred] (9.9700,-0.3500) -- (10.3500,-0.5900);
\draw[bond,eigsred] (10.3500,-0.5900) -- (10.7300,-0.3500);
\fill (9.9700,-0.3500) circle (1.0pt);
\fill (10.3500,-0.1100) circle (1.0pt);
\fill (10.3500,-0.5900) circle (1.0pt);
\fill (10.7300,-0.3500) circle (1.0pt);
\node at (13.45,0.18000000000000005) {$\frac{p^2(1-p)^2}{2p^2-p^4}$};
\draw[bond,eigsred] (12.5700,-0.3500) -- (12.9500,-0.1100);
\draw[bond,eigsred] (12.9500,-0.1100) -- (13.3300,-0.3500);
\draw[bond,eigsgreen] (12.5700,-0.3500) -- (12.9500,-0.5900);
\draw[bond,eigsgreen] (13.3300,-0.3500) -- (12.9500,-0.5900);
\fill (12.5700,-0.3500) circle (1.0pt);
\fill (12.9500,-0.1100) circle (1.0pt);
\fill (12.9500,-0.5900) circle (1.0pt);
\fill (13.3300,-0.3500) circle (1.0pt);
\draw[bond,eigsgreen] (13.5700,-0.3500) -- (13.9500,-0.1100);
\draw[bond,eigsgreen] (14.3300,-0.3500) -- (13.9500,-0.1100);
\draw[bond,eigsred] (13.5700,-0.3500) -- (13.9500,-0.5900);
\draw[bond,eigsred] (13.9500,-0.5900) -- (14.3300,-0.3500);
\fill (13.5700,-0.3500) circle (1.0pt);
\fill (13.9500,-0.1100) circle (1.0pt);
\fill (13.9500,-0.5900) circle (1.0pt);
\fill (14.3300,-0.3500) circle (1.0pt);
\draw[black!12,line width=.35pt] (0,-0.9299999999999999) -- (16,-0.9299999999999999);
\draw[bond,eigsblue] (0.1500,-1.7000) -- (1.1000,-1.7000);
\fill (0.1500,-1.7000) circle (1.15pt);
\fill (1.1000,-1.7000) circle (1.15pt);
\draw[map] (1.40,-1.7000000000000002) -- (2.05,-1.7000000000000002);
\node at (4.25,-1.1700000000000002) {$\frac{p(1-p)^3}{(1-p^2)^2}$};
\draw[bond,eigsred] (2.3700,-1.7000) -- (2.7500,-1.4600);
\draw[bond,eigsblue] (2.7500,-1.4600) -- (3.1300,-1.7000);
\draw[bond,eigsgreen] (2.3700,-1.7000) -- (2.7500,-1.9400);
\draw[bond,eigsgreen] (3.1300,-1.7000) -- (2.7500,-1.9400);
\fill (2.3700,-1.7000) circle (1.0pt);
\fill (2.7500,-1.4600) circle (1.0pt);
\fill (2.7500,-1.9400) circle (1.0pt);
\fill (3.1300,-1.7000) circle (1.0pt);
\draw[bond,eigsblue] (3.3700,-1.7000) -- (3.7500,-1.4600);
\draw[bond,eigsred] (3.7500,-1.4600) -- (4.1300,-1.7000);
\draw[bond,eigsgreen] (3.3700,-1.7000) -- (3.7500,-1.9400);
\draw[bond,eigsgreen] (4.1300,-1.7000) -- (3.7500,-1.9400);
\fill (3.3700,-1.7000) circle (1.0pt);
\fill (3.7500,-1.4600) circle (1.0pt);
\fill (3.7500,-1.9400) circle (1.0pt);
\fill (4.1300,-1.7000) circle (1.0pt);
\draw[bond,eigsgreen] (4.3700,-1.7000) -- (4.7500,-1.4600);
\draw[bond,eigsgreen] (5.1300,-1.7000) -- (4.7500,-1.4600);
\draw[bond,eigsred] (4.3700,-1.7000) -- (4.7500,-1.9400);
\draw[bond,eigsblue] (4.7500,-1.9400) -- (5.1300,-1.7000);
\fill (4.3700,-1.7000) circle (1.0pt);
\fill (4.7500,-1.4600) circle (1.0pt);
\fill (4.7500,-1.9400) circle (1.0pt);
\fill (5.1300,-1.7000) circle (1.0pt);
\draw[bond,eigsgreen] (5.3700,-1.7000) -- (5.7500,-1.4600);
\draw[bond,eigsgreen] (6.1300,-1.7000) -- (5.7500,-1.4600);
\draw[bond,eigsblue] (5.3700,-1.7000) -- (5.7500,-1.9400);
\draw[bond,eigsred] (5.7500,-1.9400) -- (6.1300,-1.7000);
\fill (5.3700,-1.7000) circle (1.0pt);
\fill (5.7500,-1.4600) circle (1.0pt);
\fill (5.7500,-1.9400) circle (1.0pt);
\fill (6.1300,-1.7000) circle (1.0pt);
\node at (8.85,-1.1700000000000002) {$\frac{p^2(1-p)^2}{(1-p^2)^2}$};
\draw[bond,eigsred] (6.9700,-1.7000) -- (7.3500,-1.4600);
\draw[bond,eigsblue] (7.3500,-1.4600) -- (7.7300,-1.7000);
\draw[bond,eigsred] (6.9700,-1.7000) -- (7.3500,-1.9400);
\draw[bond,eigsblue] (7.3500,-1.9400) -- (7.7300,-1.7000);
\fill (6.9700,-1.7000) circle (1.0pt);
\fill (7.3500,-1.4600) circle (1.0pt);
\fill (7.3500,-1.9400) circle (1.0pt);
\fill (7.7300,-1.7000) circle (1.0pt);
\draw[bond,eigsblue] (7.9700,-1.7000) -- (8.3500,-1.4600);
\draw[bond,eigsred] (8.3500,-1.4600) -- (8.7300,-1.7000);
\draw[bond,eigsred] (7.9700,-1.7000) -- (8.3500,-1.9400);
\draw[bond,eigsblue] (8.3500,-1.9400) -- (8.7300,-1.7000);
\fill (7.9700,-1.7000) circle (1.0pt);
\fill (8.3500,-1.4600) circle (1.0pt);
\fill (8.3500,-1.9400) circle (1.0pt);
\fill (8.7300,-1.7000) circle (1.0pt);
\draw[bond,eigsred] (8.9700,-1.7000) -- (9.3500,-1.4600);
\draw[bond,eigsblue] (9.3500,-1.4600) -- (9.7300,-1.7000);
\draw[bond,eigsblue] (8.9700,-1.7000) -- (9.3500,-1.9400);
\draw[bond,eigsred] (9.3500,-1.9400) -- (9.7300,-1.7000);
\fill (8.9700,-1.7000) circle (1.0pt);
\fill (9.3500,-1.4600) circle (1.0pt);
\fill (9.3500,-1.9400) circle (1.0pt);
\fill (9.7300,-1.7000) circle (1.0pt);
\draw[bond,eigsblue] (9.9700,-1.7000) -- (10.3500,-1.4600);
\draw[bond,eigsred] (10.3500,-1.4600) -- (10.7300,-1.7000);
\draw[bond,eigsblue] (9.9700,-1.7000) -- (10.3500,-1.9400);
\draw[bond,eigsred] (10.3500,-1.9400) -- (10.7300,-1.7000);
\fill (9.9700,-1.7000) circle (1.0pt);
\fill (10.3500,-1.4600) circle (1.0pt);
\fill (10.3500,-1.9400) circle (1.0pt);
\fill (10.7300,-1.7000) circle (1.0pt);
\node at (13.45,-1.1700000000000002) {$\frac{(1-p)^4}{(1-p^2)^2}$};
\draw[bond,eigsgreen] (13.0700,-1.7000) -- (13.4500,-1.4600);
\draw[bond,eigsgreen] (13.8300,-1.7000) -- (13.4500,-1.4600);
\draw[bond,eigsgreen] (13.0700,-1.7000) -- (13.4500,-1.9400);
\draw[bond,eigsgreen] (13.8300,-1.7000) -- (13.4500,-1.9400);
\fill (13.0700,-1.7000) circle (1.0pt);
\fill (13.4500,-1.4600) circle (1.0pt);
\fill (13.4500,-1.9400) circle (1.0pt);
\fill (13.8300,-1.7000) circle (1.0pt);
\draw[black!12,line width=.35pt] (0,-2.2800000000000002) -- (16,-2.2800000000000002);
\draw[bond,eigsgreen] (0.1500,-3.0500) -- (1.1000,-3.0500);
\fill (0.1500,-3.0500) circle (1.15pt);
\fill (1.1000,-3.0500) circle (1.15pt);
\draw[map] (1.40,-3.0500000000000003) -- (2.05,-3.0500000000000003);
\node at (4.25,-2.5200000000000005) {$\frac{p(1-p)^2}{(1-p^2)^2}$};
\draw[bond,eigsred] (3.3700,-3.0500) -- (3.7500,-2.8100);
\draw[bond,eigsgreen] (3.7500,-2.8100) -- (4.1300,-3.0500);
\draw[bond,eigsgreen] (3.3700,-3.0500) -- (3.7500,-3.2900);
\draw[bond,black!65] (3.7500,-3.2900) -- (4.1300,-3.0500);
\fill (3.3700,-3.0500) circle (1.0pt);
\fill (3.7500,-2.8100) circle (1.0pt);
\fill (3.7500,-3.2900) circle (1.0pt);
\fill (4.1300,-3.0500) circle (1.0pt);
\draw[bond,eigsgreen] (4.3700,-3.0500) -- (4.7500,-2.8100);
\draw[bond,black!65] (4.7500,-2.8100) -- (5.1300,-3.0500);
\draw[bond,eigsred] (4.3700,-3.0500) -- (4.7500,-3.2900);
\draw[bond,eigsgreen] (4.7500,-3.2900) -- (5.1300,-3.0500);
\fill (4.3700,-3.0500) circle (1.0pt);
\fill (4.7500,-2.8100) circle (1.0pt);
\fill (4.7500,-3.2900) circle (1.0pt);
\fill (5.1300,-3.0500) circle (1.0pt);
\node at (8.85,-2.5200000000000005) {$\frac{p^2(1-p)^2}{(1-p^2)^2}$};
\draw[bond,eigsred] (8.4700,-3.0500) -- (8.8500,-2.8100);
\draw[bond,eigsgreen] (8.8500,-2.8100) -- (9.2300,-3.0500);
\draw[bond,eigsred] (8.4700,-3.0500) -- (8.8500,-3.2900);
\draw[bond,eigsgreen] (8.8500,-3.2900) -- (9.2300,-3.0500);
\fill (8.4700,-3.0500) circle (1.0pt);
\fill (8.8500,-2.8100) circle (1.0pt);
\fill (8.8500,-3.2900) circle (1.0pt);
\fill (9.2300,-3.0500) circle (1.0pt);
\node at (13.45,-2.5200000000000005) {$\frac{(1-p)^2}{(1-p^2)^2}$};
\draw[bond,eigsgreen] (13.0700,-3.0500) -- (13.4500,-2.8100);
\draw[bond,black!65] (13.4500,-2.8100) -- (13.8300,-3.0500);
\draw[bond,eigsgreen] (13.0700,-3.0500) -- (13.4500,-3.2900);
\draw[bond,black!65] (13.4500,-3.2900) -- (13.8300,-3.0500);
\fill (13.0700,-3.0500) circle (1.0pt);
\fill (13.4500,-2.8100) circle (1.0pt);
\fill (13.4500,-3.2900) circle (1.0pt);
\fill (13.8300,-3.0500) circle (1.0pt);
\draw[black!12,line width=.35pt] (0,-3.6300000000000003) -- (16,-3.6300000000000003);
\draw[bond,black!65] (0.1500,-4.4000) -- (1.1000,-4.4000);
\fill (0.1500,-4.4000) circle (1.15pt);
\fill (1.1000,-4.4000) circle (1.15pt);
\draw[map] (1.40,-4.4) -- (2.05,-4.4);
\node at (4.25,-3.87) {$1$};
\draw[bond,black!65] (3.8700,-4.4000) -- (4.2500,-4.1600);
\draw[bond,black!65] (4.2500,-4.1600) -- (4.6300,-4.4000);
\draw[bond,black!65] (3.8700,-4.4000) -- (4.2500,-4.6400);
\draw[bond,black!65] (4.2500,-4.6400) -- (4.6300,-4.4000);
\fill (3.8700,-4.4000) circle (1.0pt);
\fill (4.2500,-4.1600) circle (1.0pt);
\fill (4.2500,-4.6400) circle (1.0pt);
\fill (4.6300,-4.4000) circle (1.0pt);
\draw[densely dashed,black!45,line width=.6pt] (0,-5.0) -- (16,-5.0);
\node[anchor=west,font=\small\bfseries] at (0,-5.38) {One realisation};
\draw[bond,eigsred,line width=1pt] (0.2000,-6.6500) -- (1.0000,-6.6500);
\fill (0.2000,-6.6500) circle (1.15pt);
\fill (1.0000,-6.6500) circle (1.15pt);
\node at (.60,-8.45) {$\Xi^0$};
\draw[bond,eigsred] (1.8500,-6.6500) -- (2.6000,-6.3350);
\draw[bond,eigsred] (2.6000,-6.3350) -- (3.3500,-6.6500);
\draw[bond,eigsred] (1.8500,-6.6500) -- (2.6000,-6.9650);
\draw[bond,eigsblue] (2.6000,-6.9650) -- (3.3500,-6.6500);
\fill (1.8500,-6.6500) circle (1.0pt);
\fill (2.6000,-6.3350) circle (1.0pt);
\fill (2.6000,-6.9650) circle (1.0pt);
\fill (3.3500,-6.6500) circle (1.0pt);
\node at (2.6,-8.45) {$\Xi^1$};
\draw[bond,eigsred] (4.5000,-6.6500) -- (5.2328,-6.1540);
\draw[bond,eigsred] (5.2328,-6.1540) -- (6.1000,-5.9780);
\draw[bond,eigsgreen] (4.5000,-6.6500) -- (5.3672,-6.4740);
\draw[bond,eigsgreen] (6.1000,-5.9780) -- (5.3672,-6.4740);
\draw[bond,eigsred] (6.1000,-5.9780) -- (6.9672,-6.1540);
\draw[bond,eigsred] (6.9672,-6.1540) -- (7.7000,-6.6500);
\draw[bond,eigsred] (6.1000,-5.9780) -- (6.8328,-6.4740);
\draw[bond,eigsred] (6.8328,-6.4740) -- (7.7000,-6.6500);
\draw[bond,eigsgreen] (4.5000,-6.6500) -- (5.3672,-6.8260);
\draw[bond,eigsgreen] (6.1000,-7.3220) -- (5.3672,-6.8260);
\draw[bond,eigsred] (4.5000,-6.6500) -- (5.2328,-7.1460);
\draw[bond,eigsred] (5.2328,-7.1460) -- (6.1000,-7.3220);
\draw[bond,eigsred] (6.1000,-7.3220) -- (6.8328,-6.8260);
\draw[bond,eigsblue] (6.8328,-6.8260) -- (7.7000,-6.6500);
\draw[bond,eigsgreen] (6.1000,-7.3220) -- (6.9672,-7.1460);
\draw[bond,eigsgreen] (7.7000,-6.6500) -- (6.9672,-7.1460);
\fill (4.5000,-6.6500) circle (1.0pt);
\fill (6.1000,-5.9780) circle (1.0pt);
\fill (6.1000,-7.3220) circle (1.0pt);
\fill (7.7000,-6.6500) circle (1.0pt);
\fill (5.2328,-6.1540) circle (1.0pt);
\fill (5.3672,-6.4740) circle (1.0pt);
\fill (6.9672,-6.1540) circle (1.0pt);
\fill (6.8328,-6.4740) circle (1.0pt);
\fill (5.3672,-6.8260) circle (1.0pt);
\fill (5.2328,-7.1460) circle (1.0pt);
\fill (6.8328,-6.8260) circle (1.0pt);
\fill (6.9672,-7.1460) circle (1.0pt);
\node at (6.1,-8.45) {$\Xi^2$};
\draw[bond,eigsred,line width=.75pt] (8.9000,-6.6500) -- (9.4761,-6.1198);
\draw[bond,eigsred,line width=.75pt] (9.4761,-6.1198) -- (10.1824,-5.7820);
\draw[bond,eigsred,line width=.75pt] (8.9000,-6.6500) -- (9.6063,-6.3122);
\draw[bond,eigsblue,line width=.75pt] (9.6063,-6.3122) -- (10.1824,-5.7820);
\draw[bond,eigsred,line width=.75pt] (10.1824,-5.7820) -- (10.9181,-5.5142);
\draw[bond,eigsred,line width=.75pt] (10.9181,-5.5142) -- (11.7000,-5.4740);
\draw[bond,eigsgreen,line width=.75pt] (10.1824,-5.7820) -- (10.9643,-5.7418);
\draw[bond,eigsgreen,line width=.75pt] (11.7000,-5.4740) -- (10.9643,-5.7418);
\draw[bond,eigsred,line width=.75pt] (8.9000,-6.6500) -- (9.6357,-6.3822);
\draw[bond,eigsgreen,line width=.75pt] (9.6357,-6.3822) -- (10.4176,-6.3420);
\draw[bond,eigsred,line width=.75pt] (8.9000,-6.6500) -- (9.6819,-6.6098);
\draw[bond,eigsgreen,line width=.75pt] (9.6819,-6.6098) -- (10.4176,-6.3420);
\draw[bond,eigsgreen,line width=.75pt] (11.7000,-5.4740) -- (11.1239,-6.0042);
\draw[bond,black!65,line width=.75pt] (11.1239,-6.0042) -- (10.4176,-6.3420);
\draw[bond,eigsgreen,line width=.75pt] (11.7000,-5.4740) -- (10.9937,-5.8118);
\draw[bond,black!65,line width=.75pt] (10.9937,-5.8118) -- (10.4176,-6.3420);
\draw[bond,eigsred,line width=.75pt] (11.7000,-5.4740) -- (12.4819,-5.5142);
\draw[bond,eigsred,line width=.75pt] (12.4819,-5.5142) -- (13.2176,-5.7820);
\draw[bond,eigsred,line width=.75pt] (11.7000,-5.4740) -- (12.4357,-5.7418);
\draw[bond,eigsblue,line width=.75pt] (12.4357,-5.7418) -- (13.2176,-5.7820);
\draw[bond,eigsred,line width=.75pt] (13.2176,-5.7820) -- (13.9239,-6.1198);
\draw[bond,eigsred,line width=.75pt] (13.9239,-6.1198) -- (14.5000,-6.6500);
\draw[bond,eigsgreen,line width=.75pt] (13.2176,-5.7820) -- (13.7937,-6.3122);
\draw[bond,eigsgreen,line width=.75pt] (14.5000,-6.6500) -- (13.7937,-6.3122);
\draw[bond,eigsred,line width=.75pt] (11.7000,-5.4740) -- (12.4063,-5.8118);
\draw[bond,eigsred,line width=.75pt] (12.4063,-5.8118) -- (12.9824,-6.3420);
\draw[bond,eigsred,line width=.75pt] (11.7000,-5.4740) -- (12.2761,-6.0042);
\draw[bond,eigsred,line width=.75pt] (12.2761,-6.0042) -- (12.9824,-6.3420);
\draw[bond,eigsgreen,line width=.75pt] (12.9824,-6.3420) -- (13.7643,-6.3822);
\draw[bond,eigsgreen,line width=.75pt] (14.5000,-6.6500) -- (13.7643,-6.3822);
\draw[bond,eigsred,line width=.75pt] (12.9824,-6.3420) -- (13.7181,-6.6098);
\draw[bond,eigsred,line width=.75pt] (13.7181,-6.6098) -- (14.5000,-6.6500);
\draw[bond,eigsred,line width=.75pt] (8.9000,-6.6500) -- (9.6819,-6.6902);
\draw[bond,eigsgreen,line width=.75pt] (9.6819,-6.6902) -- (10.4176,-6.9580);
\draw[bond,eigsgreen,line width=.75pt] (8.9000,-6.6500) -- (9.6357,-6.9178);
\draw[bond,black!65,line width=.75pt] (9.6357,-6.9178) -- (10.4176,-6.9580);
\draw[bond,eigsgreen,line width=.75pt] (11.7000,-7.8260) -- (10.9937,-7.4882);
\draw[bond,black!65,line width=.75pt] (10.9937,-7.4882) -- (10.4176,-6.9580);
\draw[bond,eigsred,line width=.75pt] (11.7000,-7.8260) -- (11.1239,-7.2958);
\draw[bond,eigsgreen,line width=.75pt] (11.1239,-7.2958) -- (10.4176,-6.9580);
\draw[bond,eigsred,line width=.75pt] (8.9000,-6.6500) -- (9.6063,-6.9878);
\draw[bond,eigsred,line width=.75pt] (9.6063,-6.9878) -- (10.1824,-7.5180);
\draw[bond,eigsred,line width=.75pt] (8.9000,-6.6500) -- (9.4761,-7.1802);
\draw[bond,eigsred,line width=.75pt] (9.4761,-7.1802) -- (10.1824,-7.5180);
\draw[bond,eigsgreen,line width=.75pt] (10.1824,-7.5180) -- (10.9643,-7.5582);
\draw[bond,eigsgreen,line width=.75pt] (11.7000,-7.8260) -- (10.9643,-7.5582);
\draw[bond,eigsred,line width=.75pt] (10.1824,-7.5180) -- (10.9181,-7.7858);
\draw[bond,eigsred,line width=.75pt] (10.9181,-7.7858) -- (11.7000,-7.8260);
\draw[bond,eigsred,line width=.75pt] (11.7000,-7.8260) -- (12.2761,-7.2958);
\draw[bond,eigsred,line width=.75pt] (12.2761,-7.2958) -- (12.9824,-6.9580);
\draw[bond,eigsred,line width=.75pt] (11.7000,-7.8260) -- (12.4063,-7.4882);
\draw[bond,eigsblue,line width=.75pt] (12.4063,-7.4882) -- (12.9824,-6.9580);
\draw[bond,eigsred,line width=.75pt] (12.9824,-6.9580) -- (13.7181,-6.6902);
\draw[bond,eigsblue,line width=.75pt] (13.7181,-6.6902) -- (14.5000,-6.6500);
\draw[bond,eigsred,line width=.75pt] (12.9824,-6.9580) -- (13.7643,-6.9178);
\draw[bond,eigsblue,line width=.75pt] (13.7643,-6.9178) -- (14.5000,-6.6500);
\draw[bond,eigsred,line width=.75pt] (11.7000,-7.8260) -- (12.4357,-7.5582);
\draw[bond,eigsgreen,line width=.75pt] (12.4357,-7.5582) -- (13.2176,-7.5180);
\draw[bond,eigsred,line width=.75pt] (11.7000,-7.8260) -- (12.4819,-7.7858);
\draw[bond,eigsgreen,line width=.75pt] (12.4819,-7.7858) -- (13.2176,-7.5180);
\draw[bond,eigsgreen,line width=.75pt] (14.5000,-6.6500) -- (13.9239,-7.1802);
\draw[bond,black!65,line width=.75pt] (13.9239,-7.1802) -- (13.2176,-7.5180);
\draw[bond,eigsgreen,line width=.75pt] (14.5000,-6.6500) -- (13.7937,-6.9878);
\draw[bond,black!65,line width=.75pt] (13.7937,-6.9878) -- (13.2176,-7.5180);
\fill (8.9000,-6.6500) circle (0.75pt);
\fill (11.7000,-5.4740) circle (0.75pt);
\fill (11.7000,-7.8260) circle (0.75pt);
\fill (14.5000,-6.6500) circle (0.75pt);
\fill (10.1824,-5.7820) circle (0.75pt);
\fill (10.4176,-6.3420) circle (0.75pt);
\fill (13.2176,-5.7820) circle (0.75pt);
\fill (12.9824,-6.3420) circle (0.75pt);
\fill (10.4176,-6.9580) circle (0.75pt);
\fill (10.1824,-7.5180) circle (0.75pt);
\fill (12.9824,-6.9580) circle (0.75pt);
\fill (13.2176,-7.5180) circle (0.75pt);
\fill (9.4761,-6.1198) circle (0.75pt);
\fill (9.6063,-6.3122) circle (0.75pt);
\fill (10.9181,-5.5142) circle (0.75pt);
\fill (10.9643,-5.7418) circle (0.75pt);
\fill (9.6357,-6.3822) circle (0.75pt);
\fill (9.6819,-6.6098) circle (0.75pt);
\fill (11.1239,-6.0042) circle (0.75pt);
\fill (10.9937,-5.8118) circle (0.75pt);
\fill (12.4819,-5.5142) circle (0.75pt);
\fill (12.4357,-5.7418) circle (0.75pt);
\fill (13.9239,-6.1198) circle (0.75pt);
\fill (13.7937,-6.3122) circle (0.75pt);
\fill (12.4063,-5.8118) circle (0.75pt);
\fill (12.2761,-6.0042) circle (0.75pt);
\fill (13.7643,-6.3822) circle (0.75pt);
\fill (13.7181,-6.6098) circle (0.75pt);
\fill (9.6819,-6.6902) circle (0.75pt);
\fill (9.6357,-6.9178) circle (0.75pt);
\fill (10.9937,-7.4882) circle (0.75pt);
\fill (11.1239,-7.2958) circle (0.75pt);
\fill (9.6063,-6.9878) circle (0.75pt);
\fill (9.4761,-7.1802) circle (0.75pt);
\fill (10.9643,-7.5582) circle (0.75pt);
\fill (10.9181,-7.7858) circle (0.75pt);
\fill (12.2761,-7.2958) circle (0.75pt);
\fill (12.4063,-7.4882) circle (0.75pt);
\fill (13.7181,-6.6902) circle (0.75pt);
\fill (13.7643,-6.9178) circle (0.75pt);
\fill (12.4357,-7.5582) circle (0.75pt);
\fill (12.4819,-7.7858) circle (0.75pt);
\fill (13.9239,-7.1802) circle (0.75pt);
\fill (13.7937,-6.9878) circle (0.75pt);
\node at (11.7,-8.45) {$\Xi^3$};
\draw[map] (1.2,-6.65) -- (1.65,-6.65);
\draw[map] (3.6,-6.65) -- (4.1,-6.65);
\draw[map] (8.0,-6.65) -- (8.55,-6.65);
\draw[map] (14.8,-6.65) -- (15.35,-6.65);
\node at (15.75,-6.65) {$\cdots$};
\end{tikzpicture}
\endgroup

%% file: figure/wheatstone_permutation_rules.tex
\begin{tikzpicture}[
  x=1cm,y=1cm,line cap=round,line join=round,
  edge/.style={line width=.75pt},
  replaced/.style={draw=red!45!black,line width=.9pt},
  vertex/.style={circle,fill=black,inner sep=1.65pt},
  terminal/.style={circle,fill=black,inner sep=2.5pt},
  newvertex/.style={circle,fill=red!45!black,inner sep=1.65pt},
  every node/.style={font=\small}
]
\begin{scope}
  \coordinate (s) at (-1.8,0);
  \coordinate (u) at (0,1.4);
  \coordinate (t) at (1.8,0);
  \coordinate (v) at (0,-1.4);
  \draw[edge] (u)--(t) (s)--(v) (u)--(v);
  \draw[replaced] (s)--(u) (v)--(t);
  \node[terminal,label=left:$s$] at (s) {};
  \node[terminal,label=right:$t$] at (t) {};
  \node[vertex,label=above:$u$] at (u) {};
  \node[vertex,label=below:$v$] at (v) {};
  \node at (0,-2.05) {$A=W$};
\end{scope}
\draw[-{Stealth[length=2mm,width=1.6mm]},black!55,line width=.65pt]
  (2.45,0)--node[above=5pt,text=black]{$su,vt\mapsto W$}(4.35,0);
\begin{scope}[xshift=6.8cm]
  \coordinate (s) at (-1.8,0);
  \coordinate (u) at (0,1.4);
  \coordinate (t) at (1.8,0);
  \coordinate (v) at (0,-1.4);
  \coordinate (su1) at (-1.25,1.15);
  \coordinate (su2) at (-.55,.25);
  \coordinate (vt1) at (.55,-.25);
  \coordinate (vt2) at (1.25,-1.15);
  \draw[edge] (u)--(t) (s)--(v) (u)--(v);
  \draw[replaced] (s)--(su1)--(u)--(su2)--(s) (su1)--(su2);
  \draw[replaced] (v)--(vt1)--(t)--(vt2)--(v) (vt1)--(vt2);
  \foreach \p in {su1,su2,vt1,vt2}
    \node[newvertex] at (\p) {};
  \node[terminal,label=left:$s$] at (s) {};
  \node[terminal,label=right:$t$] at (t) {};
  \node[vertex,label=above:$u$] at (u) {};
  \node[vertex,label=below:$v$] at (v) {};
  \node at (0,-2.05) {$B$};
\end{scope}
\end{tikzpicture}